\documentclass[10pt]{article}
\usepackage[utf8]{inputenc}
\usepackage{subfiles}
\usepackage[english]{babel}
\usepackage[a4paper,top=3 cm,bottom=3 cm,inner=2.5 cm,outer=2.5 cm]{geometry}
\usepackage{graphicx}
\usepackage{subfigure, tikz}
\usetikzlibrary{patterns}
\usetikzlibrary{positioning}
\usetikzlibrary{decorations.markings}
\usetikzlibrary{decorations.pathmorphing, decorations.markings, arrows.meta}

\usepackage{bm}
\usepackage{mathtools}
\usepackage{amsmath, amssymb, amsthm} 
\usepackage{amsfonts}
\usepackage{mathbbol}
\usepackage{enumitem}
\usepackage{cancel}
\usepackage{indentfirst}
\usepackage{tikz-cd}

\newtheorem{theorem}{Theorem}[section] 
\newtheorem{lemma}[theorem]{Lemma}
\newtheorem{proposition}[theorem]{Proposition}
\newtheorem{definition}[theorem]{Definition}
\newtheorem{corollary}[theorem]{Corollary}
\newtheorem{rem}[theorem]{Remark}

\theoremstyle{definition}
\setlist[itemize]{label=\textbullet}
\usepackage{hyperref}
\usepackage{eurosym}

\usepackage{wasysym}
\usepackage{url}
\usepackage{titlesec}
\usepackage{mathrsfs}
\usepackage{leftidx}
\usepackage{braket}
\usepackage{physics}

\usepackage{csquotes}
\usepackage[
backend=biber,
giveninits=true,
style=alphabetic,
isbn = false,
url = false,
maxbibnames = 4,
]{biblatex}
\newcommand{\di}{\text{d}}

\newcommand{\dvol}[1]{\mathrm{d}#1}

\begin{document}
\par 
\bigskip 
\LARGE 
\noindent 
\textbf{Equilibrium states for non-relativistic Bose gases and the Gross-Pitaevskii limit} 
\bigskip \bigskip
\par 
\rm 
\normalsize 
 
\large
\noindent 
{\bf Stefano Galanda$^{1,a}$}, {\bf Nicola Pinamonti$^{2,3,b}$}\\
\par
\small

\noindent$^1$ Department of Mathematics, 
University of York, York YO10 5GH, United Kingdom. \smallskip

\noindent$^2$ Dipartimento di Matematica, 
Universit\`a di Genova - Via Dodecaneso, 35, I-16146 Genova, Italy. \smallskip

\noindent$^3$ Istituto Nazionale di Fisica Nucleare - Sezione di Genova, Via Dodecaneso, 33 I-16146 Genova, Italy. \smallskip
\smallskip

\noindent E-mail: 
$^a$stefano.galanda@york.ac.uk, 
$^b$nicola.pinamonti@unige.it
\\

\normalsize
${}$ \\ \\
 {\bf Abstract}\\
In this paper, we present the construction of equilibrium states for a gas of weakly interacting non-relativistic bosons, focusing on the case of a non-trivial background field in infinitely extended space. Building upon a method introduced by Araki and further developed by Fredenhagen and Lindner, we derive the generating function of the correlation functions of the theory as a suitable series. By applying a Hubbard-Stratonovich transformation, we rewrite this quantity into a more mathematically tractable form, allowing us to establish the convergence of the corresponding loop vertex expansion in certain intermediate regimes. Furthermore, we isolate the tree diagrams that produce the scattering length in the dispersion relations of the two-point function of the state within the Gross-Pitaevskii regime. Finally, we use this scattering length to  renormalise the background and the two-point function of the fluctuations and we discuss convergence of the generating function of the connected correlation functions of the renormalised theory in the limit of vanishing temperature.

\tableofcontents

\section{Introduction}

We study a gas of weakly interacting charged bosons in an infinitely extended space at finite temperature and in states where the internal $U(1)$ symmetry is spontaneously broken. 
This is equivalent to a gas of particles in an equilibrium state at finite temperature  where the  phenomenon of Bose-Einstein condensation (BEC) occurs
\cite{Bose, Einstein}. 
For the case of fixed $\mathcal{N}$ particle states, 
the Hamiltonian of the system,
is described by 
an operator acting on the Hilbert space
$\otimes_s^{\mathcal{N}} L^2(\mathbb{R}^3)$ and it takes the form 
\[
\mathsf{H}_\mathcal{N} = \sum_{i=1}^\mathcal{N}\left(-\frac{\Delta_i}{2m} - \mu\right) 
+ \sum_{1\leq i\leq \mathcal{N}} V(x_i)
+ \sum_{1\leq i<j\leq \mathcal{N}} v(x_i-x_j)
\]
where $\Delta$ is the ordinary Laplace operator on $\mathbb{R}^3$,
$m$ is the particle mass and $\mu$ the chemical potential of the system.
$V$ is a trapping potential and $v$ is the pair potential which describes the force among pairs of particles. 
We assume, for simplicity,
that $v$ is a smooth compactly supported function whose Fourier transform $\hat{v}$ is everywhere positive. 

We do not fix the number of particles a priori and we choose to work with the corresponding quantum field theory.
The corresponding Hamiltonian density given in terms of the quantum fields $\Phi$ and $\Phi^*$ and $|\Phi|^2=\Phi^*\Phi$ is thus
\begin{equation}\label{eq:Hamiltonian-density-field}
\mathcal{H} = \frac{1}{2} 
\Phi^* K \Phi
+
\frac{1}{2} K\Phi^* \Phi
-
\tilde{\mu}\Phi^*\Phi
+ 
V\Phi^*\Phi+
\frac{1}{2}
|\Phi|^2 (v*|\Phi|^2)
\end{equation}
where $v*|\Phi|^2$ denotes the ordinary convolution of $v$ with $|\Phi|^2$ and
the operator 
$K=-(2m)^{-1}\Delta$. Furthermore, 
$\tilde{\mu}$ is a suitable renormalisation of the chemical potential (we shall use the operator $K_{-\tilde{\mu}} \coloneqq {K}-\tilde{\mu}$ when it has positive spectrum).
Furthermore, we shall remove the trapping potential $V$ keeping the chemical potential finite and positive in order to allow for a non trivial bound state for the system to exist. We denote by $H$ the corresponding Hamiltonian.

For weakly interacting dilute trapped bosons the existence of BEC  has been proved by Lieb and Seiringer \cite{LiebSeiringer}
in the Gross-Pitaevskii regime.
The {\bf Gross-Pitaevskii (GP)} regime is obtained considering a suitable limit where the interaction potential  described by $v$
scales as
\begin{equation}\label{eq:vn}
v_N(x) = N^2 v(xN),
\end{equation}
for large $N$, where the condensate density $\rho$ scales also as $N$ in such a way that $\rho \hat{v}_N(0)$ is kept constant and where also the trapping potential $V$ does not depend on $N$. 

In that context, the quantum fields $\Phi$ and $\Phi^*$ acquire a non vanishing background expectation value which scales as
$\sqrt{N}$. Similarly, the Hamiltonian density $\mathcal{H}$ scales also as $N$ and, in the limit of $N$ to infinity, the rescaled Hamiltonian ${H}/N$ is dominated by celebrated Gross-Pitaevskii energy functional 
\begin{equation}
\label{eq:gross-pitaevskii0energy-density}
{H}_{GP}=\left|\bar{\varphi}_0 \frac{\Delta}{2m} \varphi_0\right|+V|\varphi_0|^2+ 4\pi \mathfrak{a}_0|\varphi_0|^4
\end{equation}
where $\varphi_0$ is a solution of the Gross-Pitaevskii equation and 
$\mathfrak{a}_0$ is the scattering length of the unconstrained pair potential $v$ (see the appendix \ref{se:scattering-length}).
We refer to the work of Lieb, Yngvason and Seiringer \cite{Lieb-Seiringer-Yngavson} for a rigorous derivation of the energy functional, see also e.g. \cite{LY-ground}.

Notice that in the limit of large $N$, 
$N v_N$ 
converges to a distribution proportional to the Dirac delta function and the constant of proportionality is $\hat{v}(0)$,
essentially because
$N \hat{v}_N(0) = \hat{v}(0)$.
Other regimes which we shall analyse and where $Nv_N$ and also $\rho v_N$ converge to a constant are realised taking
\begin{equation}\label{eq:vnb}
v_{N,b}(x) = N^{3b-1} v(xN^b)
\end{equation}
for $b\in (0,1]$ at the place of $v$ and analysing the limit of large $N$. The case $b\to0$ is sometime called {\bf mean field regime}.

\vspace{0.5cm}
In this paper we study the homogenous case $V=0$ on $\mathbb{R}^3$ but keeping a non vanishing chemical potential $\tilde\mu$, which is assumed to be positive in order to allow for a classical non trivial background value for the fields to describe a spontaneous breaking of the $U(1)$ internal symmetry.  
Actually, the Hamiltonian density for the corresponding quantum field theory is then given as in \eqref{eq:Hamiltonian-density-field}. 
Following closely a strategy used in \cite{BogoliubovBEC}, 
and later further developed in \cite{Beliaev1, Beliaev2},
we start considering the case where the fields $\Phi$ and $\Phi^*$ are decomposed by their classical background component $\phi_0$, assumed to be real valued, plus its quantum fluctuations $\Psi$ 
\[
\Phi = \phi_0 + \Psi, \qquad
{\Phi}^* = \phi_0 + \Psi^*.
\]
Furthermore, at the beginning, the background part of the field $\phi_0$ is considered to be a real homogeneous solution of the classical equation of motion, hence
\[
\phi_0^2 = \frac{\tilde{\mu}}{\hat{v}(0)}
\]
where the renormalised chemical potential $\tilde{\mu} =\mu + v(0)/2 -q_\beta
{\hat{v}(0)}
$ is considered. This renormalisation, through  $q_{\beta}$, takes into account the normal ordering with respect to the thermal reference state of the fluctuations over which the the interaction Hamiltonian perturb the evolution.  
In particular, $q_\beta$ vanishes in the limit where the inverse temperature $\beta$ diverges, see e.g. \eqref{eq:qbeta} below.

In this context the GP limit is realised considering $v_N$ at the place of $v$ and considering the limits where $N$ tends to infinity keeping the renormalised chemical potential $\tilde\mu$ finite. 
Notice that in this case $\phi_0^2$ scales as $N$, and  $Nv_N$ converges to a delta function multiplied by $\hat{v}(0)$ for large $N$.
Then, a naive approach would lead to conclude that the background value of the energy density is dominated by $\phi_0$ and that the energy density of $\phi_0/\sqrt{N}$ is described by an expression similar to \eqref{eq:gross-pitaevskii0energy-density} but with $\hat{v}(0)$ in front of the $|\phi_0|^4$ term. However, $\hat{v}(0)$ is only the first term in the Born series which describes the scattering length, see e.g. equation \ref{eq:born-series}.

For the case of constrained dilute particle system in a box of dimension $\Lambda$ and considering an $N$ particle system at vanishing temperature 
it is  known that considering the plain constant background value of the field $\phi_0$ overestimates the background energy of the system. Instead, see e.g. the work of 
 Lieb and Yngvason \cite{LY-ground}, of Lieb, Yngvason and Seiringer \cite{Lieb-Seiringer-Yngavson},
 of Nam, Rougerie and Seiringer
\cite{NRS} and of 
Boccato, Brennecke, Cenatiempo and Schlein 
\cite{BBCS18}, at leading order the background energy is
\[
E_N  = 4\pi (N-1) \mathfrak{a}_0 \phi_R^2+O(1)
\]
where $\mathfrak{a}_0$ is the scattering length of the potential $v$, equal to $\hat{v}(0)/(8\pi)$ only in the lowest order Born approximation, and $\phi_R^2$ is the renormalised background density (which is sometime set to $1$).
More recently, it has been proven by
Boccato, Brennecke, Cenatiempo and Schlein in \cite{BoccatoBrenneckeCenatiempoSchlein},
see also \cite{HainzlSchleinTriay},
that in the Gross-Pitaevskii regime also subleading contributions of the energy density are given in terms of the scattering length. 
More precisely, the obtained ground state energy as a function of $N$ (with renormalised condensate density equal to $\phi_R^2$) for the vanishing temperature case  obtained in \cite{BoccatoBrenneckeCenatiempoSchlein}  is
\[
E_N = 4\pi (N-1) \mathfrak{a}_0 \phi_R^2+e_\Lambda \mathfrak{a}_0^2\phi_R^2
-\frac{1}{2} \sum_{p\in \Lambda_+^*} \left(p^2 + 8\pi \mathfrak{a}_0 \phi_R^2 - \sqrt{p^4 + 16\pi \mathfrak{a}_0 \phi_R^2 p^2} - \frac{(8\pi \mathfrak{a}_0\phi_R^2)^2}{2p^2}\right) + O(N^{-1/4})
\]
where $e_\Lambda$ is a suitable constant depending on the dimension of the trapping region and the sum is over all non vanishing discrete momenta. Furthermore the dispersion relation of the fluctuations over the ground state at momentum $p$ are
\begin{equation}\label{eq:dispoersion-relations}
e(p)=\sqrt{p^4 + 16\pi \mathfrak{a}_0 \phi_R^2 p^2}.
\end{equation}
For the case of finite temperature, recently, some bounds satisfied by the free energy have been obtained in \cite{BoccatoBastiCenatiempoDeuchert}. The case of infinitely extended systems in the mean field regime has been studied by Derezi\'nski and Napi\'orkowski in \cite{DerNap2014}. Recently, Derezi\'nski and Pettinari have carefully analysed the form of the dispersion relations at finite temperature in  \cite{DerPet26}.

We notice that, in the limits of vanishing temperature and for states with a large number of particles (large $N$), the ground state energy contains both the renormalised background density $\phi_R^2$ and the scattering length $\mathfrak{a}_0$ instead of $\phi_0^2$ and $\hat{v}$.
In the limit of a large number of particles, the background field, the zero mode is the most relevant contribution and linearised perturbations dominate over other contributions. 

In this paper we analyse role of these corrections in the construction of an interacting quantum field theory at finite temperature in an infinite spacetime, from the point of view of quantum field theory.
We shall see that it is thus necessary to renormalise both the background value of the field and the propagators of the linearised theory to correctly estimate the form of the quantum state also at finite temperature.
We argue that only after taking into account these renormalisations, the fluctuations and the non linear contributions can be controlled to be small.
Furthermore, these renormalisations, are necessary to control the adiabatic limit for non vanishing temperature. 
We shall see that the two-point function of the theory  can be computed for all values of the temperature and in the limit of vanishing teperature it behaves in the way predicted by Bogoliubov in \cite{BogoliubovBEC}. 
In passing, we shall identify the diagrams relevant for the  renomalization of $\hat{v}(0)$ to $\mathfrak{a}_0$ in the dispersion relations
of the fluctuations as 
those arising from a Dyson equation with an effective interaction kernel (the self-energy in the quantum field theory language) obtained as the sum of the so called ladder diagram formed with $v$ and the free propagators. 
Following this strategy, we prove that the background value of the field gets renormalised to a solution of the Gross-Pitaevskii equation. Moreover, in the limit of vanishing temperature, we prove that the dispersion relations read from the two-point function of the theory contains the scattering length and agrees with \eqref{eq:dispoersion-relations}. 
Finally, an estimate of the critical temperature above which no condensate can occur is derived.

\subsection{Strategy and results}

We use various techniques coming from various research areas to estimate 
the generating function of the correlation functions of the considered interacting theory at finite temperature as a suitable power series and to prove its convergence in certain regime. 
In particular, we start from methods and techniques of AQFT \cite{HaagLQP, HaagKastler} and of pAQFT \cite{HollandsWald2001, HW02, BDF09, BrunettiFredenhagen00, Fredenhagen2015, KasiaBook, Dutsch:2019aa} to obtain, a priori as a formal power series, the generating function of the thermal correlation functions for the interacting theory.

Then, we use techniques of constructive QFT developed by Brydges and Kennedy \cite{BrydgesKennedy}
and by Rivasseau Abdesselam, Gurau, Mangnen and Seneor \cite{AbdesselamRivasseau, RivasseauGurau, MagnenRivasseauSeneor} to analyse the logarithm of the generating function with an Hubbard Stratonovich transformation and to prove that the series defining the generating function of the correlation function admits a loop vertex expansion which  convergences  in some particular regime.
The range of parameters for which convergence is obtained, though improved wrt the results of 
\cite{GalandaPinamonti},
it is not sufficient to reach the Gross-Pitaevskii regime,
because of the scaling as $N$ of the condensate density $\phi_0^2$.
For that regime, at this first stage, only weaker notions of summability are obtained. 

In order to improve the convergence, and to understand the behaviour of the system in the GP regime, we explicitly analyse the most relevant tree diagrams arising from the loop vertex expansion in the limit of vanishing temperature. 
We identify the graphs which give origin to all the corrections of the scattering length in the GP regime to the first Born approximation, as particular ladder diagrams. 
These correspond to all the higher loop corrections for adjacent vertices present in the previously derived loop vertex expansion. 
We also recognize the class of tree diagrams which are relevant to the renormalisation of the background value of the field, relating the original beackground $\phi_0$ to the solution of the Gross-Pitaevskii equation containing the scattering length $\mathfrak{a}_0$ at the place of $\hat{v}(0)$ and denoted $\phi_R$. 

Finally, after renormalising the background to $\phi_R$ and considering partial resummations to get the dominant two-point function involving the scattering length (with dispersion relations $\sqrt{p^4 +  16\pi  \phi_R^2 \mathfrak{a}_0 p^2}$ in the limit of vanishing temperature), we prove that the remaining diagrams in the loop vertex expansion can be summed. Hence, showing that after the above steps, all the remaining contributions are less dominant in the considered regime.  

Below, further details on the methods and on the precise statements of our results are presented. 

\subsubsection{AQFT methods}
We directly work with the infinitely extended system. Therefore, equilibrium states cannot be described by a density operator, a trace class operator, on the bosonic Fock space of the vacuum representation \cite{KMS}. 
Nevertheless, the observables of the theory are still well defined and the commutation relations among them remain meaningful. 
In other words, the $*$-algebra of field observables is well defined also in this context. Therefore, in order to investigate the form of the equilibrium state in the presence of a non-vanishing background phase, we work with the
formalisms of AQFT \cite{HaagKastler} and with its more modern application used to
rigorously 
treat interacting quantum fields perturbatively
\cite{HollandsWald2001, HW02, BDF09, BrunettiFredenhagen00, Fredenhagen2015, KasiaBook, Dutsch:2019aa}. 
In this setting, the equilibrium states are characterised by the celebrated KMS condition \cite{KMS} which
extends the notion of Gibbs states to the case of infinitely extended systems.
Indeed, the failure of existence of a trace class operator to describe the density matrix at equilibrium on the vacuum representation in the infinite volume limit, is a manifestation of the fact that thermal equilibrium states and the vacuum state are not normal states. Namely, they correspond to different phases of the same system. 

\subsubsection{KMS state and Araki's perturbation theory}
If the generators of the perturbation of the free Hamiltonian are within the considered observables algebra, it is possible to analyse the form of the generating function of the thermal correlation function of the interacting theory perturbatively. 
 Actually, if we consider a particular system where
 $\rho_\beta^0 = e^{-\beta H_0}$, the density matrix of the free Gibbs state, and $\rho_\beta = e^{-\beta H}$, the density matrix of the interacting Gibbs state,
 can be given as suitable trace class operators on some suitable Hilbert space, we notice that they are related by the relative partition function
 \[
 \rho_\beta =e^{-\beta H}= e^{-\beta {H}} e^{\beta {H}_0} e^{-\beta {H}_0} = U(\mathrm{i}\beta) \rho_\beta^0.
 \]
In this case the relation of $\rho_\beta$ to $\rho_{\beta}^0$ is given by $U(i\beta)$, the relative partition function, which can be constructed as a suitable analytic continuation of the time ordered exponential of the interaction Hamiltonian $H-H_0$. 

Araki in \cite{Araki} showed that even when the equilibrium state cannot be described by a trace-class density matrix, a similar idea still holds. Actually, if it is  possible to describe the cocycle $U(t)$ which intertwines the free and interacting time evolution as the time ordering exponential of the interaction Hamiltonian $H_I$, an interacting equilibrium state can then be obtained by analytic continuation under mild technical assumptions. 
A similar construction can be used in the context of interacting quantum fields, as shown for the first time by Fredenhagen and Linder \cite{FredenhagenLindnerKMS_2014}, with the caveat that the state is there given as a formal power series in the coupling constant which multiplies the interaction Hamiltonian.
More precisely, denoting by $\omega^\beta$ the KMS state of the background (free) theory at inverse temperature $\beta$, the KMS state of the interacting theory used to evaluate a generic observable $O$  can be expressed as
\[
\omega^{\beta I} (O) = \frac{\omega^\beta(O T e^{- \int_0^\beta  H_I(u) du } )}{\omega^\beta( T e^{- \int_0^\beta  H_I(u) du } )},
\]
where now $T$ is a suitable time ordering in the imaginary time, which becomes meaningful thanks to the analytic properties of the KMS state $\omega^\beta$.

In this work, in order to use this expression, we 
farther elaborate on this 
analytic continuation 
and we introduce an auxiliary abelian parabolic theory, where fields depend also on a parabolic time $u$ and are periodic with period $\beta$. Then, we discuss how to derive the correlation functions of the original quantum field theory from this auxiliary model.
In particular, the generating function of the correlation functions is obtained as 
\[
Z(J,{J}^*) = \omega^{\beta I}(e^{\mathrm{i} \Psi(J)+\mathrm{i} \Psi^*({J^*})})
\]
where $J$ and $J^*$ are seen as independent and arbitrary smearing functions of the fields.

The analogy of this expression with the construction of interacting fields by means of stochastic quantization is manifest. Actually, $\omega^\beta( T e^{- \int_0^\beta  H_I(u) du } )$ can be understood as the partition function of a stochastic process of periodic parabolic time \cite{Hairer, Gubinelli1,Gubinelli2} see also the recent developments in \cite{DuchGubinelliRinaldi}. 
In the context of interacting quantum field theory, rigorous derivations of the measures which describe the quantum states, for boson with weakly non local interaction have been recently obtained in \cite{FroehlichKnowlesSchleinSohinger2, FKSS3, LPR05}. 
Instead, for the case of weakly interacting fields, with a non local interaction and in presence of a non trivial background \cite{Benfatto} has initiated an analysis, see also \cite{Salmhofer}, and more recently also the case of two space dimensions at vanishing temperature has been developed in \cite{SerenaGiuliani}.

\subsubsection{The loop vertex expansion (LVE)}
The interaction Hamiltonian density we are considering is fourth order in the fields.  Hence, in presence of a non-vanishing background it is such that
\[
\mathcal{H}_I= \phi_0^2
(\Psi+\Psi^*)(v*
(\Psi+\Psi^*)
+
\phi_0 (\Psi+\Psi^*)(v*|\Psi\Psi^*|)
+
\phi_0 |\Psi\Psi^*| (v*(\Psi+\Psi^*))+   \frac{1}{2}|\Psi\Psi^*| (v*|\Psi\Psi^*|).
\]
In order to treat the non linearities, and simplify them, we use an Hubbard-Stratonovich transformation.

This kind of transformations have been extensively used and discussed in the context of constructive interacting quantum field theory by Rivasseau \cite{RivasseauBook} and collaborators
\cite{AbdesselamRivasseau,RivasseauGurau}.
This corresponds to adding an extra auxiliary field $A$, the Hubbard-Stratonovich field, whose correlation functions are described by $-\delta(u)v(x)$. With this extra field at disposal, the interaction Hamiltonian density is 
\[
\mathcal{H}_I^A =  \phi_0 A 
(\Psi+\Psi^*)
+
A |\Psi\Psi^*| .
\]
Furthermore, since this is quadratic in the fields $\Psi$ and $\Psi^*$ and there is no direct backreaction of $\Psi$ to $A$, we can take into account its effect in the corresponding linear theory. 

Therefore, following the above reasoning, the interacting intermediate state which results from this procedure is 
\[
\omega^{\beta A} (O) 
 = \frac{\omega^\beta(O T e^{- \int_0^\beta  H^A_I(u) du } )}{\omega^\beta( T e^{- \int_0^\beta  H^A_I(u) du } )}
\]
where $H_I^A$ is the interaction Hamiltonian associated to the density $\mathcal{H}_I^A$,  
and for the reasons explained above it can be computed exactly.
In particular, the  generating function where $J$ and $J^*$ are independent smooth and compactly supported complex valued function is 
\[
Z_A(J,J^*) = \omega^{\beta A} (e^{\mathrm{i} \Psi(J^*)+\mathrm{i} \Psi^*({J})}).
\]
As we shall prove in section \ref{se:LVE}, the latter is such that 
\[
 \log Z_A(J,J^*)  = 
- {2} \Tr \int_0^{\sqrt{\lambda} } \dvol s A 
\left(\mathsf{G}_\beta^{ s A }-\mathsf{G}_\beta\right)
{+2}
   \langle (\mathrm{i}J+\phi_0 A), \mathsf{G}_\beta^{\sqrt{\lambda}A}  (\mathrm{i}J^*+\phi_0 A) \rangle.
\]
where $\mathsf{G}^A_\beta$ and $\mathsf{G}_\beta$ are respectively the fundamental solutions of $-\partial_u+{K}+A$ and $-\partial_u+{K}$ for functions in $\mathbb{R}^3\times [0,\beta]$ extended by periodicity to $\mathbb{R}^4$. If $A$ is of compact support, the trace on $\mathbb{R}^3\times [0,\beta]$ is well defined and the pairing is for periodic functions.

The last step to obtain $Z(J,J^*)$ consists in evaluating the auxiliary field in a state whose covariance is $-\delta \otimes v$ 
(to be more precise, $\mathrm{i}A$ is a Gaussian stochastic field and the covariance of the Gaussian state is $\delta \otimes v$). 
We realise this evaluation by applying the operator $e^{-\Gamma}$, with $\Gamma = \frac{1}{2}\langle \delta \otimes v,\delta^2/(\delta A \delta A) \rangle_2$, and afterwards setting $A=0$.
The generating function is then such that
\[
Z(J,J^*)=\left.\frac{e^{-\Gamma}Z_A(J,J^*))}{e^{-\Gamma}Z_A(0,0)}\right|_{A=0}.
\]
Proving that this evaluation furnishes a convergent expression is a priori not obvious. Actually, this is the step where the difficulties of the interacting theory reappear.

In \cite{GalandaPinamonti} we have shown that this procedure produces finite results when the parameters of the theory satisfy certain bounds. 
In this paper we improve the obtained bounds analysing 
$\log Z(J,J^*)-\log Z(0)$ 
by means of the loop vertex expansion (LVE) discussed extensively by Rivasseau and Abdesselam \cite{AbdesselamRivasseau,Rivasseau} and which takes its root in the seminal paper of Brydges and Kennedy \cite{BrydgesKennedy}.

According to this result $\log Z(J,J^*)$ can be expressed as a sum over the tree graphs among an arbitrary number of vertices and rooted in $1$.
Written in terms of $V=e^{-\Gamma} W$  where $W=-\log Z_A$ this is
\begin{equation}\label{Beliaev}
\log Z(J,J^*) =   -V
- \sum_{n=2}^\infty \frac{1}{n!}
\sum_{T\in\mathcal{T}_n}
\mathcal{M}
\prod_{b\in E(T)}\int_{0}^1  \dvol w_b\;  {\Gamma}_b
\exp\left[-\sum_{k< l} w_T(k,l) {\Gamma}_{kl}\right]
 \underbrace{V \otimes \dots \otimes V}_n \left. \right|_{A=0}
\end{equation}
where $\mathcal{T}_n$ denotes the set of tree graphs among $n$ vertices, $\mathcal{M}$ maps the $n$-th fold tensor product to the pointwise product, $E(T)$ denotes the set of edges of the tree $T$,  $\Gamma_b=\Gamma_{b_1b_2}$ and $\Gamma_{ij}=\langle \delta \otimes v ,\delta/{\delta A_i}\otimes \delta/{\delta A_j} \rangle_2$. Finally 
$w_T(k,l)$ is the minimum of the weights $w_b$ associated to the edges in the unique path in $T$ joining $k,l$.
We show in Section \ref{se:LVE} that the problem we are considering reduces to this form and for a proof of the validity of that formula, which is usually employed to estimate the coefficient of the Mayer expansion of the connected correlation functions of a generic field theory, we refer to the original paper \cite{BrydgesKennedy} and to \cite{AbdesselamRivasseau}.
We call it (BKAR) tree graph expansion.

\subsubsection{The convergence of LVE}

In Section \ref{se:convergence} we analyse the convergence of the loop vertex expansion of $Z(J,J^*)$, and we obtain the following result which corresponds to Theorem \ref{thm:converge} in the main text
\begin{theorem}\label{thm:1}
Consider the equilibrium state at inverse temperature $\beta$. Consider also the $\epsilon$ regularised operator $K_\epsilon$ as in \eqref{eq:Hamiltonian-density-field}, $\epsilon>0$. If $\phi_0$, $v$,  $J$ and $J^*$ used to construct $Z(J,J^*)$ are such that 
\[
\frac{r}{{4}}:=\left(\frac{\|v\|_1}{\epsilon} \right) \left( \frac{2(2\pi m)^{\frac{3}{2}}}{\sqrt{\beta}}+\left(\|\phi_0\|_\infty+\frac{\|J\|_{K,1}+\|J^*\|_{K,1}}{{2}}\right)^2 \right)< \frac{1}{16}
\]
where $\|\cdot\|_K$ is a suitable $K$ dependent norm given in \eqref{eq:normJK}, then the loop vertex expansion given in \eqref{Beliaev} and used to define
$
\log Z(J,J^*) - \log Z(0,0) 
$
is absolutely convergent. 
\end{theorem}
We stress that, in order to obtain this result, the internal $U(1)$ symmetry of the original system ($\Phi\to e^{\mathrm{i}\theta}\Phi$, $\Phi^*\to e^{-\mathrm{i}\theta}\Phi^*$) has been explicitly broken by the $\epsilon$ (infrared) regulator used to perturb the kinetic term of the Hamiltonian describing the evolution of the fluctuations $\Psi$ and $\Psi^*$.
In the state we aim to construct this $U(1)$ symmetry is however only spontaneously broken, hence, limits where the $\epsilon$-regularization is removed need to be considered.
On the other hands, the necessity of studying the limits where the spectrum of the operator governing the fluctuations become gapeless is a consequence of the Goldstone theorem, see  \cite[Theorem 15.1]{Strocchi}, and not an artifact of the used method. 
For the case of trapped bosons we refer to 
\cite{LSY-07} and for a discussion of the Goldstone theorem in relativistic systems at finite temperature to \cite{BFP21}. 
The employed range of parameters allows to keep convergence in the limits where $\epsilon$ vanishes, provided that also $v$ tends to zero with analogous rate while keeping $\|v\|_1/\epsilon$ small and at the same time $\phi_0^2$ small.

Unfortunately, when $v$ is substituted by $v_N$ in \eqref{eq:vn}, it is true that $\|v_N\|_1$ vanishes in the limit of large $N$ however, keeping fixed the chemical potential $\tilde{\mu}$, under the very same limit $\phi_0^2$ grows as $N$.
For this reason, the hypotheses of this Theorem are not satisfied in the Gross-Pitaevskii regime for large $N$
where $\phi_0^2$ grows as $N$ and where the quantity which remains constant is $\hat{v}(0)\phi_0^2$.
In order to overcome this problem and to get information about that regime, we analyse the terms of the loop vertex expansion in the limit of vanishing temperature. 
Indeed, in this limit the vertex corresponding to the loop, responsible of the term $C/\sqrt{\beta}$ in $r$, is of lower order in $N$
and can thus be discarded. Only the contributions involving $\phi_0$ matters. 
We start our analysis discarding the effect of $e^{-\sum w_T \Gamma}$ in the loop vertex expansion in Section  \ref{se:scattering-length-corrections},
where we see that the tree graphs contributions of the loop vertex expansion with only the leg vertex 
$ \langle (\mathrm{i}J+\phi_0 A), \mathsf{G}_\beta^{\sqrt{\lambda}A}  (\mathrm{i}J^*+\phi_0 A) \rangle$ in $W(J,J^*)$
can be summed. This gives origin to an expression for the truncated two-point function which is well defined in the limit  of $\epsilon\to 0$ also for large $\phi_0$.  

We see in \eqref{eq:omegaJJ} and \eqref{eq:omegaJJ1}, and according to the analysis of Section \ref{se:pretrubed-propagators}, that the propagators of the theory in this approximation are very close to those predicted by Bogoliubov \cite{BogoliubovBEC}. 
They are characterised by dispersion relations of the form $\sqrt{p^4 +2 \hat{v}(0)\phi_0^2 p^2}$ and obtained from an interaction Hamiltonian density of the form 
\begin{equation}
\label{eq:ren-H}
\mathcal{H}_{Bog}=(\Psi+\Psi^*) \phi_0^2 \hat{v}(0)  (\Psi+\Psi^*)
=(\Psi+\Psi^*) \hat{\mathcal{V}}_0(0,0) (\Psi+\Psi^*)
\end{equation}
where $\hat{\mathcal{V}}_0(0,0)= \phi_0^2 8\pi \mathfrak{a}_0^{(0)} = \phi_0^2 \hat{v}(0)$ with $\mathfrak{a}_0^{(0)}$ the lowest order contribution to the Born series defining the scattering length of the potential $v$, see Appendix \ref{se:scattering-length} and equation \eqref{eq:born-series}.

We then proceed in Section \ref{se:higher-loop-corrections} to analyse higher loop corrections attached to the tree graphs, producing the Bogoliubov transformation of the correlation functions discussed above by the action of $e^{-\sum w_T\Gamma}$ in the loop vertex expansion.
We prove that the most relevant contributions in the Gross-Pitaevskii regime at very low temperature, arises when the loop diagrams are formed applying 
$e^{-\sum w_T\Gamma}$ to adjacent vertices only. 
We see in Section \ref{se:higher-loop-corrections}
that the diagrams which are relevant, 
in the limit of vanishing temperature
are particular ladder diagrams and the effect they have, in the Dyson series defining the linearised propagator, is to effectively renormalise the interaction Hamiltonian with an  effective vertex
\begin{equation}\label{eq:Vintro}
 \hat{\mathcal{V}}(0,0)=\sum_{n}\hat{\mathcal{V}}_n (0,0) =  \phi_0^28\pi \mathfrak{a}_0 
\end{equation}
where all the contributions in the Born series defining the scattering length appear.

In Section \ref{se:one-point-function-renormalisation} we see how this renormalised effective vertex modifies also the background. In particular, the background value of the field $\phi_0$ gets also renormalised to $\phi_R$ which is now a solution of the Gross-Pitaevskii equation containing the whole scattering length $\mathfrak{a}_0$.

We then consider the renormalisation of the background value of the field to $\phi_R$, and in the LVE \eqref{Beliaev} we consider the partial resummations of the tree diagrams which produce the propagators with the scattering length.
We prove in Theorem \ref{thm:partial-sum} that in the limit $N\to\infty$ at fixed $\epsilon$ the LVE with the 
partial resummations coincide with the original LVE \eqref{Beliaev}. 
Finally, revisiting the decomposition and the convergence proof of the whole loop vertex expansion we get the following theorem, which corresponds to 
\ref{thm:1} in the text.
\begin{theorem}
Consider the loop vertex expansion of $\log Z(J,J^*) - \log Z(0)$ obtained with a decomposition
of the field $\Phi = \phi_R+X$, 
$\Phi^* = \phi_R+X^*$
with respect to the renormalised background $\phi_R$ which is now a constant solution of the Gross-Pitaevskii equation \eqref{eq:GR-ren-background}.
Consider explicitly the result of the partial resummations of the tree graphs which produce in the leg vertices $\tilde{\mathsf{G}}_{\mathcal{L}}^A$, which is the fundamental solution of 
$\textrm{diag}{(-\partial_u + K+ A,\partial_u + K+ A)}+\mathcal{L}^\beta\mathsf{J}_2$
where $\mathsf{J}_2$ is the matrix with all entries equal to $1$ and $\mathcal{L}^\beta$ is operator implicitly given in \eqref{eq:ell-scattering} (which reduces to $\mathcal{V}$ in \eqref{eq:Vintro} with $\phi_0$ replaced by $\phi_R$ in the case of vanishing temperature). 
The remaining 
loop vertex expansion of $\log Z(J,J^*)-\log Z(0)$ converges absolutely for arbitrarily large $N$, for $\epsilon\sim N^{-1}$ and with $\mathfrak{a}_0 \phi_R^2$ kept constant, if $\beta\to \infty$. 
\end{theorem}

It is interesting to notice that the estimate of the remaining graphs are convergent if the following inequality holds
\[
 \left(\frac{\|v_N\|_1}{\epsilon} \right) \left( \frac{2(2\pi m)^{\frac{3}{2}}}{\sqrt{\beta}}+\left(\frac{\|8\pi \mathfrak{a}_0\phi_R\|_\infty}{\sqrt{N}}+\frac{\|J\|_{K,1}+\|J^*\|_{K,1}}{2}\right)^2 \right)< \frac{1}{16}
\]
which improves the result of Theorem  \ref{thm:1} above and permits to cover also the interesting case of the Gross-Pitaevskii regime at least when the temperature tends to $0$.

Finally, we observe that a first bound satisfied by the critical temperature $T_c$ above which no condensation occurs is obtained imposing finite density of the renormalised condensate in such a way that ($\phi_R^2=\tilde\mu/(8\pi \mathfrak{a}_0)>0$) and thus
\begin{equation*}
    \beta^{-3/2}\mathrm{Li}_{\frac{3}{2}}(e^{-\beta\epsilon  })< \left(\frac{2 \pi}{m}\right)^{\frac{3}{2}} \frac{\tilde\mu_\infty}{ \hat{v}(0)}
\end{equation*}
where $\tilde\mu_\infty=\tilde\mu+q_\beta
\hat{v}(0)
$ and $\mathrm{Li}_{\frac{3}{2}}$ is the polylogarithmic function  of order $3/2$ given in integral form in equation \eqref{eq:qbeta}.
We observe that in the limit $\epsilon\to 0$ and for $k_B=1$, the previous inequality gives for the critical temperature
\[
T_c<\frac{4\pi}{2m} {\zeta\left(\frac{3}{2}\right)^{-\frac{2}{3}}}\rho_\infty^{\frac{2}{3}}
\]
where $\zeta$ is the Riemann zeta function, 
$\rho_\infty=\frac{\tilde\mu_\infty}{ \hat{v}(0)}$ is the maximal value of the density of the background for fixed $\tilde{\mu}_\infty$ and $\hat{v}(0)$.
The quantity at the right hand side coincides with $T_{fc}$, the critical temperature of the free Bose Gas at density $\rho_\infty$, see e.g. \cite{NapiorkowskiSolovej}.

\subsection{Content of the paper}
In Section \ref{se:framework} we present the analysis of the observable algebra for a complex charged non-relativistic field and
the extension of this algebra to encompass Wick squares. We recall the form of the KMS state and the construction of the interacting KMS state in the associated parabolic theory.
In Section \ref{se:hubbard-stratonovich} we discuss the Hubbard-Stratonovich transformation and how it modifies  the generating function of the correlation functions of the theory up to the evaluation of the auxiliary field in a suitable Gaussian state.
Section \ref{se:partitionfunction-loopvertex} contains the discussion about the evaluation of the auxiliary field in the Gaussian state and the combinatoric which gives origin to the loop vertex expansion of the logarithm of the partition function. In Section \ref{se:convergence} we present the restults about convergence of this loop vertex expansion in suitable regimes with the corresponding estimates. Section \ref{se:scattering-length-corrections} contains the analysis of the diagrams which are relevant to renomalize the two-point function of the fluctuations and of the one point function of the theory.
Then, in Section \ref{se:renormalisation} we discuss the corresponding renormalisation of the loop vertex expansion proving its convergence in the GP regime. 
Finally, in Appendix \ref{se:thermal-propagators-appendix} we collect some technical results about thermal propagators and their perturbation and in Appendix \ref{se:scattering-length} some basic facts about the scattering length.

\section{Framework}\label{se:framework}

\subsection{Complex non relativistic fields}

We are considering a charged scalar field theory that, due to the $U(1)$ symmetry, is described by a complex valued scalar field. We follow \cite{ArakiWoods} to describe the finite temperature case.
In particular, the observables of the linear theory are generated by the linear fields 
and are denoted by 
\[
\Phi(f) = \int  \Phi(x) f(x) d^3x, \qquad  \Phi^{*}(f) = \int \Phi^*(x) f(x)  d^3x, \qquad f\in \mathcal{S}(\mathbb{R}^{3}).
\]
They satisfy the {\bf commutation relations}
\[
[\Phi(f),\Phi^*(h)] = \int f h d^3x = ( \bar{f},h )
\]
where $( \cdot,\cdot )$ denotes the ordinary sesquilinear product of $\mathfrak{H}=L^2(\mathbb{R}^3)$, which contains $\mathcal{S}(\mathbb{R}^{3})$ so that the extension to $\mathfrak{H}$ is straightforward.
The {\bf antilinear involution} on the algebra is
\[
(\Phi(f))^* = \Phi^*(\overline{f})
\] 
which is compatible with the commutation relations. The linear fields together with the identity and the operations enumerated above generate the algera of quantum fields $\mathcal{A}_F$. 
The free {\bf time evolution} is generated by a suitable selfadjoint operator $\mathsf{K}$ which acts on a suitable domain in $\mathfrak{H}$. The generator $\mathsf{K}$ is further assumed to be real, namely it commutes with the complex conjugation
$\overline{\mathsf{K}f} = {\mathsf{K}\overline{f}}$.
The free time evolution is then a one parameter group of $*$-automorphisms $\tau_t$ of $\mathcal{A}_F$ and its action on the generators is
\[
\tau_t (\Phi(f)) = \Phi(e^{-\mathrm{i} t \mathsf{K}}f), \qquad \tau_t (\Phi^*(f)) = \Phi^*(e^{\mathrm{i} t \mathsf{K}}f).
\]

At least when the real selfadjoint operator $\mathsf{K}$ is strictly positive we can construct the corresponding unique 
equilibrium state $\omega^\beta$ at inverse temperature $\beta$.
If the spectrum of $\mathsf{K}$ reaches $0$ condensation might occur.
The thermal two-point functions with respect the free time evolution of the free theory are thus
\begin{equation}\label{eq:Boperators}
\begin{aligned}
\omega^\beta(\Phi(f)\Phi^*(h)) &= \mathsf{B}_-^\beta(f,h) = \int f(x) \left(\frac{1}{1-e^{-\beta \mathsf{K}}}h\right)(x) d^3x = ( \bar{f}, B_- h)
\\
\omega^\beta(\Phi^*(h)\Phi(f)) &= \mathsf{B}_+^\beta(h,f) = \int f(x) \left(\frac{e^{-\beta \mathsf{K}}}{1-e^{-\beta \mathsf{K}}}h\right)(x) d^3x = ( \bar{f}, B_+ h)
\end{aligned}
\end{equation}
while $\omega^\beta(\Phi(h)\Phi(f))=0$ and $\omega^\beta(\Phi^*(h)\Phi^*(f))=0$.
Notice that the operators $B_\pm$ are selfadjoint and if $\mathsf{K}$ is strictly positive also bounded. Since $\mathsf{K}$ is real, it commutes with complex conjugation, also $B_\pm$ are real and
$(\bar{f},B_\pm h) = (\overline{B_\pm f}, h)$. 
The compatibility of the two-point functions with the commutation relations follows from the observation that $B_--B_+=\mathbb{1}$ and thus also the corresponding two-point distributions $\mathsf{B}_\pm^\beta$ are such that 
\[
\mathsf{B}^\beta_-(x,y)-\mathsf{B}^\beta_+(y,x)=\delta(x,y).
\]
The quasi-free state constructed in this way satisfies the KMS condition \cite{KMS} at inverse temperature $\beta$ with respect to $\tau_t$. 

The limit $\beta\to\infty$ gives back the vacuum $\omega^\infty$.
We also observe that the ordinary description in terms of operators on suitable Hilbert spaces is recovered making use of the Gelfand-Naimark-Segal (GNS) representation. It is however important to notice that the GNS representations of $\mathcal{A}_F$ furnished by the GNS construction corresponding to $\omega^\beta$ and $\omega^\infty$ are not equivalent. Namely, it is not possible to describe one of the two states with a trace class operators on the Hilbert space of the other.

In this paper, we shall consider as free time evolution the one parameter group of $*$-automorphisms of the algebra $\mathcal{A}_F$ generated by 
$\mathsf{K}=K_{\epsilon}= -\frac{\Delta}{2m}+\epsilon$, with $\epsilon>0$, defined on $D(K_{\epsilon}) = H_s^2 \subset L^{2}(\mathbb{R}^3)$. It is immediate to check that $K_{\epsilon}$ is selfadjoint, strictly positive and real. Then, the two-point function of the theory admits a simple representation in the Fourier domain as
\begin{equation}\label{eq:Bpm}
\begin{aligned}
\omega^\beta(\Phi(f)\Phi^*(h)) &= \mathsf{B}_-^\beta(f,h) 
= \frac{1}{(2\pi)^3}\int 
{\hat{{f}}(-p)} \frac{1}{1-e^{-\beta \left(\frac{p^2}{2m}+\epsilon\right)}} \hat{h}(p)d^3p
\\
\omega^\beta(\Phi^*(h)\Phi(f)) &= \mathsf{B}_+^\beta(h,f) 
= \frac{1}{(2\pi)^3}\int 
{\hat{{h}}(-p)} \frac{e^{-\beta \left(\frac{p^2}{2m}+\epsilon\right)}}{1-e^{-\beta \left(\frac{p^2}{2m}+\epsilon\right)}} \hat{f}(p)d^3p.
\end{aligned}
\end{equation}
Finally, we observe that, as distributions, $\mathsf{B}_\pm^\beta$ are symmetric. Namely, it holds
$\mathsf{B}_+^\beta(f,h) =\mathsf{B}_+^\beta(h,f)$ and
$ \mathsf{B}_-^\beta(f,h) =\mathsf{B}_-^\beta(h,f)$.

\subsection{Extension}

To describe the generator of the interacting time evolution in a $*$-algebra, we extend it encompassing objects like the normal ordered Wick squares in the algebra.
Namely, we need to add to the set of generators also normal ordered Wick squares, formally denoted 
\begin{align*}
:|\Phi|^2:(f) 
&= \int \Phi^*(x)\Phi(x)  f(x)d^3x
\end{align*}
where $f$ is smooth and compactly supported.
More precisely, the normal ordered Wick square is obtained as 
\begin{align*}
:\Phi^*\Phi:(A)  = \frac{1}{2} \langle  
\Phi^*\Phi + \Phi\Phi^* - \omega^\infty(\Phi^*\Phi + \Phi\Phi^*) , A \rangle 
&= \int \Phi^*(x)  a(x,y) \Phi(y)d^3x d^3y
\end{align*}
where $a$ is a symmetric distribution which generates a bounded operator $A$ on $\mathfrak{H}$.
For our purposes it is necessary to have operators of the form $ A=e^{-itK_{\epsilon}} f e^{itK_{\epsilon}}$ 
where $f$ is a real valued compactly supported smooth function used as a multiplicative operator. 
These operators give origin to well defined commutation relations. Denoting by $\mathcal{B}_s(\mathfrak{H})$ the bounded self-adjoint operators on $\mathfrak{H}=L^2(\mathbb{R}^3)$, the commutation relations satisfied by these abstract generators are
\begin{equation} \label{eq:extended-commutation}
    \begin{aligned}
    [: \Phi^*\Phi :(A),:\Phi^*\Phi:(B)] &= -\mathrm{i}:\Phi^*\Phi:(\mathrm{i}[A,B]), \qquad \qquad & A,B \in \mathcal{B}_s(\mathfrak{H})\\
    [: \Phi^*\Phi :(A),\Phi(f)] &= \Phi(A f), \qquad \qquad & A \in \mathcal{B}_s(\mathfrak{H}), \quad f \in \mathfrak{H}
    \\
    [: \Phi^*\Phi :(A),\Phi^*(f)] &= \Phi^*(A f), \qquad \qquad & A \in \mathcal{B}_s(\mathfrak{H}), \quad f  \in \mathfrak{H}.
    \end{aligned}
\end{equation}
The $*$-operation acts on 
$:\Phi^*\Phi(A):$ 
in the following way
\begin{equation}\label{eq:extended-*}
    (:\Phi^*\Phi:(A))^* =
    :\Phi^*\Phi:(A^\dagger). 
\end{equation}
We notice that $\{1,\Phi(f), \Phi^*(h), :\Phi^*\Phi:(A)\}$ where $f,h\in\mathfrak{H}$ and $A$ is in $\mathcal{B}_s(\mathfrak{H})$ generate a $*$-algebra. 
However, this algebra is too large, in fact it cannot be tested on $\omega^\beta$, actually
\begin{align*}
\omega^\beta (:\Phi^*\Phi:(A)) &= \Tr (B_+ A).
\end{align*}
and if $B_+A$ is not trace class the latter is not well defined. 
We furthermore observe that for $A = e^{itK_{\epsilon}} fe^{-itK_{\epsilon}}$ and for $K_{\epsilon} = -\Delta/(2m) +\epsilon$, with $\epsilon>0$, it holds that 
\[
\Tr (B_+ A) = \int\frac{e^{-\beta \left(\frac{p^2}{2m}+\epsilon\right)}}{1-e^{-\beta \left(\frac{p^2}{2m}+\epsilon\right)}} \hat{f}(p)d^3p
\]
which is well defined.
Therefore, the set of operators we are working with is
\[
O =\{A\in \mathcal{B}(\mathfrak{H}), A^\dagger=A, \Tr (B_+A)<\infty \},
\]
the set of weighted trace class operators.
We then use as generators $\{1,\Phi(f), \Phi^*(h), :\Phi^*\Phi:(A) \}$ where $f,h\in\mathfrak{H}$ and $A \in O$.
We denote this algebra by $\mathcal{A}$ and the extension of the quasi-free state $\omega^\beta$ defined above in \eqref{eq:Boperators} to this extended algebra is then straightforward.
We describe this procedure in a concrete realisation of $\mathcal{A}$ presented in the next subsection.

\subsection{Thermal representation}

We construct a concrete representation of $\mathcal{A}$ and of the quasi-free state $\omega^{\beta}$ on $\mathcal{A}$ using the functional formalisms introduced by \cite{BDF09} see also \cite{KasiaBook, Dutsch:2019aa}.
The representation
$\mathcal{A}_\beta$ is given as a subset of the set of smooth functionals $\mathcal{F}$ over the off shell field configurations $(\varphi,\varphi^*)\in \mathcal{C} = (C^{\infty}(M)\cap L^2(\mathbb{R}^3))^{\times 2}$. Here we require the field configurations to be in $L^2$ to be compatible with operators used to smear the Wick squares.

The {\bf product} in $\mathcal{F}$ is constructed by means of $B_{\pm}$ given in \eqref{eq:Boperators} and $\mathsf{B}^\beta_{\pm}$ 
in 
\eqref{eq:Bpm}
as follows
\[
F_1\star_\beta F_2 = \mathcal{M} e^{ D_{12}}
 F_1\otimes F_2,
\]
where $\mathcal{M}$ maps the tensor product into the pointwise product, and the operator $D_{12}$ acts on a tensor product of functionals defined as 
\begin{equation}\label{eq:Dij}
D_{ij}  =\int  d^3x d^3y \mathsf{B}^\beta_-(x,y)  \frac{\delta}{\delta \varphi_i(x)}\otimes \frac{\delta}{\delta \varphi_j^*(y)}+\int  d^3x d^3y \mathsf{B}^\beta_+(x,y)  \frac{\delta}{\delta \varphi_i^*(x)}\otimes \frac{\delta}{\delta \varphi_j(y)}  
\end{equation}
where the functional derivatives $\frac{\delta}{\delta \varphi_i}$, $\frac{\delta}{\delta \varphi^*_i}$ act on the $i$-th element of the tensor product.

The
{\bf $*$-operation} is given in terms of the complex conjugation and it acts on $F$ as 
\[
F^*(\varphi,\varphi^*)= \overline{F(\overline{\varphi^*},\overline{\varphi})}.
\]
The map $\alpha_\beta$, which realises the representation $\mathcal{A} \mapsto \mathcal{A}_\beta$, is a $*$-homomorphism of algebras defined on the generators of the algebra in terms of $\mathsf{B}^{\beta}_\pm$ (with symmetric integral kernel) in the following way
\[
\phi(f):=\alpha_{\beta} \Phi(f) = \int f(x) \varphi(x) d^3x, \qquad  
\phi^*(f):=\alpha_{\beta} \Phi^*(f) = 
 \int f(x) \varphi^*(x) d^3x,
 \]
and, denoting by $a(x,y)$ the 
symmetric and real
kernel of $A$,
 \begin{align*}
\alpha_\beta :\Phi^*\Phi(A): 
&= \int a(x,y)\varphi^*(x)\varphi(y) d^3xd^3y + \int a(x,y)B_{+}(x,y) d^3xd^3y
\\
&= (\overline{\varphi^*},A \varphi) + \Tr(B_+ A)\\
&= \alpha_\beta:\Phi^*\Phi(A):_\beta +  \Tr(B_+ A).
 \end{align*}
Where we introduced the notation
\begin{align*}
\alpha_\beta:\Phi^*\Phi(A):_\beta = :\phi^*\phi:_\beta(A):= (\overline{\varphi^*},A \varphi) 
 \end{align*}
This leads to the following definition
\begin{definition}
$\mathcal{A}_\beta$ is the smallest subset of $\mathcal{F}$ which contains 
\[
\{1,\phi(f),\phi^*(f),:\phi^*\phi:_\beta(A)\},
\]
where $f\in \mathfrak{H}$ and $A\in O$,
and which is closed under the $\star_\beta$ product and the $*$-operation.
\end{definition}
The $*$-homomorphism $\alpha_\beta$, defined on the generators, extends to a $*$-isomorphisms of algebras
\[
\alpha_\beta \mathcal{A} \to \mathcal{A}_\beta.
\]
This follows because, as it is possible to check,
\[
\alpha_\beta(AB) = (\alpha_\beta A)\star_\beta (\alpha_\beta B), \qquad (\alpha_\beta A)^* = \alpha_\beta (A^*).
\]

Finally, notice that the expectation values of the elements of $\mathcal{A}$ in the KMS state $\omega_\beta$, correspond to the evaluation on the vanishing field configurations in $\mathcal{A}_\beta$. Actually, 
\[
\omega^\beta (A) = \mathrm{ev}_0{(\alpha_\beta A}) := \left. \alpha_\beta A \right|_{\varphi=0,\varphi^*=0}.
\]

\subsection{Thermal states for the interacting theory}\label{se:thermal-state-interacting-theory}

\subsubsection{Interaction Hamiltoninan and chemical potential renormalisation}

Suppose now to have a perturbation of the free time evolution given by an interaction Hamiltonian in $\mathcal{A}$ formed by two contributions: a quadratic term, that we wish to treat perturbatively and that is proportional to the chemical potential $\mu$, and a second contribution which is quartic in the fields.
In total, it takes the form
\begin{align*}
{H}^r_I(g)
&= \frac{1}{2}\int :|\Phi|^2(x)|\Phi|^2(y): v(x-y)    {g(x)g(y)} d^3 x\; d^3y
- \mu :|\Phi|^2:(g^2)
 \\
&= \frac{1}{2}\int  :|\Phi|^2:(x)  v(x-y) :|\Phi|^2:(y)  {g(x)g(y)}d^3 x\; d^3y
- \left(\frac{v(0)}{2}+\mu\right) :|\Phi|^2:(g^2)
\end{align*}
where $g$ is an adiabatic cutoff which is eventually removed. In the limit $g\to1$, the contribution $-v(0) :|\Phi|^2:(g^2)/2$ results in a renormalisation of the chemical potential $\mu$. Thus, we work with a renormalised chemical potential of the form 
\[
\tilde{\mu}_{\infty} = \mu+\frac{v(0)}{2}. 
\]
Since $v$ is smooth and compactly supported,  up to a completion, $H^r_I(g)$ is an element of $\mathcal{A}$.

We now show how the representation of these elements in $\mathcal{A}_\beta$ induces an additional chemical potential renormalisation. We recall that 
\begin{align*}
\alpha_\beta (H_I^r) 
=&  \frac{1}{2}\int (\alpha_\beta :|\Phi|^2:(x)) \star_\beta (\alpha_\beta :|\Phi|^2:(y)) v(x-y) {g(x)g(y)}  d^3 x\; d^3y
\\
=&  \frac{1}{2}\int :|\phi|^2:_\beta(x) \star_\beta  :|\phi|^2:_\beta(y) v(x-y) {g(x)g(y)}  d^3 x\; d^3y
\\&+
\int :|\phi|^2:_\beta(x)   q_\beta(y) v(x-y) {g(x)g(y)} d^3 x\; d^3y
+
\frac{1}{2}\int q_\beta(x) q_\beta(y) v(x-y) {g(x)g(y)}  d^3 x\; d^3y
\end{align*}
where $q_\beta(x) = B_{+}(x,x)$
and thus it depends on the choice of the considered linear theory.
If the dynamics of the linear theory is determined by $K_{\epsilon} = -\frac{\Delta}{2m} +\epsilon$, for strictly positive $\epsilon$, it is equal to
\begin{equation}\label{eq:qbeta}
q_\beta = \frac{1}{(2\pi)^3} \int 
\frac{e^{-\beta \left(\frac{p^2}{2m}+\epsilon\right)}}{1-e^{-\beta \left(\frac{p^2}{2m}+\epsilon\right)}} d^3p = \left(\frac{m}{2 \pi \beta}\right)^{\frac{3}{2}} \mathrm{Li}_{\frac{3}{2}}(e^{-\beta \epsilon}).
\end{equation}
We observe that this second contribution can be understood, in the limit $g \to 1$, as further, state dependent, chemical potential renormalisation which appears when we represent the interaction Hamiltonian in $\mathcal{A}_\beta$
\begin{equation}\label{eq:ren-chem-pot}
\tilde{\mu}_\beta = \mu + \frac{v(0)}{2} - q_\beta\hat{v}(0).
\end{equation}
Moreover, $\tilde{\mu}_\beta-\tilde{\mu}_\infty$ vanishes in the limit $\beta\to\infty$. 

The last contribution in $\alpha_\beta(H^r_I)$ is a constant that can be discarded. Therefore, up to this temperature dependent renormalisation of the chemical potential, we use 
\begin{equation}\label{eq:HIbeta}
H_I = \frac{1}{2} \int :|\Phi|^2:_\beta(x) v(x-y) :|\Phi|^2:_\beta(y) g(y) g(x)  d^3 x\; d^3y -\tilde{\mu}_\beta :|\Phi|^2:(g^2)
\end{equation}
as interaction Hamiltonian of our system. 
Notice that if $\tilde{\mu}_\beta$ is negative, by the principle of perturbative agreement \cite{HW05, DragoHackPinamonti}, we may shift the corresponding contribution to the linear theory removing it from the interaction Hamiltonian. 

However, in order to have a non trivial condensate,
we shall require $\tilde{\mu}_{\beta} > 0$. In this case, we shall use a non trivial background value of the field to be able to shift the part proportional to the chemical potential in the free theory keeping positivity of its spectrum.
As we shall see below, the requirement of a positive renormalised chemical potential furnishes a first bound on  the critical temperature above of which no condensate can occur.

\subsubsection{Araki state for interacting theory}
Up to the temperature dependent renormalisation of the chemical potential \eqref{eq:qbeta}, we may thus consider the interaction Hamiltonian $H_I$ given in \eqref{eq:HIbeta}. This is an an element of $\mathcal{A}$ (or of $\mathcal{A}_\beta$ if we work in the thermal representation introduced above),
hence we may use it to perturb the time evolution at the algebraic level finding, at least in the sense of formal power series, the one parameter group of $*$-automorphisms of $\mathcal{A}$ describing the interacting time evolution $\tau_t^I$. We have that for every element of $\mathcal{A}$
\[
\tau^I_{t} (A) = U(t) \tau_t(A) U(t)^*
\]
where $U(t)$ is the cocycle which intertwines the free and interacting time evolution and it is given in terms of the interaction Hamiltonian formally as
\[
U(t) = e^{ i t (H_0+H_I)}e^{ -i t H_0}
\]
less formally it holds that 
\[
U(0)=1,\qquad U(t+s)=U(t) \tau_t( U(s)), \qquad  -\mathrm{i}\frac{d}{d t} U(t)=U(t) \tau_t({H}_I)
\]
hence
\[
-\mathrm{i}\frac{d}{dt} \tau_t^I(A) =  \tau^I_t([H_I,A]), \qquad A\in \mathcal{A}.
\]
Using recursively this relation we can give $H_I$ and $U(t)$ as a formal power series with coefficients in $\mathcal{A}$. 

With $U(t)$ at disposal, after analysing its analytic continuation, we can construct an equilibrium state for the interacting theory following the theory of Araki \cite{Araki} (see also \cite{FredenhagenLindnerKMS_2014} for $*$-algebras). The resulting state is 
given as
\[
\omega^{\beta I}(A) 
= \frac{\Tr (e^{-\beta H} A) }{\Tr (e^{-\beta H} )} = 
\frac{\omega^\beta(e^{-\beta H}e^{\beta H_0} A)}{\omega^\beta(e^{-\beta H}e^{\beta H_0})} 
=
\frac{\omega^\beta(U(\mathrm{i}\beta) A)}{\omega^\beta(U(\mathrm{i}\beta))}. 
\]
More rigorously, 
if $U(\mathrm{i}\beta)$ is known as a formal power series with  coefficients in $\mathcal{A}$, 
the formula in the last equality above acquires also meaning at least as formal power series with coefficients under control.
Actually, the following expansion holds
\begin{equation}\label{eq:thermal-correlations}
\omega^{\beta I} (A) = \sum_{n\geq 0} (-1)^n \int_{\beta S_n} \omega^{\beta}_C (A \otimes \tau_{i u_1} H_I \otimes \dots \otimes \tau_{i u_n} H_I)  du^n 
\end{equation}
where $\beta S_n$ is the $n$ dimensional simplex of edge $\beta$, namely 
\[
\beta S_n = \{t\in \mathbb{R}^n | 0< t_1 <\dots <t_n < \beta\}.
\] 
Furthermore $\omega^\beta_C(O_1\otimes \dots \otimes O_n)$ denotes the {\bf connected correlation functions} of $\omega^\beta$, they can be computed in the algebra of observables following e.g. Proposition 2.5 in   \cite{GalandaPinamonti}.
In particular, representing  $O_j \in \mathcal{A}_\beta$ by means of $\alpha_\beta$, the connected correlation function are computed as
\[
\omega^\beta_C (O_1\otimes\dots \otimes O_n) = \sum_{G\in \mathcal{G}^{c}_n}\left. \mathcal{M} \prod_{\{i,j\}\in E(G) \atop i<j}(e^{D_{ij}}-1)
\left(\alpha_\beta O_1\otimes \dots \otimes  \alpha_\beta O_n \right) \right|_{\phi=0,\phi^*=0},  \quad O_i\in \mathcal{A}, \quad \alpha_\beta O_i \in \mathcal{A}_\beta
\]
where $\mathcal{G}^c_n\subset \mathcal{G}_n$ is the set of simple connected graphs among $n$ vertices and the $i$-th and $j$-th functional derivative in $D_{ij}$ act respectively on $O_i$ and $O_j$ and $D_{ij}$ is given in \eqref{eq:Dij}. The extension to the imaginary time of $F(t) = \omega^\beta_C (\tau_{t_1} O_1\otimes\dots\otimes \tau_{t_n} O_n)$, for $\Im(t)  \in \beta S_n$, is performed extending the corresponding correlation functions.

\subsubsection{Parabolic Time ordering}

The form of the expectation values given in \eqref{eq:thermal-correlations} suggest to use an auxiliary parabolic theory to obtain the correlation function or the generating functions of the equilibrium theory.
To obtain this new representation we observe that 
\eqref{eq:thermal-correlations}
simplifies to the evaluation of a suitable time ordered exponential. 
To be more precise, we introduce a {\bf parabolic (anti) time ordering} operator for $u_i\in [0,\beta)$
\[
T(\tau_{\mathrm{i}u_1} O_1\otimes \dots \otimes \tau_{\mathrm{i}u_n} O_n) = \tau_{\mathrm{i}u_{j_1}} O_{j_1}\otimes \dots  \otimes \tau_{\mathrm{i}u_{j_n}} O_n
\]
where $j:\{1,\dots, n\}\to\{1,\dots, n\} $ is the permutation constructed in such a way that
\[
u_{j_1} <  u_{j_2} < \dots < u_{j_n}.  
\] 
The extension to $u_i \in \mathbb{R}$ is done by periodicity. Namely, for every $j$ 
\[
T(\tau_{\mathrm{i}u_1} O_1\otimes \dots \otimes \tau_{\mathrm{i}u_j}O_j \otimes \dots \otimes  \tau_{\mathrm{i}u_n} O_n) = T(\tau_{\mathrm{i}u_1} O_1\otimes \dots \otimes \tau_{\mathrm{i}(u_j+\beta)}O_j \otimes \dots \otimes  \tau_{\mathrm{i}u_n} O_n).
\]
With this map at disposal we have that the expectation values in the interacting thermal state given in \eqref{eq:thermal-correlations} takes the form 
\begin{align*}
\omega^{\beta I}(A) 
&= 
 \sum_{n\geq 0} \frac{(-1)^n}{n!} \int_{ (0,\beta)^n} \omega^{\beta}_C (T(A\otimes \tau_{\mathrm{i} u_1} H_I \otimes  \dots \otimes \tau_{\mathrm{i} u_n} H_I))  du_1 \cdots du_n  
 \\
&= 
 \sum_{n\geq 0} \frac{(-1)^n}{n!} \frac{1}{\beta}\int du_0\int_{ (0,\beta)^n} \omega^{\beta}_C (T(\tau_{\mathrm{i}u_0}A\otimes \tau_{\mathrm{i} u_1} H_I \otimes  \dots \otimes \tau_{\mathrm{i} u_n} H_I))  du_1 \cdots du_n 
\end{align*}
where the last equality holds thanks to the KMS condition, and in particular thanks to the invariance under translations and the periodicity.
Therefore, introducing the notations
\[
\tilde{H}_I(u) = \tau_{\mathrm{i} u} H_I(u)
\]
where $H_I$ could depend locally on $u$, and the notation
\[
\tilde{A} = \frac{1}{\beta} \int_0^\beta \tau_{\mathrm{i} u} A du,
\]
where now $A$ does not explicitly depend on $u$, equation \eqref{eq:thermal-correlations} can be written in the following compact form
\begin{equation}\label{eq:thermal-state:1}
\omega^{\beta I}(A) = \frac{\omega^{\beta}(T( \tilde{A} \exp_T (-\int_0^\beta \tilde{H}_I(u) du) )}{
\omega^{\beta}(T( \exp_T (-\int_0^\beta \tilde{H}_I(u) du) )}.
\end{equation}

\subsection{Formulation as a parabolic problem}\label{se:formulation-as-parabolic-problem}

In this section we analyse the abelian $*$-algebra of field observables which realises the parabolic anti-time ordering introduced above. This abstract $*$-algebra is denoted by $\mathcal{P}$. 
We shall however work again directly with a realisation of this algebra, denoted $\mathcal{P}_\beta$, given in terms of functionals over the independent {\bf field configurations} $\psi$ and $\psi^*$ in $\mathcal{C}:=C^\infty (\mathbb{R}^3\times[0,\beta])$ 
extended by periodicity to the whole $\mathbb{R}^4$.
The generators of the field algebra we are considering are denoted by
\begin{align*}
\Psi(f) &= \int_0^\beta \int_{\mathbb{R}^3} \psi(x,u) f(x,u) d^3x du,
\\
\Psi^*(f) &= \int_0^\beta \int_{\mathbb{R}^3} \psi^*(x,u) f(x,u) d^3x du
\\
|\Psi|^2(f) &= \int_0^\beta \int_{\mathbb{R}^3} \psi^*(x,u)\psi(x,u) f(x,u) d^3x du
\end{align*}
where $f\in L^2(\mathbb{R}^3\times [0,\beta], d^3x du)$
and where $\psi$ and $\psi^*$ are the independent field configurations. 

The $*$-operation is such that 
\begin{equation}\label{eq:star-parabolic}
\Psi(f)^* = \Psi^*(\overline{f}), \qquad  
\Psi^*(f)^* = \Psi(\overline{f}), 
\qquad  
|\Psi^2|(f)^* = |\Psi^2|(\overline{f}).
\end{equation}
The product is given by means of the map over sufficiently regular functionals in $\mathcal{F}$
over the field configurations $\mathcal{C}$ which contains $\mathcal{P}_\beta$. In particular, it is given in terms of the thermal propagator $\mathsf{G}_\beta$
\begin{equation}\label{eq:gamma}
e^{\gamma}:\mathcal{F}\to\mathcal{F}, \qquad \gamma := \int  d^3x du d^3x' du' \mathsf{G}_\beta(x,u;x',u')
\frac{\delta^2}{\delta \psi(x,u)\delta \psi^*(x',u')}
\end{equation}
and the product is such that 
\begin{equation}\label{eq:product-parabolic}
A \cdot_T  B = 
e^{\gamma}((e^{-\gamma}A)(e^{-\gamma}B))
\end{equation}
where the product at the right hand side is the pointwise product of functionals. 

\begin{definition}
The abelian $*$-algebra of the parabolic theory $\mathcal{P}_\beta$ is the smallest set of functionals over $\mathcal{C}$ which 
contains
\[
\{1,\Psi(f_1),\Psi^*(f_2), |\Psi|^2(f_3)\}, \qquad f_i\in C^\infty_c(\mathbb{R}^3\times[0,\beta])
\]
and which is  closed under the product introduced in \eqref{eq:product-parabolic} and the $*$-operation given in \eqref{eq:star-parabolic}.
\end{definition}

The following proposition connects the evaluation of  observables in the interacting equilibrium state with those of the auxiliary parabolic theory.
The rules we shall employ to make this connection are spelled out as the following  formal relation on the  generators of $\mathcal{A}_\beta$ and of $\mathcal{P}_\beta$ \begin{equation}\label{eq:connection-rules}
\Psi(x,u) = \tau_{\mathrm{i}u}\Phi(x) \qquad \text{and}\qquad \Psi^{*}(x,u) = \tau_{\mathrm{i}u}\Phi^*(x)
\qquad \text{and}\qquad |\Psi|^2:= \tau_{\mathrm{i}u} :|\Phi|^2:_\beta.
\end{equation}
With this rules at disposal, we have the following 
\begin{proposition}\label{pr:extension}
Let $O_0,O_1, \dots, O_n \in \mathcal{A}$, the corresponding elements of $\alpha_\beta O_0, \alpha_\beta O_1, \ldots, \alpha_\beta O_n \in \mathcal{A}_\beta$ and their evaluation on the state
\[
F = \int_{[0,\beta]^n}\omega^\beta(\mathcal{M}T( \tau_{\mathrm{i}u_1} O_1 \otimes \dots \otimes \tau_{\mathrm{i}u_n} O_n)) du_1 \cdots du_n.
\]
Consider the elements of $\mathcal{P}_\beta$ 
\[
\tilde{O}_j = \int_0^\beta O_j(u_j) du_j,   
\]
where $O_j(u_j)$ is obtained from $\tau_{iu_j}\alpha_\beta O_j$
 applying the rules \eqref{eq:connection-rules}.
It holds that 
\[
F= \left. e^{\gamma} (\prod_{j=1}^n (e^{-\gamma}\tilde{O}_j))
\right|_{\psi=\psi^*=0}.
\]
Hence, for $O_0 \in \mathcal{A}$, it holds that \eqref{eq:thermal-state:1} takes the form 
\begin{equation}\label{eq:thermal-state:2}
\omega^{\beta I}(O_0) = 
\left.
\frac{1}{\beta}\frac{e^{\gamma} \left( (e^{-\gamma }\tilde{O}_0)  \exp \left( - e^{-\gamma}  \tilde{H}_I
 \right)\right)}
{
e^{\gamma} \exp \left(- e^{-\gamma} \tilde{H}_I\right) 
}
\right|_{\Psi=\Psi^*=0}
\end{equation}
where similarly to $\tilde{O}_0$, 
\[
\tilde{H}_I = \int_0^\beta H_I(u) du
\]
and $\tilde{H}_I(u)$ is obtained from $\tau_{iu}\alpha_\beta H_I$ by means of the rules  \eqref{eq:connection-rules}.
\end{proposition}

\begin{proof}
We start from
\begin{align*}
F = \int_{ (0,\beta)^n} \omega^{\beta} (\mathcal{M}T( \tau_{\mathrm{i} u_1} O_1 \otimes  \dots \otimes \tau_{\mathrm{i} u_n} O_n))  du_1 \cdots du_n  
\end{align*}
for $O_i\in \mathcal{A}$. Recall that a realisation of $\mathcal{A}$ is given by the map $\alpha_\beta$ and that the time ordering acts
\[
T(\tau_{\mathrm{i}u_1} O_1\otimes \dots \otimes \tau_{\mathrm{i}u_n} O_n) = \tau_{\mathrm{i}u_{j_1}} O_{j_1}\otimes \dots  \otimes \tau_{\mathrm{i}u_{j_n}} O_n
\]
where $j$ is the permutation of $\{1,\dots, n\}$ which realises the time ordering
\[
u_{j_1} < \dots <u_{j_n}.
\]
Hence,   
\[
\omega^\beta (T\mathcal{M}\tau_{\mathrm{i}u_1}O_1\otimes\dots \otimes \tau_{\mathrm{i}u_n}O_n) = 
\left. \mathcal{M} \prod_{i<j}e^{D_{ij}}
\tau_{\mathrm{i}u_{j_1}}\alpha_\beta O_{j_1} \otimes \dots \otimes  \tau_{\mathrm{i}u_{j_n}}\alpha_\beta O_{j_n} \right|_{\phi=0,\phi^*=0},
\]
where $D_{ij} $ is given in \eqref{eq:Dij} in terms of the operators $B_{\pm}$ with integral kernels $\mathsf{B}^\beta_\pm$ given in \eqref{eq:Boperators}.
We recall that 
\[
\tau_{\mathrm{i}u} \Phi(f) = \Phi( e^{u K} f) , \qquad \tau_{\mathrm{i}u} \Phi^*(f) = \Phi^*(e^{-uK}f) 
\]
and these operations are well defined when the fields are evaluated on $\omega_\beta$. Indeed, when $\mathrm{ev}_0( \exp\{{D_{ij}}\} (\cdot) ) $
is considered, the operators $e^{uK} $ and $e^{-uK}$, inserted by $\tau_{iu}$, are composed with $\mathsf{B}_{\pm}^\beta$. Then, thanks to the time ordering these compositions give origin to well defined operators. See \eqref{eq:thermal-1} for further details.
We now observe 
that for $u_i<u_j$
the action of $D_{ij}$ on $\tau_{iu_i}(\alpha_{\beta}O_i)$, $\tau_{iu_j}(\alpha_{\beta}O_j)$ is equivalent to the action of 
\begin{align*}
\tilde{D}_{ij}  =&\int  d^3x d^3x'dudu' \mathsf{B}^\beta_-(x,x')e^{(u-u')K}  \frac{\delta}{\delta \psi_i(x,u)}\otimes \frac{\delta}{\delta \psi_j^*(x',u')}
\\
&+\int  d^3x d^3x'dudu' \mathsf{B}^\beta_+(x,x') e^{-(u-u')K} \frac{\delta}{\delta \psi_i^*(x,u)}\otimes \frac{\delta}{\delta \psi_j(x',u')}  
\end{align*}
when the identifications
$\tau_{iu} \varphi(x) = \psi(x,u)$ and
$\tau_{iu} \varphi^*(x) = \psi^*(x,u)$
are taken into account. 
From \eqref{eq:operators-thermals-relations} and \eqref{eq:kernel-thermal} this is equivalent to 
\begin{align*}
\tilde{D}_{ij}  =&\int  d^3x d^3x'dudu' \mathsf{G}_\beta(x,u,x',u')  \frac{\delta}{\delta \psi_i(x,u)}\otimes \frac{\delta}{\delta \psi_j^*(x',u')}
\\
&+\int  d^3x d^3x'dudu' \mathsf{G}^\dagger_\beta(x,u,x',u')  \frac{\delta}{\delta \psi_i^*(x,u)}\otimes \frac{\delta}{\delta \psi_j(x',u')}  
\end{align*}
and furthermore, in view of the definition of $\gamma$ given in \eqref{eq:gamma} and of the property $\mathsf{G}_\beta^\dagger(u)=\mathsf{G}_\beta(-u)$ explained in Appendix \ref{se:thermal-propagators-appendix}, we have for $u_1<u_2$ that
\[
\mathcal{M} (e^{D_{12}}\tau_{\mathrm{i}u_1}O_1\otimes \tau_{\mathrm{i}u_1} O_2) = 
e^{\gamma} (e^{-\gamma}(\tau_{\mathrm{i}u_1}O_1)e^{-\gamma}(\tau_{\mathrm{i}u_2}O_2)).
\]
We observe that the time ordering map $T$ is symmetric and so is the product $\cdot_T$ given in terms of $e^{\gamma}$. Hence
\[
F= \int_{ (0,\beta)^n}
e^{\gamma}
(e^{-\gamma}\tilde{O}_1(u_1))
\dots
(e^{-\gamma}\tilde{O}_n(u_n))
du_1 \cdots du_n  
\]
where $\tilde{O}_j(u_j) = \alpha_\beta \tau_{iu_j}O$.
The last statement of the proposition follows applying the obtained result to $\omega^{\beta I}(O_0)$ in the form given in \eqref{eq:thermal-state:1}.
\end{proof}

Discarding the term proportional to the renormalised chemical potential (setting for the moment $\tilde\mu_\beta=0$),
we have that $H_I$ given in \eqref{eq:HIbeta}
is constructed in terms of a non local product of Wick squares. 
Up to a completion it is thus an element of $\mathcal{A}$ and 
such remains in $\mathcal{A}_\beta$ once represented
\[
\alpha_\beta {H}_I =  \frac{1}{2}\int d^3x d^3x' {g(x)g(x')}  v(x-y) (\alpha_\beta:\!\!|\Phi|^2\!\!:_\beta(x))\star_\beta (\alpha_\beta:\!\!|\Phi|^2\!\!:_\beta(x')).
\]
We see that this representation involves a suitable product of elements of $\mathcal{P}_\beta$, even if from the point of view of 
$\mathcal{P}_\beta$ they appear to have the same parabolic time. Let us argue why it is consistent with the result of the above proposition. A direct application of the rules \ref{eq:connection-rules} to $\tau_{iu} \alpha_\beta (H_I) = \tilde{H}_I$  
gives 
\begin{align*}
\tilde{H}_I &=  \frac{1}{2}
\int 
|\Psi|^2(x,u)
v(x-x')
  |\Psi|^2(x',u') \delta(u-u')
  d\mu_g
\\&+ 
\frac{1}{2}
\int 
\Psi^*(x,u)
\mathsf{B_-^\beta}(x,x')
v(x-x')
  \Psi(x',u') \delta(u-u')
  d\mu_g
\\&+ 
\frac{1}{2}
\int 
\Psi(x,u)
\mathsf{B_+^\beta}(x,x')
v(x-x')
  \Psi^*(x',u') \delta(u-u')
  d\mu_g
\\&+ 
\frac{1}{2}
\int 
\mathsf{B_+^\beta}(x,x')
\mathsf{B_-^\beta}(x,x')
v(x-x')
\delta(u-u')d\mu_g
\end{align*}
where $\int  f d\mu_g =
\int_0^\beta du 
\int_0^\beta du'
\int_{\mathbb{R}^3}d^3x
\int_{\mathbb{R}^3}d^3x'
{g(x)g(x')} 
  f(x,u,x',u')
$
and the latter can be rewritten as
\begin{align*}
\tilde{H}_I 
&=  \lim_{\epsilon\to 0}
\frac{1}{2} e^\gamma
\int 
e^{-\gamma}|\Psi|^2(x,u)
v(x-x')
  e^{-\gamma}|\Psi|^2(x',u') \delta(u-u' +\epsilon) d\mu_g.
\end{align*}
We observe that due to the symmetries of the thermal propagator, the limit of vanishing $\epsilon$ does not depend on the direction and we can thus take it implicitly leading to
\begin{equation}\label{eq:tildeHI}
\tilde{H}_I= 
\frac{1}{2} 
\int
|\Psi|^2(x,u)
\cdot_T
|\Psi|^2(x',u') v(x-x')
\delta(u-u') 
d\mu_g
\end{equation}
Finally we recall that the limit where $g$ tends to $1$ is eventually taken.

Following the framework introduced in this section,
the correlation function of the original theory are then obtained by means of suitable limit procedure.
Namely, by considering the expectation value of products of fields whose support in the parabolic time shrinks to $0$, with the caveat that the direction in which this limit is taken matters.
We also observe that the equation of motion satisfied by the interacting theory for $\Psi$ is a non linear heat equation for a $\phi^4_3$ regularised interaction which is thus close to the analysis performed in the theory of stochastic quantization in \cite{Hairer, Gubinelli1,Gubinelli2}, see also 
\cite{DuchGubinelliRinaldi}.
Although, the difference with respect to those works, lies in the fact that, in order to get a state at finite temperature (KMS state), we require periodicity in the parabolic time.

\section{Hubbard-Stratonovich transformation}\label{se:hubbard-stratonovich}

The framework discussed in the previous section and in particular, the construction of the thermal equilibrium state for an interacting theory presented in Section \ref{se:thermal-state-interacting-theory} and its formulation as a parabolic problem as in Section \ref{se:formulation-as-parabolic-problem} for a generic $H_I$ can be utilized also when the quantum field has a non trivial background.
In this case, after decomposing the fields as 
\[
\Phi= \phi_0+\Psi, \qquad \Phi^*=\phi_0+\Psi^*
\]
where $\phi_0$ is the real valued classical background we shall specify below and $\Psi$ the fluctuations, 
we shall apply the results of the previous sections directly to $\Psi$.
In order to do it, we decomposed the Hamiltonian density of the system in contributions homogeneous in powers of $\Psi$ and $\Psi^*$. We work directly in
$\mathcal{P}_\beta$ where the product is constructed with respect to part of the quadratic contribution of the Hamiltonian density which takes the form
\[
{\tilde{\mathcal{H}}}_{20} = \Psi^* ({K}-\tilde{\mu}) \Psi + |\Psi|^2v*(\phi_0^2)
\]
where, in order to simplify the notation, we denoted by $\tilde{\mu} \equiv \tilde{\mu}_{\beta}$ the renormalised chemical potential close to \eqref{eq:ren-chem-pot}, 
\[
{K} = -\frac{\Delta}{2m}
\]
and we have the constraint on $\phi_0$
\[
v*(\phi_0^2)-\tilde{\mu}\geq 0.
\]
The decomposed Hamiltonian is
\[
\tilde{\mathcal{H}}=\tilde{\mathcal{H}}_0+\tilde{\mathcal{H}}_1+\tilde{\mathcal{H}}_2+\tilde{\mathcal{H}}_3+\tilde{\mathcal{H}}_4
\]
where
\begin{align*}
\tilde{\mathcal{H}}_0&= \phi_0({K}-\tilde\mu)\phi_0 + \frac{1}{2}\phi_0^2 v *(\phi_0^2 )    
\\
\tilde{\mathcal{H}}_1&= \phi_0({K}-\tilde\mu)(\Psi+\Psi^*) + \phi_0^2 v *(\phi_0(\Psi+\Psi^*))
\\
\tilde{\mathcal{H}}_2&= \Psi^* ({K}-\tilde\mu) \Psi + |\Psi|^2 v*(\phi_0^2) + \frac{1}{2}  \phi_0 (\Psi+\Psi^*) v* (\phi_0(\Psi+\Psi^*))
\\
\tilde{\mathcal{H}}_3&= |\Psi|^2 v* ((\Psi+\Psi^*)\phi_0)
\\
\tilde{\mathcal{H}}_4&= \frac{1}{2} |\Psi|^2 v* |\Psi|^2.
\end{align*}
The convolution involves only the equal parabolic time fields and all the components of these densities are at equal parabolic time. Furthermore, the product of the various fields is the time ordered product mentioned above. Namely,
\[
\tilde{\mu} = \mu + \frac{v(0)}{2}-q_\beta\hat{v}(0)
\]
where the $q_\beta$ vanishes in the limit $\beta\to\infty$.

We observe that the easiest way to treat this extra chemical potential renormalisation is to take it into account from the beginning, affecting in this way $\tilde{\mathcal{H}}_1$ and $\mathcal{\tilde{H}}_0$. Actually, due to background independence, taking the decomposition of the fields $\Phi$ to $\phi_0+\Psi$ before or after considering any normal ordering which
renormalises $\mu$ 
is equivalent for the final theory. 
In other words, by direct computation we can confirm that taking into account this renormalisation for $\Psi$ or 
$\Phi$ the form of the obtained $\tilde{\mathcal{H}}_1$ is equivalent. The form of $\tilde{\mathcal{H}}_0$ might change by a constant. However this constant is immaterial in the construction of the partition function because it is removed by any normalization. 
What changes is the particular form of this extra renormalisation (the value of $q_\beta$) which depends on the reference state which is used to construct the product in $\mathcal{P}_\beta$ and thus it depends also on the Hamiltonian density of the free theory. 

The background $\phi_0$ is chosen to be constant in the parabolic time and to make $\tilde{\mathcal{H}}_1$ equal to $0$. Namely, it is determined as solution of the equation 
\[
({K} -\tilde{\mu})\phi_0(x) + (v* (\phi_0^2))(x)\phi_0(x)=0.
\]
Incidentally, this coincides with a stationary configuration
of the background energy functional which is also constant in the parabolic time 
\[
\tilde{\mathcal{H}}_0=\phi_0({K} -\tilde{\mu})\phi_0 + \frac{1}{2} \phi_0^2 v* (\phi_0^2).
\]
For constant $\phi_0$, we thus have that
\begin{equation}\label{eq:phi0}
\phi_0^2 = \frac{\tilde{\mu}}{\hat{v}(0)}
\end{equation}
which is meaningful only if $\tilde\mu > 0$.
We furthermore observe that with this choice
\[
{\tilde{\mathcal{H}}}_{20} = \Psi^* K \Psi + |\Psi|^2v*(\phi_0^2)
\]
and the spectrum of the corresponding operator $K = -(2m)^{-1}\Delta$ is positive. Later, to make it strictly positive we shall add a small $\epsilon>0$ (infrared) regularization $K\to K_\epsilon=K+\epsilon$. 
This regularization will be eventually removed.

{\bf Remark:} As a first immediate consequence, a condition on the critical temperature $T_c$ above which no condensation may occur is obtained imposing finite density of the condensate ($\phi_0^2>0$). Actually, recalling the form of $q_\beta$ given in \eqref{eq:qbeta}, and utilizing the regularized free theory, it holds that $\tilde{\mu}$ is positive only if
\begin{equation*}
    \beta^{-3/2}\mathrm{Li}_{\frac{3}{2}}(e^{-\beta \epsilon}) < \left(\frac{2 \pi}{m}\right)^{\frac{3}{2}} \frac{\mu + \frac{v(0)}{2}}{\hat{v}(0)}.
\end{equation*}
This provides just a crude estimate of $T_c=1/\beta_c$. Further corrections arise when the background value of the field is renormalised to a solution of the Gross-Pitaevskii equation and consequently also the quadratic theory gets renormalised.

\subsection{Quadratic theory}
We use the Hubbard-Stratonovich transformation of the Hamiltonian densities to reduce its order with the price of in introducing an external field $A(u)$. Namely, let
\[
\tilde{\tilde{\mathcal{H}}} = 
\tilde{\tilde{\mathcal{H}}}_0 +
\tilde{\tilde{\mathcal{H}}}_1
+\tilde{\tilde{\mathcal{H}}}_2
\]
where 
\begin{align*}
\tilde{\tilde{\mathcal{H}}}_0&=
\tilde{\mathcal{H}}_0
\\
\tilde{\tilde{\mathcal{H}}}_1&= \tilde{\mathcal{H}}_1+
g \phi_0 A(\Psi+\Psi^*) = g  A\phi_0(\Psi+\Psi^*)
\\
\tilde{\tilde{\mathcal{H}}}_2&= \Psi^* ({K}-\tilde{\mu}) \Psi + |\Psi|^2v*(\phi_0^2)+ g A |\Psi|^2. 
\end{align*}
We discard $\tilde{\tilde{\mathcal{H}}}_0$ which is constant in the field $\Psi$, $\Psi^*$ and $A$ and we consider the following Hamiltonian
\[
\tilde{\tilde{\mathcal{H}}}_{A} 
=
\tilde{\tilde{\mathcal{H}}}_2+\tilde{\tilde{\mathcal{H}}}_1
=
\tilde{\tilde{\mathcal{H}}}_{20} + Q_A
\]
where
\[
\tilde{\tilde{\mathcal{H}}}_{20} = \Psi^* (\mathcal{K}-\tilde{\mu}) \Psi + |\Psi|^2v*(\phi_0^2)
,
\qquad
Q_A = g A |\Psi|^2+g  A\phi_0(\Psi+\Psi^*). 
\]
If $\phi_0^2$ is chosen as in \eqref{eq:phi0} 
\[
\tilde{\tilde{\mathcal{H}}}_{20} = \Psi^* ({K}-\tilde{\mu}) \Psi +  \hat{v}(0) \phi_0^2 |\Psi|^2
= \Psi^*{K} \Psi.
\]
Here 
\[
{K} = -\frac{\Delta}{2m}
\]
and thus the spectrum of ${K}$ is  $\sigma(K) \subset [0,\infty)$. 

To improve the regularity in the infrared regime we add a regulator in the free theory that we shall remove later. 
This corresponds to add an $\epsilon >0$ to ${K}$ in order to have it positive definite.
Hence 
\[
{K}_\epsilon= -\frac{\Delta}{2m} + \epsilon
\]
and the corresponding spectrum is $[\epsilon,\infty)$.

We see that $\phi_0$ does not appear in ${K}$ and thus $\tilde{\tilde{\mathcal{H}}}_{20}$, the free theory, does note depend on the density of the condensate. The latter, reappears in the interaction via perturbation theory through the interaction Hamiltonian $Q_A$.

Consider now
\begin{equation}\label{eq:QA}
Q_A(u) = \int_\Sigma \di^3 x  g(x) A(x,u) (|\Psi|^2(x,u) + (\Psi(x,u)+\Psi^*(x,u))\phi_0 ).
\end{equation}
We have the {\bf intermediate state}
\begin{equation}\label{eq:omegabetaA}
\omega^{\beta A}(B) := \frac{\omega^{\beta}(BU^{\phi_0}_A(\mathrm{i}\beta))}{\omega^{\beta}(U^{\phi_0}_A(\mathrm{i}\beta))}
=\sum_{{n}\geq 0} (-1)^n\int_{\beta S_n} \di^n u\;\omega_C^\beta(B;\underbrace{Q_A(u_1);\dots; Q_A(u_n)}_n) , \qquad B\in \mathcal{P}_\beta.
\end{equation}
Furthermore
\begin{equation} \label{eq:LA-UA}
U_A^{\phi_0}(\mathrm{i}\beta) =
e^\gamma e^{-e^{-\gamma} L_A},
\qquad \text{with}\qquad
L_A =\int_0^\beta \di u \int \di^3 x\;  g(x) A(x,u)  (|\Psi|^2+(\Psi+\Psi^*)\phi_0) 
\end{equation}
where  $e^{\gamma}$ is given in \eqref{eq:gamma}. The corresponding partition function or generating function
is
\begin{equation} \label{eq:ZA}
Z_A(J,J^*) = \omega^\beta( e^{\mathrm{i} \Psi(J)+\mathrm{i} \Psi^*({J^*})} U_A^{\phi_0})
\end{equation}
where $J$ and $J^*$ are compactly supported smooth functions which here can depend also on $u$. 
Notice that considering $J$ which are not constant in $u$ is necessary to be able to recover the correlation functions of the original theory, with a suitable limit procedure. 

\subsection{Removing the auxiliary field }

To recover the original theory, namely to obtain the partition function of the interacting equilibrium state $\omega^{\beta I}$, we consider the map
\[
\Gamma = \frac{1}{2} \int \di^3 x \di^3 x' \di u \di u'     v(x-x')\delta(u-u') \frac{\delta^2}{\delta A( x',u')\delta A( x,u)}
\]
which acts on functionals of the auxiliary field $A$.
We recall Proposition \ref{pr:extension} and from \eqref{eq:tildeHI} we have, as an element of  $\mathcal{P}_\beta$, that
\[
 U(\mathrm{i}\beta) = e^\gamma e^{- e^{-\gamma} \tilde{H}_I}.
\]
With $L_A$ and $U_A^{\phi_0}(\mathrm{i}\beta)$ introduced in  \eqref{eq:LA-UA} we have the following proposition.

\begin{proposition}\label{pr:GammaUA}
It holds that 
\[
U(\mathrm{i} \beta) = \left. e^{-\Gamma} U_A^{\phi_0}(\mathrm{i}\beta) \right|_{A=0}
\]
\end{proposition}
\begin{proof}
Since $L_A$ is linear in $A$ and $\Gamma$ contains two functional derivatives in $A$, it holds that $\Gamma e^{-L_A} = \langle L_A^{(1)},L_A^{(1)}\rangle_\Gamma e^{-L_A}$
where $\langle L_A^{(1)},L_A^{(1)} \rangle_\Gamma = \frac{1}{2} \Gamma(L_A)^2$. Hence
\[
e^{-\Gamma} e^{-L_A} = \sum_{n\geq 0} \frac{1}{n!}(- \langle L_A^{(1)}, L_A^{(1)} \rangle_\Gamma)^n e^{-L_A} 
= \sum_{n\geq 0} \frac{1}{n!}(- \frac{1}{2}\Gamma ( L_A^2 ))^n e^{-L_A}
= e^{- \frac{1}{2} \Gamma (L_A^2 )} e^{-L_A}.
\]
We now observe that 
\begin{align*}
\tilde{H}_I 
&=  \left\langle \frac{1}{2}|\Psi|^{2}
\otimes
|\Psi|^{2}
+
 \phi_0 |\Psi|^{2} \otimes  |\Psi|^{2}(\Psi+\Psi^*)
+ \frac{1}{2}\phi_0^2(\Psi+\Psi^*)\otimes (\Psi+\Psi^*)
,\delta v\right\rangle_{2}
=
\frac{1}{2} \Gamma (L_A^2 ).
\end{align*}
The very same results hold considering the time ordered exponentials because 
$e^{\gamma}$ commutes with $e^{-\Gamma}$.  
Hence we have
\begin{align*}
\left. e^{-\Gamma} U_A^{\phi_0}(\mathrm{i}\beta) \right|_{A=0}
=
\left. e^{-\Gamma} e^\gamma e^{- e^{-\gamma}L_A}  \right|_{A=0}  
=
e^\gamma \left. e^{-\Gamma}  e^{- e^{-\gamma}L_A}  \right|_{A=0}
=
e^\gamma   e^{- e^{-\gamma} \tilde{H}_I}  
\end{align*}
thus concluding the proof.
\end{proof}
Hence, by an application of Proposition \ref{pr:GammaUA}, see also \cite{GalandaPinamonti},
the relative partition function in the interacting equilibrium state with $J$ that could depend also on $u$ can be obtained from $Z_A(J,J^*)$ in \eqref{eq:ZA} as 
\begin{align}\label{eq:Z}
Z(J,J^*):=\omega^{\beta I}( e^{\mathrm{i} (\Psi(J) + \Psi^*(J^*))})
 & =\frac{\omega^{\beta}( 
 e^{\mathrm{i} (\Psi(J) + \Psi^*(J^*))}
 U(\mathrm{i}\beta))}{\omega^{\beta }(U(\mathrm{i}\beta))}
=\left.\frac{e^{-\Gamma} Z_A(J,J^*)}{e^{-\Gamma} Z_A(0,0)}\right|_{A=0},
\end{align}
where the time ordered product among the exponential of the fields $\Psi(J)$, $\Psi^*(J^*)$ and of $U$ is kept implicit.

\section{Combinatoric of the relative partition function and its loop vertex expansion}\label{se:partitionfunction-loopvertex}

\subsection{The relative partition function in the auxiliary theory}

We work in the parabolic theory constructed above, where we denoted by $\mathcal{P}_\beta$ the algebra of observables of that theory. Consider a generic and arbitrary $L\in \mathcal{P}_\beta$. The relative partition function, constructed out of $L$, is defined as
\[
U(L):= e^\gamma \exp e^{-\gamma} (-L)
\]
which is equivalent to
\[
U(L) = 1+\sum_{n\geq 1} \frac{(-1)^n }{n!} \mathcal{M}e^{\sum_{i<j}\tilde{D}_{ij}} \underbrace{L \otimes \dots \otimes L}_n,
\]
where
\begin{align*}
\tilde{D}_{ij} &= \int_0^\beta  du \int_0^\beta du'\int \dvol^3 x \dvol^3 x'  \mathsf{G}_\beta(x,u;x',u')
\left(
\frac{\delta}{\delta \psi_i(x,u)}\otimes \frac{\delta}{\delta \psi^*_j(x',u')}
+
\frac{\delta}{\delta \psi^*_i(x',u')}\otimes \frac{\delta}{\delta \psi_j(x,u)}
\right)
\end{align*}
and the functional derivatives $\frac{\delta }{ \delta \psi_i(x,u)}$ act on the $i$-th component of the tensor product. Finally, recall that $\mathcal{M}(A_1\otimes \dots \otimes A_n)=A_1\dots A_n$. 

Further rearrangements of the above expression give that 
\begin{align*}
U(L) 
&= 1+\sum_{n\geq 1} \frac{(-1)^n }{n!} \mathcal{M}\prod_{i<j} (e^{\tilde{D}_{ij}}-1+1) \underbrace{L \otimes \dots \otimes L}_n
\\
&= 1+\sum_{n\geq 1} \frac{(-1)^n }{n!} \mathcal{M}\sum_{G\in \mathcal{G}_n}\prod_{\{i,j\}\in E(G)} 
(e^{\tilde{D}_{ij}}-1) \underbrace{L\otimes \dots \otimes L}_n
\end{align*}
where $\mathcal{G}_n$ is the set of simple unoriented graphs with $n$ ordered vertices\footnote{In our convention, neither multiple lines nor tadpoles, called loops in the mathematical language, are allowed for the elements of $\mathcal{G}_n$. Moreover, $\mathcal{G}_n$ contains also the empty graph which is the graph with no edges.}.
In the product above, the notation $\{i,j\}\in E(G)$ assumes that $i<j$.
We can now group the elements in the sum giving origin to $U(L)$ in a different way
\begin{align*}
U(L) 
&= \sum_{n\geq 0} \frac{(-1)^n }{n!} \sum_{k=1}^n \sum_{\{I_1,\dots I_k \}\in \pi^k_n} \prod_{l=1}^k\sum_{G\in \mathcal{G}^c_{I_l}}\prod_{\{i,j\}\in E(G)} 
\mathcal{M} (e^{D_{ij}}-1) \underbrace{L \otimes \dots \otimes L}_n
\end{align*}
where $\mathcal{G}^c_{I_l}$ is the set if connected graphs among the vertices in $I_l$ and
where $\pi^k_n$ is the set of partitions of $n$ elements in $k$ non empty subsets. Hence, introducing
\[
B(1) = L
\]
and
\[
B(k) = \sum_{G\in \mathcal{G}^c_{k}} \mathcal{M} \prod_{\{i,j\}\in E(G)} 
(e^{\tilde{D}_{ij}}-1) \underbrace{L \otimes \dots \otimes L}_k
\]
where $\mathcal{G}^c_{k}$ is the subset of $\mathcal{G}_k$ formed by connected graphs, we have 
\begin{align*}
U(L) 
&= 1+\sum_{n\geq 1} \frac{(-1)^n }{n!} \sum_{k=1}^n \sum_{\{I_1,\dots I_k \}\in \pi^k_n} \prod_{l=1}^k B(|I_l|).
\end{align*}
Since $B(|I_l|)$ depends only on the number of elements in $I_l$, the sum can be rearranged considering the one over $k$ at last. In this way
\begin{align*}
U(L) 
&= 1+\sum_{k\geq 1 }\frac{1}{k!} \left(\sum_{n=1}^\infty \frac{(-1)^n }{n!} B(n)\right)^k
\\
&= \exp \left(\sum_{n=1}^\infty \frac{(-1)^n }{n!} B(n)\right),
\end{align*}
and hence
\begin{equation}\label{eq:mayer}
\log(U(L)) = - L + 
\sum_{n=2}^\infty \frac{(-1)^n }{n!} 
\sum_{G\in \mathcal{G}^c_{n}} \mathcal{M} \prod_{\{i,j\}\in E(G)} 
(e^{\tilde{D}_{ij}}-1) \underbrace{L \otimes \dots \otimes L}_n.
\end{equation}

\subsection{Quadratic theory}\label{eq:quadratic-theory}
Let us start to specialise the analysis of the previous section to a quadratic $L_A$, which is a polynomial of order two in $\psi,\psi^*$ 
\begin{equation}\label{eq:LA}
L_A = \sqrt{\lambda} \int_0^\beta\int  \left(A_1 \psi^*\psi +A_2 \psi+ A_2'\psi^* \right) \dvol^3 x \dvol u,
\end{equation}
where $A_1=A$, $A_2= \phi_0 A+\mathrm{i}J$ and $A_2'= \phi_0 A+ \mathrm{i}J^*$. 
Thus, only two kind of graphs are present in $\log(U(L_A))$. Either loops with $A_1$ in the vertices, or lines with $A_2$ and $A_2'$ at the extrema.\\

Let us compute the associated partition function as explained in the previous section. First, consider the case in which $A_2$ and $A_2'$ are set to zero and denote simply $A_1 = A$. The corresponding partition function is
\begin{align*}
D_1(A)&=\left. \sum_{n=2}^\infty \frac{(-1)^n }{n!} 
\sum_{G\in \mathcal{G}^c_{n}} \mathcal{M} \prod_{\{i,j\}\in E(G)} 
(e^{\tilde{D}_{ij}}-1) \underbrace{L_A \otimes \dots \otimes L_A}_n \right|_{A_1=A, A_2 = A_2' =0, \psi=\psi^*=0}
\\
&=2\sum_{n=2}^\infty \frac{(-1)^n \lambda^{\frac{n}{2}}}{n!} 
 (n-1)!   \Tr ((A \mathsf{G}_\beta )^{n} \chi)
\\
&= -2\sqrt{\lambda}   \Tr A \sum_{n=1}^{\infty}\int_0^1 \mathsf{G}_\beta (-\sqrt{\lambda}  sA\mathsf{G}_\beta)^n \chi   \dvol s
\\
&= -2 \Tr \int_0^{\sqrt{\lambda} } \dvol s A 
\left(\mathsf{G}_\beta^{ s A }-\mathsf{G}_\beta\right)\chi
\end{align*}
where $\chi \in C^{\infty}_0(\mathbb{R}^3)$ is an immaterial cutoff function which is $1$ on the support of $A$,
which in passing can be assumed to be of compact support, 
and it is introduced to ensure the cyclicity of the trace and we used the following conventions 
\[
(-\partial_u +{K})\mathsf{G}_\beta=\mathrm{I}, \qquad (-\partial_u+{K}+A) \mathsf{G}_\beta^A=\mathrm{I}.
\]
Thus 
\[
(-\partial_u +{K} )\mathsf{G}_\beta^A = \mathrm{I} -A \mathsf{G}_\beta^A,
\]
which gives the Dyson equation
\[
\mathsf{G}_\beta^A = \mathsf{G}_\beta -\mathsf{G}_\beta A \mathsf{G}_\beta^A
\]
and iterating this relation
\begin{align*}
\mathsf{G}_\beta^A &= \sum_{n\geq 0} \mathsf{G}_\beta (-A \mathsf{G}_\beta)^n.
\end{align*}
Then, for non vanishing $A_2, A_2'$ such that $A_2 = \mathrm{i} J$ and $A_2' = \mathrm{i} J^*$, the partition function is computed subtracting the above. Namely
\begin{align*}
&D_2(A,\mathrm{i}J, \mathrm{i}J^*)\\
&=\left. \sum_{n=2}^\infty \frac{(-1)^n}{n!} 
\sum_{G\in \mathcal{G}^c_{n}} \mathcal{M} \prod_{\{i,j\}\in E(G)} 
(e^{\tilde{D}_{ij}}-1) \underbrace{L_A \otimes \dots \otimes L_A}_n \right|_{A_1=A,A_2=\mathrm{i}J, A_2'= \mathrm{i}J^*, \psi,\psi^*=0}-D_1(A)
\\
&=\sum_{n=2}^\infty  (-1)^n 2 \lambda^{\frac{n}{2}} 
 \langle \mathrm{i} J, \mathsf{G}_\beta (A \mathsf{G}_\beta)^{n-2} \mathrm{i} J^* \rangle,
\end{align*}
where we took into account that the only possible connected graphs among $n$ vertices with $L_A$ depending on $J$ and $J^*$ have either once $J$ and $J^*$ and $n-2$ times $A$, or none $J,J^*$ and only $n$ times $A$.  
Although, the graphs which have only $A$s have been already computed in $D_1(A)$. Therefore, recalling the definition of $\mathsf{G}_\beta^{\sqrt{\lambda}A}$, we have
\begin{align*}
D_2(A,\mathrm{i}J, \mathrm{i}J^*)
&= 2 \langle {\mathrm{i}J}, \mathsf{G}_\beta^{\sqrt{\lambda}A}   \mathrm{i}J^* \rangle.
\end{align*}
Finally, let us consider $A_2 = \mathrm{i} J + \phi_0 A$ and $A_2' = \mathrm{i} J^* + \phi_0 A$. This case is achieved mapping $\psi\to \psi+\phi_0$ in the above $L_A$ and repeating the same computations setting $\psi, \psi^* = 0$ while keeping $\phi_0$ non vanishing. Namely,
\begin{equation}\label{eq:LAJJ}
L_{A,J,J^*}(\psi,\psi^*) =  \sqrt{\lambda} \int  (A \psi^*\psi^2 +\mathrm{i}J \psi+\mathrm{i}J^* \psi^* ) \dvol x.
\end{equation}
is transformed to
\[
L_{A,J,J^*}(\psi+\phi_0,\psi^*+\phi^*_0) = L_{A,J,J^*}(\phi_0,\phi^*_0) +   \sqrt{\lambda}\int  (A  \psi^*\psi +(\mathrm{i}J+A\phi^*_0)\psi + (\mathrm{i}J^*+A\phi_0)\psi^* ) \dvol x.
\]
hence the desired contribution is
\[
L_{A,J-\mathrm{i}A\phi_0^*,J^*-\mathrm{i}A\phi_0})(\psi,\psi^*)
=
L_{A,J,J^*}(\psi+\phi_0,\psi^*+\phi^*_0) - L_{A,J,J^*}(\phi_0,\phi^*_0)
\]
and thus to compute 
\[
\log(U(L_{A,J-\mathrm{i}A\phi_0^*,J^*-\mathrm{i}A\phi_0})(0,0))
= \log(U(L_{A,J,J^*})(\phi_0,\phi^*_0)) 
+ L_{A,J,J^*}(\phi_0,\phi^*_0)
\]
we can make use of the previous analysis to obtain that 
\begin{equation}\label{eq:LAJ}
\begin{aligned}
&\log(U(L_{A,J-\mathrm{i}A\phi_0^*,J^*-\mathrm{i}A\phi_0})(0,0))\\
&=\log(U(L_{A,J,J^*})(\phi_0,\phi^*_0)) 
+ L_{A,J,J^*}(\phi_0,\phi^*_0)
\\
&=   
D_1(A) + D_2(A,\mathrm{i}J+\phi^*_0 A, \mathrm{i}J^* +\phi_0 A)
\\
&= 
- {2} \Tr \int_0^{\sqrt{\lambda} } \dvol s A 
\left(\mathsf{G}_\beta^{ s A }-\mathsf{G}_\beta\right)
+
  2 \langle (\mathrm{i}{J}+\phi^*_0 A), \mathsf{G}_\beta^{\sqrt{\lambda}A}  (\mathrm{i}J^*+\phi_0 A) \rangle.
\end{aligned}
\end{equation}
In the next section, we shall evaluate the external field $A$ in the appropriate state, in order to obtain the partition (generating) function of the correlation functions of the full theory starting from this last result.

\subsection{Loop vertex expansion (LVE)}\label{se:LVE}
Recalling \eqref{eq:Z}, the analysis of Section \ref{eq:quadratic-theory}, the form of $L_{A,J,J^*}$ given in
\eqref{eq:LAJJ} and the relations given in \eqref{eq:LAJ}, the generating function of the correlation function for a real valued $\phi_0$ is 
\begin{align*}
G(J,J^*)=\frac{Z(J,J^*)}{Z(0,0)}
&=\frac{\left. e^{- \Gamma} U(L_{A,J-\mathrm{i}A\phi_0,J^*-\mathrm{i}A\phi_0})(0,0) \right|_{A=0}}
{\left. e^{- \Gamma} U(L_{A,-\mathrm{i}A\phi_0,-\mathrm{i}A\phi_0})(0,0)
 \right|_{A=0} }.   
\end{align*}
and we shall study its logarithm
 \[
 \log G(J,J^*) = \log Z(J,J^*)-\log Z(0,0).
 \]
So, in view of the results of the previous section, it remains to analyse the action of $e^{-\Gamma}$  
and how it affects $\log Z(J,J^*)$. To this end, we make use of a Mayer expansion. Recall that
\[
\Gamma= \frac{1}{2}\int_0^\beta du \int_0^\beta du'\int \dvol^3 x\dvol^3 x'  v(x-x')\delta(u-u')\frac{\delta^2 }{\delta A(x,u)\delta A(x',u')}
\]
and let us further introduce
\[
\tilde\Gamma_{ij} = \int_0^\beta du \int_0^\beta du' \int \dvol^3 x\dvol^3 x'  v(x-x')\delta(u-u') \frac{\delta }{\delta A(x,u)^i}\otimes \frac{\delta }{\delta A(x',u')^j},
\]
where $\frac{\delta}{\delta A^j}$ acts in the $j$-th component of the tensor product. Then,
\begin{align*}
Z(J,J^*) 
&= \left. e^{-\Gamma} 
e^{- {2}\Tr \int_0^{\sqrt{\lambda} } \dvol s A
\left(\mathsf{G}_\beta ^{ s A }-\mathsf{G}_\beta\right)}
e^{+2
   \langle ({\mathrm{i}}J+A\phi_0), \mathsf{G}_\beta^{\sqrt{\lambda}A}  ({\mathrm{i}}J^*+A\phi_0) \rangle
} 
\right|_{A=0}\\
&=   
\left.
e^{-\Gamma} \exp e^{\Gamma}
(-V)
\right|_{A=0}
\end{align*}
where now
\[
V := e^{-\Gamma} \left( {2} \Tr \int_0^{\sqrt{\lambda} } \dvol s A
\left(\mathsf{G}_\beta^{ s A }-\mathsf{G}_\beta\right)
-2
   \langle ({\mathrm{i}}J+A\phi_0), \mathsf{G}_\beta^{\sqrt{\lambda}A}  ({\mathrm{i}}J^*+A\phi_0) \rangle
\right).
\]
At this stage, an ordinary Mayer expansion such as the one derived in \eqref{eq:mayer} is used with suitable adaptations 
\begin{equation}\label{eq:logZ-JJ-Gc}
\log Z(J,J^*) =  -V 
 + \sum_{n=2}^\infty \frac{(-1)^n }{n!}
\sum_{G\in \mathcal{G}^c_{n}} \mathcal{M} \prod_{\{i,j\}\in E(G)} 
(e^{-\tilde{\Gamma}_{ij}}-1) \underbrace{V \otimes \dots \otimes V}_n \left. \right|_{A=0},
\end{equation}
where some appearing infrared singularities is removed considering $\log Z(J,J^*) - \log Z(0)$. The latter, following the results by Brydges and Kennedy \cite{BrydgesKennedy}, admits a tree graph expansion.

Consider the function 
\[
\tilde{\Gamma}_{ij}(s) = s\tilde{\Gamma}_{ij}
\]
which is such that  
\[
\tilde{\Gamma}_{ij}(1)=\tilde{\Gamma}_{ij}, \qquad \tilde{\Gamma}_{ij}(0)=0, \qquad \qquad\dot{\tilde{\Gamma}}_{ij}(s)=\tilde{\Gamma}_{ij}, \qquad \int_{0}^1 \dot{\tilde{\Gamma}}_{ij}(s) = \tilde{\Gamma}_{ij}.
\]
The advantage of working with the function $\tilde{\Gamma}_{ij}(s)$, is in its better regularity compared to $\tilde\Gamma_{ij}$. Then, following \cite{BrydgesKennedy} 
\begin{equation}\label{eq:logZj-BK}
\begin{aligned}
&\log Z(J,J^*)\\
&=  -V  
- \sum_{n=2}^\infty \frac{
1 }{n!} 
\sum_{T\in\mathcal{T}_n}
\mathcal{M}
\prod_{b\in E(T)}\int_{0}^1  \dvol s_b\;  \dot{\tilde{\Gamma}}_b(s_b)
\exp\left[-\sum_{k < l} \int_{s_T(k,l)}^1 \dvol s\; \dot{\tilde{\Gamma}}_{kl}(s)\right]
 \underbrace{V \otimes \dots \otimes V}_n \left. \right|_{A=0}
\end{aligned}
\end{equation}
where $\mathcal{T}_n$ are the connected rooted\footnote{We define as root of a tree graph its first vertex.} tree graphs among $n$ vertices. The product runs over the edges of $T$ whose set is denoted as $E(T)$. Finally, $s_b$ is a weight attached to the edge $b$ of the tree graph, and we used the notation
\[
s_T(k,l)=\sup\{ s_b | b \in \text{ unique path in $T$ joining $k,l$}\}.
\]
Redefining the weights as $w_b=1-s_b$ we get
\begin{equation}\label{eq:logZj}
\begin{aligned}
&\log Z(J,J^*) \\
&=   -V
-\sum_{n=2}^\infty \frac{1 }{n!} 
\sum_{T\in\mathcal{T}_n}
\mathcal{M}
\prod_{b\in E(T)}\int_{0}^1  \dvol w_b\;  \tilde{\Gamma}_b
\exp\left[-\sum_{k< l} w_T(k,l) \tilde{\Gamma}_{kl}\right]
 \underbrace{V \otimes \dots \otimes V}_n \left. \right|_{A=0},
\end{aligned}
\end{equation}
where now
\[
w_T(k,l)= \inf \{w_b | \text{ unique path in $T$ joining $k,l$}\}.
\]
For an alternative proof of this formula see \cite{AbdesselamRivasseau}. In particular, using the strategy followed by Abdesselam and Rivasseau in the same mentioned paper, $\log Z(J,J^*)$ given in \eqref{eq:logZj} can be written in yet another way which keeps implicit the evaluation of the Gaussian correlations. Namely, we define a modified Wiener measure with argument $\Upsilon$
\begin{equation}\label{eq:LVE}
\begin{aligned}
&\log Z(J,J^*)\\
&= 
-\sum_{n\geq 1} \frac{1}{n!} 
\sum_{T\in\mathcal{T}_n} \left(\prod_{b\in E(T)}\int_{0}^1 d w_b \right) \int d\mu_{\Upsilon(w_T)} \prod_{\{a,b\}\in E(T)} \Upsilon_{ab} \frac{\delta^2}{\delta A_a \delta A_b} \prod_{k=1}^n W(J,J^*,A_k)
\end{aligned}
\end{equation}
where $d\mu_{\Upsilon(w_T)}$ is the Wiener measure of covariance
\begin{equation*}
    \Upsilon_{ij}(w_T)(x_i,u_i;x_j,u_j) = w_T(i,j) \Upsilon(x_i,u_i;x_j,u_j) = w_T(i,j)\frac{1}{2}v(x_i-x_j)\delta(u_i-u_j)
\end{equation*}
associated to the stochastic fields $(\mathrm{i}A_1,\dots , \mathrm{i}A_n)$. In the above, we also introduced the notation 
\begin{equation}\label{eq:vertexW}\begin{aligned}
W &= W_L+W_{JJ^*}+W_{J\phi_0}+W_{\phi_0J^*}+W_{\phi_0\phi_0} 
\\
   &=  2\Tr \int_0^{\sqrt{\lambda} } \dvol s A
\left(\mathsf{G}_\beta^{ s A }-\mathsf{G}_\beta\right)
-2
   \langle {\mathrm{i}}J, \mathsf{G}_\beta^{\sqrt{\lambda}A}  {\mathrm{i}}J^*\rangle
-2
   \langle {\mathrm{i}}J, \mathsf{G}_\beta^{\sqrt{\lambda}A}  A\phi_0 \rangle\\
&\quad -2
   \langle \phi_0A, \mathsf{G}_\beta^{\sqrt{\lambda}A}  {\mathrm{i}}J^* \rangle
   -2
   \langle \phi_0, A\mathsf{G}_\beta^{\sqrt{\lambda}A}  A\phi_0 \rangle
\end{aligned}
\end{equation}
Finally, the sum which appears is over rooted trees $T$ with $n$ vertices corresponding to a total of $2n-2$ derivatives in the Gaussian field. 

The following lemma, regarding the above Wiener measure, is of pivotal relevance for the whole later discussion of convergence.
\begin{lemma}\label{le:positive-matrices}
Consider a generic tree $T$ among $n \in \mathbb{N}$ vertices. Assign to every edge $e=(i,j)\in E(T)$ of the tree a weight $w_{e}\in[0,1]$. 
Consider the $n\times n$ positive semidefinite matrix $w_T$, whose entries are
\[
w_T(i,j) = \inf\{ w_e | e\in 
\mathrm{path}_T(i,j) , i\neq j\} \qquad w_T(i,i)=1
\]
where $\mathrm{path}_T(i,j)$ is the unique path in $T$ which joins $i,j$ and the infimum is over all edges $e$ in the path. Let $p$ be the integral kernel of a positive distribution on a suitable space $S$. Namely,
$\langle f,p f\rangle \geq 0$ for every $f \in S$. Consider
\[
\Upsilon_{ij}(x,y) = w_T(i,j) p(x,y).
\]
Then, it holds that 
\[
\sum_{i,j\in\{1,\dots, n\}}
\langle f_i, \Upsilon_{ij} f_j\rangle \geq 0
\]
for every $(f_1,\dots ,f_n)\in S^n$.
\end{lemma}
\begin{proof}
We start ordering the distinct values $b_j \in \{w_e|e\in E(T)\}$, as
\[
0\leq b_1\leq b_2 \leq  \dots \leq b_n\leq 1.
\]
For each $b_m$ consider the forest $\mathcal{F}_m(T)$, obtained by the (disjoint) trees obtained from $T$, keeping only the edges $e\in E(T)$ such that $w_e \geq b_m$. 
Denote by $T_m^{(1)},T_m^{(2)}, \dots $ the connected components in $\mathcal{F}_m(T)$. Let us define the vector $\chi_{T_m^{(r)}} \in \mathbb{R}^n$, so that $\chi_{T_m^{(r)}}^j=1$ if $j$ is a vertex in the tree $T_m^{(r)}$, namely $j \in V(T_m^{(r)})$, and let $\chi_{T_m^{(r)}}^j = 0$ otherwise. The $n\times n$ matrix $\chi_{T_m^{(r)}}\chi_{T_m^{(r)}}^t$ is positive semidefinite by construction. Notice that the matrix $w_T$ admits the following decomposition
\[
w_T = \sum_{m>0} (b_m-b_{m-1}) \chi_{T_m^{(r)}}\chi_{T_m^{(r)}}^t
\]
were $b_0=0$. Indeed, the coefficients $b_m-b_{m-1}$ are positive, and hence $w_T$ is a sum of positive semidefinite matrices. Now, since $p$ is positive, for every $(f_1,\dots, f_n)$ the matrix
\[
S_{ij} = \langle f_i ,  p f_j \rangle
\]
is positive semidefinite. In fact, for every $v\in \mathbb{R}^n$, considering $\psi = \sum_i v^if_i$ we have that
\[
v^t S v = \langle \psi,p \psi\rangle \geq 0.
\]
To conclude the proof we observe that
\[
\sum_{i,j\in\{1,\dots, n\}}
\langle f_i, \Upsilon_{ij} f_j\rangle =\Tr (w_T S) \geq 0
\]
as the trace of a product of two positive semidefinite matrices is positive. 
\end{proof}

\begin{corollary}\label{le=estimate1}
Consider a generic spanning tree of $n \in \mathbb{N}$ vertices $T\in \mathcal{T}_n$. Assign to every edge $b\in E(T)$ a weight $w_b\in [0,1]$.
Consider $n$ generalised functions $f^k$. It holds that 
\[
\int d\mu_{\Upsilon(w_T)}(A_1,\dots , A_n) \prod_k e^{\langle f^k, A_k \rangle} 
=
e^{-\sum_{i,j} w_T(ij) \langle f^i, \Upsilon f^j\rangle}
=
e^{-\sum_{i,j} w_T(ij) \frac{1}{2}\langle f^i, v\otimes\delta  f^j\rangle} \leq 1
\]
\end{corollary}
\begin{proof}
The proof of this corollary follows from an application of Lemma \ref{le:positive-matrices} to the argument of the exponential, with $p=v\otimes\delta$
\end{proof}

Before getting into the explicit evaluation of convergence of the LVE, let us briefly explain how the result of the above corollary is going to be used. Consider for a generic generalised function $h$ the following functional derivative
\[
\langle h,\frac{\delta}{\delta A} \rangle \mathsf{G}_\beta^A = 
 -\mathsf{G}_\beta^A  h \mathsf{G}_\beta^A.
\]
This result leads to the following neat expressions for the functional derivatives of the arguments of the above LVE
\[
\langle h,\frac{\delta}{\delta A} \rangle W_{JJ^*} = 2\sqrt{\lambda} 
\langle \mathrm{i}J,\mathsf{G}_\beta^{\sqrt{\lambda} A} h \mathsf{G}_\beta^{\sqrt{\lambda} A} \mathrm{i}J^*\rangle
\]
\begin{align*}
\langle h,\frac{\delta}{\delta A} \rangle W_L 
&= {2}
\Tr \int_0^{\sqrt{\lambda} } \dvol s h
\left(\mathsf{G}_\beta^{ s A }-\mathsf{G}_\beta\right)
- {2}
\Tr \int_0^{\sqrt{\lambda} } \dvol s  s A  \mathsf{G}_\beta^{ s A }h\mathsf{G}_\beta^{ s A }
\\
&= {2} \Tr \int_0^{\sqrt{\lambda} } \dvol s h \sum_{n \geq 1} s^n
\left( \mathsf{G}_{\beta} (-A \mathsf{G}_{\beta})^n \right) +{2}  \Tr \int_0^{\sqrt{\lambda} } \dvol s h \sum_{n \geq 1} s^n n
\left( \mathsf{G}_{\beta} (-A \mathsf{G}_{\beta})^n \right) \\
&= {2} \sqrt{\lambda} \Tr \left( h \left(\mathsf{G}_\beta^{ \sqrt{\lambda} A }-\mathsf{G}_\beta\right) \right)
\end{align*}
where we used the cyclic property of the trace. The nice feature of the above expressions is that the dependence upon $A$ is just via the propagators. 
In particular, this will always occur in the LVE, as at least one $\frac{\delta}{\delta A_i}$ hits every $W$. Therefore, the only norms we need to control are
\[
\|\mathsf{G}^A_\beta\|_\Upsilon = \left|\int d\mu_{\Upsilon} \mathsf{G}_\beta^A \right|\leq  \mathsf{G}_\beta  , 
\]
by the result of the above Corollary. Hence, an upper bound for 
$\log(Z(J,J^*))$ takes into account only trees contributions without decorations coming from the presence of the Wiener measure $\mu_{\Upsilon}$. 
Furthermore, in order to avoid infrared divergences, we study the normalised generating function $\log Z(J,J^*) - \log Z(0)$. Hence, to be more precise, we consider a LVE with
\begin{gather*}    
\prod_k W(J,J^*,A_k) - \prod_{k'} W(0,A_k') =\\
\prod_k (W_L(A_k)+W_{JJ^*}(A_k)+W_{J\phi_0}(A_k)+W_{\phi_0J^*}(A_k)+W_{\phi_0\phi_0}(A_k)) - \prod_{k'} (W_L(A_k')+W_{\phi_0\phi_0}(A_k)).
\end{gather*} 
As a consequence, after expanding the products over the sums, each term has a factor which contains $J$ or $J^*$. As a consequence of these last remarks, in the next section, to get an upper bound and thus a convergence estimate, it is enough to study only the trees.

\section{Convergence of the LVE}
\label{se:convergence}

Throughout the section, we use the following norm
\begin{equation}\label{eq:normJK}
\|J\|_{K,1}:= \int d^3p\frac{1}{\hat{{K}}(p)}|\hat{J}(p)|,
\end{equation}
in order to bound the terms which appear in the generating functions. We actually have the following theorem which extends previous results in \cite{GalandaPinamonti}

\begin{theorem}\label{thm:converge}
Consider a fixed inverse temperature $\beta \in (0,\infty)$ and let $K_\epsilon$, with $\epsilon>0$. Then, if
    $\|v\|_1$ is sufficiently small, $\|J\|_{K,1}\leq 1$, $\|J^*\|_{K,1}\leq 1$ and $\|\phi_0\|_\infty\leq 1$
    in such a way that
\[
 \frac{r}{4}=\left(\lambda\frac{\|v\|_1}{\epsilon} \right) \left( \frac{C}{\sqrt{\beta}}+\left(\frac{\|\phi_0\|_\infty}{\sqrt{\lambda}}+\frac{\|J\|_{K,1}+\|J^*\|_{K,1}}{{2}}\right)^2 \right)< \frac{1}{16 },
\]
where $C=2(2\pi m)^{\frac{3}{2}}$ is a positive constant,
 the series in \eqref{eq:LVE} defining
\[
\log Z(J,J^*) - \log Z(0,0) 
\]
is absolutely convergent. 
\end{theorem}
\begin{rem}
Notice that the presence of $\|J\|_{K,1}$ and $\|J^*\|_{K,1}$ in the above $r$ is immaterial when we consider the truncated correlation functions obtained from $\log(Z(J,J^*))-\log(Z(0))$. In other words, $J,J^*$ can always be rescaled to meet the required estimate.
\end{rem}
\begin{proof}
To prove this theorem we recall the loop vertex expansion in \eqref{eq:LVE}
\[
\log Z(J,J^*) = 
-\sum_{n\geq 1} \frac{1}{n!} 
\sum_{T\in\mathcal{T}_n} \left(\prod_{b\in E(T)}\int_{0}^1 d w_b \right) \int d\mu_{\Upsilon(w_T)} \prod_{\{a,b\}\in E(T)} \Upsilon_{ab} \frac{\delta^2}{\delta A_a \delta A_b} \prod_{k=1}^n W(J,J^*,A_k,\phi_0)
\]
where now 
\[
W(J,J^*,A) = W_L + W_{JJ^*} + W_{J\phi_0}+ W_{\phi_0J^*}+ W_{\phi_0\phi_0}.
\]
For a generic $n$, we consider the contribution of a generic tree $T\in \mathcal{T}_n$
\[
\tilde{X}_T=\left(\prod_{b\in E(T)}\int_{0}^1 d w_b \right) \int d\mu_{\Upsilon(w_T)} \prod_{\{a,b\}\in E(T)} \Upsilon_{ab} \frac{\delta^2}{\delta A_a \delta A_b} \prod_{k=1}^n W(J,J^*,A_k,\phi_0).
\]
In analysing the difference
$\log Z(J,J^*) - \log Z(0,0)$
we get a contribution corresponding to the graph $T$ which is 
\[
{X}_T=\left(\prod_{b\in E(T)}\int_{0}^1 d w_b \right) \int d\mu_{\Upsilon(w_T)} \prod_{\{a,b\}\in E(T)} \Upsilon_{ab} \frac{\delta^2}{\delta A_a \delta A_b} \left(
\prod_{k=1}^n W(J,J^*,A_k,\phi_0)
-
\prod_{k=1}^n W(0,0,A_k,\phi_0)\right).
\]
Notice furthermore, that for at least one $k$ in 
\[
\prod_{k=1}^n W(J,J^*,A_k,\phi_0)-
\prod_{k=1}^n W(0,0,A_k,\phi_0)
=\sum_{\{l_1,\dots,l_n\}\in\{L,JJ^*,J\phi_0,\phi_0J^*,\phi_0\phi_0\}^n} \prod_{k=1}^n W_{l_k},
\]
there is at least one $l_k$ which is in $\{JJ^*,J\phi_0,\phi_0J^*\}$. Indeed, the term involving only $\{ L, \phi_0\phi_0\}$ is cancelled in the difference. Consider then the corresponding vertex as the new root of the tree $T$. To bound $X_T$, we use the result of Corollary \ref{le=estimate1} to estimate the contribution due to the evaluation of $A$ on the Gaussian state and we observe that it just corresponds to the the evaluation $A=0$.\\

From here on, for every element of the sum over $l_k$ we have relabelled the vertices of the tree in such a way that the root is the vertex in $\{JJ^*,J\phi_0,\phi_0J^*\}$. Then, for each tree, we cut its leafs (the other vertices in the tree), both when they coincide to loops or legs, bounding them with the appropriate $|R^{(e(k)-1)}_{l_k}|$, where we are using the notation and the result in Lemma \ref{le:loops-legs}. We insert $\|v\|_1\delta(u)$ at the place of the edges in the remaining subtree, for which leafs have been catted in the previous step. We repeat the cutting process till we end up with the root. However, the possible vertices in the root are $W_{J\phi_0}$, $W_{\phi_0J^*}$ or $W_{JJ^*}$ with the appropriate insertions of $\|v\|_1 \delta (u)$ at the place of the original edges. Hence, the root is again bounded by the appropriate $C|R^{e(1),0}_{l_1}|$.
In all these terms both a factor $\beta$ and a factor proportional either to $\|J\|_1$, $\|J\|_{K,1}$, $\|J^*\|_1$ or $\|J^*\|_{K,1}$ appears, due to the presence of a $J$ or a $J^*$ in the root-vertex. Notice that, in the root, we have $|R^{e(1)}_{l_1}|$ instead of $|R^{(e(k)-1)}_{l_k}|$ because, by definition, there are no remaining edges to attach the root somewhere else. 

We notice that the Gaussian measure $d\mu_{\Upsilon(w_T)}$ in $X_T$ is used in the estimate of the various leafs contributions (see e.g. \eqref{eq: conAsenzaA}). Instead, the measure over the various weights $w_b$ assigned to the edges in $X_T$ is normalized to $1$, and hence it can be bounded by $1$. Altogether, we end up with the estimate for $X_T$
\[
|X_T| \leq 
\sum_{\{l_1,\dots,l_n\}\in\{L,JJ^*,J\phi_0,\phi_0J^*,\phi_0\phi_0\}^n}  |R^{e(1)}_{l_1}|\prod_{k=2}^n |R^{e(k)-1}_{l_k}|,
\]
where $e(k)-1$ for $k>1$ and $e(1)$ denotes the number of insertions of $\|v\|_1$ resulting from the cutting process. Here, $e(k)$ is the 
number of edges in $T$ which contains the vertex $k$. 

According to Lemma \ref{le:loops-legs} we can bound every vertex $k>1$ 
\[
\sum_{l_k\in \{L,JJ^*,J\phi_0,\phi_0J^*,\phi_0\phi_0\}} |R^{e(k)-1}_{l_k}|\leq  2 \sqrt{\lambda}
\left(c\frac{\|v\|_1}{\epsilon} \right)^{e(k)-1} (e(k)-1)! \left( \frac{C}{\sqrt{\beta}}+\left(\frac{|\phi_0|}{\sqrt{\lambda}}+\frac{\|J\|_{K,1}+\|J^*\|_{K,1}}{{2}}\right)^2 \right),
\]
for some positive constants ${c} =2 \sqrt{\lambda}$ and we recall that $C< 2(2\pi m)^{\frac{3}{2}}$, while for $k=1$
\begin{align*}
\sum_{l_k\in \{L,JJ^*,J\phi_0,\phi_0J^*,\phi_0\phi_0\}^n} |R^{e(1),0}_{l_k}|\leq & 2 \beta (\|J\|_1+\|J\|_{K,1}+\|J^*\|_1+\|J^* \|_{K,1})
\left(c\frac{\|v\|_1}{\epsilon} \right)^{e(1)} (e(1))! \\
& \times \left( \frac{C}{\sqrt{\beta}}+\left(\frac{|\phi_0|}{\sqrt{\lambda}}
+\frac{\|J\|_{K,1}+\|J^*\|_{K,1}}{{2}}\right)^2 \right) 
.
\end{align*}
Hence, switching the order of sum and product as we have a finite number of terms
\begin{align*}
|X_T| 
\leq &
\sum_{\{l_1,\dots,l_n\}\in\{L,JJ^*,J\phi_0,\phi_0J^*,\phi_0\phi_0\}^n}  |R^{e(1)}_{l_1}|\prod_{k=2}^n |R^{e(k)-1}_{l_k}|
\\
&\leq \beta  2^n (\|J\|_{1}+\|J\|_{K,1}+\|J^*\|_{1}+\|J^*\|_{K,1})
 \sqrt{\lambda}^{n-1}
\left(c\frac{\|v\|_1}{\epsilon} \right)^{n-1}\\
& \times \left( \frac{C}{\sqrt{\beta}}+\left(\frac{|\phi_0|}{\sqrt{\lambda}}+\frac{\|J\|_{K,1}+\|J^*\|_{K,1}}{{2}}\right)^2 \right)^{n}  e(1)p_T,
\end{align*}
where we used the fact that $1+\sum_{k=1}^{n} (e(k)-1)=n-1$  and where
\[
p_T=\prod_{k=1}^n(e(k)-1)! 
\]
Considering now the contributions of all trees we get
\begin{equation}\label{eq:power-series}
\begin{aligned}
    |\log Z(J,J^*)-\log Z(0,0)|&\leq |V_1|+ \beta  (\|J\|_{1}+\|J\|_{K,1}+\|J^*\|_{1}+\|J^*\|_{K,1})\left(2\lambda \frac{\|v\|_1}{\epsilon}\right)^{-1}
    \\
    &\times \sum_{n\geq 2} 2^n \left(2\lambda\frac{\|v\|_1}{\epsilon}\right)^n {\left( \frac{C}{\sqrt{\beta}}+\left(\frac{|\phi_0|}{\sqrt{\lambda}}+\frac{\|J\|_{K,1}+\|J^*\|_{K,1}}{{2}}\right)^2 \right)^{n}} \sum_{T\in \mathcal{T}_n}e(1)p_T
\end{aligned}
\end{equation}
where $V_1$ is the action of $d\mu_{\Upsilon}$ on the root of the tree. Thus, by the estimates in Lemma \ref{le:EGF-convergence}, if 
\[
r= \left(4\lambda\frac{\|v\|_1}{\epsilon} \right) \left( \frac{C}{\sqrt{\beta}}+\left(\frac{|\phi_0|}{\sqrt{\lambda}}+\frac{\|J\|_{1,K}+\|J^*\|_{1,K}}{2}\right)^2 \right)< \frac{1}{4 },
\]
the series defining $\log Z(J,J^*)-\log Z(0,0)$ is absolutely convergent and 
\begin{align*}
|\log Z(J,J^*)-\log Z(0)| 
&\leq |V_1|+ \beta(\|J\|_1+\|J\|_{K,1}+\|J^*\|_1+\|J^*\|_{K,1})\left(2\lambda\frac{\|v\|_1}{\epsilon}\right)^{-1}
(E(r)-r)
\\
&\leq |V_1|+ 
\beta(\|J\|_1+\|J\|_{K,1}+\|J^*\|_1+\|J^*\|_{K,1})
\left(2\lambda\frac{\|v\|_1}{\epsilon}\right)^{-1}
\frac{1-2r-\sqrt{1-4 r}}{2}.
\end{align*}
\end{proof}

The a priori estimates of loop and leg vertices, as well as the estimate on the sum over trees used in the above proof, are shown to hold in the next lemmata. 

\begin{lemma}\label{le:loops-legs}
Let $K_\epsilon = -\Delta/2m +\epsilon$ with $\epsilon > 0$. Consider the loop vertices
\[
W_L = {2}\Tr \int_0^{\sqrt{\lambda} } \dvol s A
\left(\mathsf{G}_\beta^{ s A }-\mathsf{G}_\beta\right)
\]
and the edge vertices 
\begin{align*}
W_{JJ^*} &= {-2}
   \langle {\mathrm{i}}J, \mathsf{G}_\beta^{\sqrt{\lambda}A}  {\mathrm{i}}J^* \rangle \qquad 
W_{J\phi_0} = {-2}
   \langle {\mathrm{i}}J, \mathsf{G}_\beta^{\sqrt{\lambda}A}  A \phi_0 \rangle\\
W_{\phi_0 J^*} &= {-2}
\langle
   \phi_0 A, \mathsf{G}_\beta^{\sqrt{\lambda}A}  {\mathrm{i}}J^* \rangle
\qquad
W_{\phi_0\phi_0} = {-2}
\langle
   \phi_0 A, \mathsf{G}_\beta^{\sqrt{\lambda}A}  A \phi_0 \rangle.
\end{align*}
Their functional derivatives with insertions satisfy the following
\begin{enumerate}[label=(\roman*)]
\item
The {\bf loop vertex} $W_L$ with $n\geq 0 $ insertions of $\|v\|_1$ and base point $x$ 
\[
  R^n_L=\langle \delta_x \otimes \underbrace{\|v\|_1\otimes \dots \otimes  \|v\|_1}_{n}, \frac{\delta^{n+1}}{\delta A \dots \delta{A}} \rangle W_L 
\]
satisfies the estimate
\begin{equation}
    |R^n_L| \leq \frac{2 C (\sqrt{\lambda})^{n+1} n!}{\sqrt{\beta}} \left(\frac{\|v\|_1}{\epsilon}\right)^n
\end{equation}
where $C$ is a positive constant such that $C<2(2\pi m)^{\frac{3}{2}}$. 
\item
The {\bf leg vertex} $W_{JJ^*}$
with $n\geq 0 $ insertions of $\|v\|_1$ and with base point $x$ 
\begin{align*}
R_{JJ^*}^n&=
\langle \delta_{x,u} \otimes \underbrace{\|v\|_1\otimes \dots \otimes  \|v\|_1}_{n}, \frac{\delta^{n+1}}{\delta A \dots \delta{A}} \rangle W_{JJ^*} 
\end{align*}
satisfies the bound
\begin{equation}\label{eq:TJJ}
\begin{aligned}
    |R_{JJ^*}^n| &\leq 2
  \sqrt{\lambda}^{n+1} 
  \left(
  \frac{\|v\|_1}{\epsilon}
  \right)^n 2^n n! 
   \|J\|_{1,K} \|J^*\|_{1,K}
\end{aligned}
\end{equation}
The {\bf leg vertex} $W_{JJ^*}$
with $n\geq 0 $ insertions of $\|v\|_1$ and without base points 
\begin{align*}
R_{JJ^*}^{0,n}&=
\langle  \underbrace{\|v\|_1\otimes \dots \otimes  \|v\|_1}_{n}, \frac{\delta^{n}}{\delta A \dots \delta{A}} \rangle W_{JJ^*} 
\end{align*}
satisfies the bound
\begin{equation}\label{eq:T0JJ}
\begin{aligned}
    |R_{JJ^*}^{0,n}| &\leq 
  2 \beta  \sqrt{\lambda}^{n} 
  \left(
  \frac{\|v\|_1}{\epsilon}
  \right)^n  n! 
   \frac{\|J^*\|_1 \|J\|_{1,K}+\|J\|_1 \|J^*\|_{1,K}}{2} 
\end{aligned}
\end{equation}
\item The {\bf leg vertex} $W_{J\phi_0}$  with $n \geq 0$ insertions of $\|v\|_1$ and with base point $x$ 
\begin{align*}
R_{J\phi_0}^n&=
\langle \delta_{x,u} \otimes \underbrace{\|v\|_1\otimes \dots \otimes  \|v\|_1}_{n}, \frac{\delta^{n+1}}{\delta A \dots \delta{A}} \rangle W_{J\phi_0} 
\end{align*}
satisfies the bound
\begin{equation}\label{eq:TJphi0}
\begin{aligned}
    |R_{J\phi_0}^n| 
  & \leq 2
  \sqrt{\lambda}^{n} 
  \left(
  \frac{\|v\|_1}{\epsilon}
  \right)^n 2^n n! 
  \|J\|_{K,1}
 \|\phi_0\|_\infty.
\end{aligned}
\end{equation}
The {\bf leg vertex} $W_{J\phi_0}$  with $n \geq 0$ insertions of $\|v\|_1$ and without  base points
\begin{align*}
R_{J\phi_0}^{0,n}&=
\langle  \underbrace{\|v\|_1\otimes \dots \otimes  \|v\|_1}_{n}, \frac{\delta^{n}}{\delta A \dots \delta{A}} \rangle W_{J\phi_0} 
\end{align*}
satisfies the bound
\begin{equation}\label{eq:T0Jphi0}
\begin{aligned}
    |R_{J\phi_0}^{0,n}| 
  & \leq 2
  \beta \sqrt{\lambda}^{n-1} 
  \left(
  \frac{\|v\|_1}{\epsilon}
  \right)^n n! 
  \|J\|_{K,1}
 \|\phi_0\|_\infty.
\end{aligned}
\end{equation}
\item 
The {\bf leg vertex} $W_{\phi_0 J^*}$ with $n \geq 0$ insertions of $\|v\|_1$ and with base point $x$, denoted by $R^n_{\phi_0 J^*}$, satisfies an analogous bound
\begin{equation}\label{eq:Tphi0J}
\begin{aligned}
    |R_{\phi_0J^*}^n| 
  & \leq 2
  \sqrt{\lambda}^{n} 
  \left(
  \frac{\|v\|_1}{\epsilon}
  \right)^n 2^n n! 
  \|J^*\|_{K,1}
 \|\phi_0\|_\infty.
\end{aligned}
\end{equation}
The {\bf leg vertex} $W_{\phi_0 J^*}$ with $n \geq 0$ insertions of $\|v\|_1$ and without base point, denoted by $R_{\phi_0 J^*}^{0,n}$, satisfies a similar bound
\begin{equation}\label{eq:T0phi0J}
\begin{aligned}
    |R_{\phi_0J^*}^{0,n}| 
  & \leq 2
  \sqrt{\lambda}^{n-1} 
  \left(
  \frac{\|v\|_1}{\epsilon}
  \right)^n n! 
  \|J^*\|_{K,1}
 \|\phi_0\|_\infty.
\end{aligned}
\end{equation}
\item The {\bf leg vertex} $W_{\phi_0\phi_0}$  with $n \geq 0$ insertions of $\|v\|_1$ and with a base point $x$, denoted by $R^n_{\phi_0\phi_0}$, satisfies a similar bound
\begin{equation}\label{eq:Tphi0phi0}
\begin{aligned}
    |R_{\phi_0 \phi_0}^n| 
  & \leq 2
  \sqrt{\lambda}^{n-1} 
  \left(
  \frac{\|v\|_1}{\epsilon}
  \right)^n 2^n n! 
 \|\phi_0\|_\infty^2.
\end{aligned}
\end{equation}
\end{enumerate}
\end{lemma}
\begin{proof}
Let us start considering a loop vertex with $n\geq 1 $ insertion of $\|v\|_1$, base point $x$ and we compute
\begin{align*}
    R^n_L &= \langle \delta_x \otimes \underbrace{\|v\|_1\otimes \dots \otimes  \|v\|_1}_{n}, \frac{\delta^{n+1}}{\delta A \dots \delta{A}} \rangle W_L\\
    &= (-1)^n n! 2 (\sqrt{\lambda})^{n+1} \lim_{y\to x}\left((\mathsf{G}_\beta^{\sqrt{\lambda}A} \|v\|_1)^n\mathsf{G}_\beta^{\sqrt{\lambda}A} \right)(x,y).
\end{align*}
Then, after having computed all the functional derivatives, and evaluated on the Gaussian state of covariance $\Upsilon$ the fields $A$, we estimate
\begin{equation}
\label{eq: conAsenzaA}
    \int d\mu_{\Upsilon(w_T)} |(\mathsf{G}_\beta^{\sqrt{\lambda}A} \|v\|_1)^n\mathsf{G}_\beta^{\sqrt{\lambda}A}(x,y)| \leq 
(\mathsf{G}_\beta \|v\|_1)^n\mathsf{G}_\beta(x,y),
\end{equation}
which is nothing but an application of the Feynman Kac formula, see \cite{GalandaPinamonti} or Corollary \ref{le=estimate1}. 
Operating in this way  we have by direct inspection
\begin{align*}
    |R^n_L| &\leq 2 n! (\sqrt{\lambda})^{n+1} \|v\|_1^n \beta^{n+1} \int \frac{e^{-\beta\hat{K}_{\epsilon}(p)}}{(1-e^{-\beta\hat{K}_{\epsilon}(p)})^{n+1}} d^3p \\
    &\leq 2 n! (\sqrt{\lambda})^{n+1} \|v\|_1^n \frac{\beta^{n+1}}{\beta^n \epsilon^n} \int \frac{e^{-\beta\hat{K}_{\epsilon}(p)}}{(1-e^{-\beta\hat{K}_{\epsilon}(p)})} d^3p \\
  &= \frac{ 2 C (\sqrt{\lambda})^{n+1} n!}{\sqrt{\beta}} \left(\frac{\|v\|_1}{\epsilon}\right)^n
\end{align*}
where $\beta^{n+1}$ arises from the integration in $(u_1,\dots u_n)$ on $[0,\beta]^{n+1}$, the factor $e^{-\beta K_{\epsilon}}$ appears because of the structure of the loop, and we used that $|(1- e^{-\beta \hat{K}_{\epsilon}})^{-1} | \leq (\beta \epsilon)^{-1}$. Finally, the constant $C$ is defined as
\[
  C = \beta^{3/2}\int {d^3p}
  \frac{e^{-\beta \hat{K}_{\epsilon}(p)}}
  {\beta\hat{K}_{\epsilon}} \leq (2\pi m)^{\frac{3}{2}}
\]
where the bound is optimal if $\epsilon=0$ and thus $K= -\Delta/2m$. The loop vertex without insertions and a base point is bounded using a similar estimate as above with $n=0$.\\

Let us consider the remaining leg vertices with $n$ insertions. We start with the $JJ^*$ term
\begin{align*}
R_{JJ^*}^n &= \langle \delta_{x,u} \otimes \underbrace{\|v\|_1\otimes \dots \otimes  \|v\|_1}_{n}, \frac{\delta^{n+1}}{\delta A \dots \delta{A}} \rangle W_{JJ^*} 
\\
&= 2\sqrt{\lambda}^{n+1}
  \int d^3y_1du_1 \int d^3y_2 du_2J(y_1)J^*(y_2)  
  \sum_{k=0}^{n} k! (n-k)!\\
  &\times(\mathsf{G}_\beta^{\sqrt{\lambda}A} \|v\|_1)^k\mathsf{G}_\beta^{\sqrt{\lambda}A}(y_1,u_1;x,u)
  (\mathsf{G}_\beta^{\sqrt{\lambda}A} \|v\|_1)^{n-k}\mathsf{G}_\beta^{\sqrt{\lambda}A}
  (x,u;y_2,u_2).
\end{align*}
Then, using the estimate in \ref{eq: conAsenzaA} we get
\begin{align*}
    |R_{JJ^*}^n| &\leq 2 \sqrt{\lambda}^{n+1}
  \int d^3y_1 du_1 \int d^3y_2 du_2 J(y_1)J^*(y_2)  
  \sum_{k=0}^{n} k! (n-k)!\\
  &\times(\mathsf{G}_\beta \|v\|_1)^k\mathsf{G}_\beta(y_1,u_1;x,u)
  (\mathsf{G}_\beta \|v\|_1)^{n-k}\mathsf{G}_\beta
  (x,u;y_2,u_2)\\
  &= 2 \sqrt{\lambda}^{n+1} \| v \|_1^n \sum_{k=0}^{n} k! (n-k)! \left( \int d^3y_1du_1 J(y_1) \mathsf{G}_\beta^{k+1}(y_1,u_1;x,u) \right)\\
  &\times \left(\int d^3y_2 du_2J^*(y_2) \mathsf{G}_\beta^{n-k + 1}  (x,u;y_2,u_2)\right),
\end{align*}
from which the estimate uniform in $x$ and $u$ given in \eqref{eq:TJJ} is obtained
\[
\begin{aligned}
    |R_{JJ^*}^n| & \leq 2
   \sqrt{\lambda}^{n+1} \|v\|_1^n
\sum_{k=0}^n
k! (n-k)!
\left| \int d^3y_1\frac{e^{-(\beta-u){K}_{\epsilon}}}{K_{\epsilon}^{k+1}}
 (y_1,x)J(y_1) \right|
 \left|
 \int d^3y_2
 \frac{e^{-u{K}_{\epsilon}}}{K_{\epsilon}^{n-k+1}}
 (x,y_2)J^*(y_2) \right|
\\
  & \leq 2
  \sqrt{\lambda}^{n+1} 
  \left(
  \frac{\|v\|_1}{\epsilon}
  \right)^n 2^n n!
   \|J\|_{K,1}\|J^*\|_{K,1}.
\end{aligned}
\]
The $JJ^*$ vertex without base points and with $n$ insertions of $\|v\|_1$ 
\begin{align*}
R_{JJ^*}^{0,n} &= \langle \underbrace{\|v\|_1\otimes \dots \otimes  \|v\|_1}_{n}, \frac{\delta^{n}}{\delta A \dots \delta{A}} \rangle W_{JJ^*} 
\end{align*}
satisfies the bound
\begin{align*}
    |R_{JJ^*}^{0,n}| &\leq 2 \sqrt{\lambda}^{n}
  \int d^3y_1du_1 \int d^3y_2 du_2J(y_1)J^*(y_2)  
   n! 
  (
  (\mathsf{G}_\beta \|v\|_1)^{{n}}\mathsf{G}_\beta)
  (y_1,u_1;y_2,u_2)\\
  &= 2 \sqrt{\lambda}^{n} \| v \|_1^n n! \left( \int d^3y_1du_1 J(y_1) d^3y_2du_2 J^*(y_2) \mathsf{G}_\beta^{n+1}(y_1,u_1;y_2,u_2) \right)
\end{align*}
from which the estimate
given in \eqref{eq:T0JJ} is obtained
\[
\begin{aligned}
    |R_{JJ^*}^{0,n}| &\leq 2
   \beta \sqrt{\lambda}^{n} \|v\|_1^n n!  
\left| \int d^3y_1\int d^3y_2 \frac{1}{K_{\epsilon}^{n+1}}
 (y_1,y_2)J(y_1) J^*(y_2) \right|
\\
  & \leq 2
  \beta \sqrt{\lambda}^{n} 
  \left(
  \frac{\|v\|_1}{\epsilon}
  \right)^n  n! \frac{
   \|J\|_{K,1} \|J^*\|_1+\|J^*\|_{K,1} \|J\|_1}{2}.
\end{aligned}
\]
In the above, the extra factor $\beta$, compared to the estimate for $R^n_{JJ}$, appears because of the lack of base point $\delta_{x,u}$. To study the vertices which contain the condensate $\phi_0$ we observe that in $W_{J\phi_0}$, $W_{\phi_0J^*}$ and $W_{\phi_0\phi_0}$ the auxiliary field $A$ always multiplies $\phi_0$. There, we can use a strategy similar to the one used in \cite{GalandaPinamonti} to get rid of it using the parabolic equations of motion. Namely
\begin{align*}
W_{J\phi_0}= -2 \langle
   \mathrm{i} J, \mathsf{G}_\beta^{\sqrt{\lambda}A}  A \phi_0 \rangle
   &= -
   \frac{2}{\sqrt{\lambda}}\langle
   \mathrm{i} J, \mathsf{G}_\beta^{\sqrt{\lambda}A}  \sqrt{\lambda} A \mathsf{G}_\beta (\partial_u+K)\phi_0 \rangle
   \\
   &= -
   \frac{2}{\sqrt{\lambda}}
   \langle
   \mathrm{i} J, \left(\mathsf{G}_\beta^{\sqrt{\lambda}A}  - \mathsf{G}_\beta \right)(\partial_u+K)\phi_0 \rangle
   \\
   &= -
   \frac{2}{\sqrt{\lambda}}
   \langle
   \mathrm{i} J, \left(\mathsf{G}_\beta^{\sqrt{\lambda}A}  - \mathsf{G}_\beta \right)K\phi_0 \rangle
\end{align*}
Once they have been brought in this form, vertices of this kind with $n$ insertions of $\|v\|_1$ and base point $\delta_{x,u}$ are estimated in a similar manner as above, with a $\|J\|_K$ replaced by an extra factor $\|v\|_1/\epsilon$.

Similarly
\begin{align*}
W_{\phi_0J^*}= -2\langle
   A\phi_0, \mathsf{G}_\beta^{\sqrt{\lambda}A} \mathrm{i} J^* \rangle
   &= -
   \frac{2}{\sqrt{\lambda}}
   \langle
   K\phi_0, \left(\mathsf{G}_\beta^{\sqrt{\lambda}A}  - \mathsf{G}_\beta \right) \mathrm{i}J^* \rangle
\end{align*}

and
\begin{equation}\label{eq:Wphi0phi0}
\begin{aligned}
W_{\phi_0 \phi_0}= - 2 \langle
   A \phi_0, \mathsf{G}_\beta^{\sqrt{\lambda}A}  A \phi_0 \rangle
   &= -
   \frac{2}{\lambda}
   \langle
   \phi_0, \sqrt{\lambda}A\left(\mathsf{G}_\beta^{\sqrt{\lambda}A}  - \mathsf{G}_\beta \right)K\phi_0 \rangle
   \\
   &= -
   \frac{2}{\lambda}
   \langle
   \phi_0, (\partial_u+K)\mathsf{G}_\beta \sqrt{\lambda}A\left(\mathsf{G}_\beta^{\sqrt{\lambda}A} \right)K\phi_0 \rangle
   +
   \frac{2}{\sqrt{\lambda}}
   \langle
   \phi_0, A  \mathsf{G}_\beta K\phi_0 \rangle
   \\
   &= -
   \frac{2}{\lambda}
   \langle
   \phi_0, K \left(\mathsf{G}_\beta^{\sqrt{\lambda}A}
   -
   \mathsf{G}_\beta\right)K\phi_0 \rangle
   +
   \frac{2}{\sqrt{\lambda}}
   \langle
   \phi_0, A  \mathsf{G}_\beta K\phi_0 \rangle
\end{aligned}
\end{equation}

The corresponding contributions in the loop vertex expansion with $n$ insertion of $\|v\|_1$ and an external point are such that 
\begin{align*}
R_{J\phi_0}^n&= 
\langle \delta_{x,u} \otimes \underbrace{\|v\|_1\otimes \dots \otimes  \|v\|_1}_{n}, \frac{\delta^{n+1}}{\delta A \dots \delta{A}} \rangle W_{J\phi_0} 
\end{align*}
and similarly for $R_{\phi_0J^*}^n$ and $R_{\phi_0\phi_0}^n$.
Operating as in the analysis of $R_{JJ}^n$ 
 we get a proof of \eqref{eq:TJphi0}
\[
\begin{aligned}
    |R_{J\phi_0}^n| 
    &\leq 2\sqrt{\lambda}^{n} \|v\|_1^n
\sum_{k=0}^n
k! (n-k)!
\left| \int d^3y_1\frac{e^{-(\beta-u){K}_{\epsilon}}}{K_{\epsilon}^{k+1}}
 (y_1,x)J(y_1) \right|
 \left|
 \int d^3y_2
 \frac{e^{-u{K}_{\epsilon}}}{K_{\epsilon}^{n-k+1}}
 (x,y_2) K \phi_0 \right|
\\
  & \leq 2
  \sqrt{\lambda}^{n} 
  \left(
  \frac{\|v\|_1}{\epsilon}
  \right)^n 2^n n! 
   \left|
 \int d^3p
 \frac{1}{\hat{K}(p)}
 \hat{J}(p) \right| |\phi_0|
\end{aligned}
\]
where as before $e^{-(\beta -u)K_{\epsilon}}$ and $e^{-uK_{\epsilon}}$ are operator bounded by $1$. Moreover, in the last contribution $K_{\epsilon}^{n-k}$ is an operator which is bounded by $\epsilon^{-(n-k)}$ for $n-k>0$ and, for $n-k=0$ the composition of $K_{\epsilon}^{-1}$ with $K_{\epsilon}$ gives the identity and its integral kernel is the Dirac delta. In a similar way we obtain an estimate for \eqref{eq:T0Jphi0}
\[
\begin{aligned}
    |R_{J\phi_0}^{0,n}| 
    &\leq 2 \beta \sqrt{\lambda}^{n-1} \|v\|_1^n
\sum_{k=0}^n
n! 
\left| \int dy_1 \int dy_2 \frac{1}{K_{\epsilon}^{n+2}}
 (y_1,y_2)J(y_1) K \phi_0 \right|
\\
  & \leq 2\beta
  \sqrt{\lambda}^{n-1} 
  \left(
  \frac{\|v\|_1}{\epsilon}
  \right)^n  n! 
   \left|
 \int d^3p
 \frac{1}{\hat{K}(p)}
 \hat{J}(p) \right| |\phi_0|
\end{aligned}
\]
With analogous methods, we get a proof of \eqref{eq:Tphi0J}
\[
\begin{aligned}
    |R_{\phi_0 J^*}^n| 
  & \leq 2
  \sqrt{\lambda}^{n} 
  \left(
  \frac{\|v\|_1}{\epsilon}
  \right)^n 2^n n! 
   \left|
 \int d^3p
 \frac{1}{\hat{K}(p)}
 \hat{J}^*(p) \right| |\phi_0|
\end{aligned}
\]
of \eqref{eq:T0phi0J}
\[
\begin{aligned}
    |R_{\phi_0 J^*}^{0,n}| 
  & \leq 2
  \beta \sqrt{\lambda}^{n-1} 
  \left(
  \frac{\|v\|_1}{\epsilon}
  \right)^n  n! 
   \left|
 \int d^3p
 \frac{1}{\hat{K}(p)}
 \hat{J}^*(p) \right| |\phi_0|
\end{aligned}
\]
and, for $n + 1 \geq 2$, a proof of  \eqref{eq:Tphi0phi0}
\[
\begin{aligned}
    |R_{\phi_0 \phi_0}^n| 
  & \leq 2
  \sqrt{\lambda}^{n-1} 
  \left(
  \frac{\|v\|_1}{\epsilon}
  \right)^n 2^n n! 
 |\phi_0|^2.
\end{aligned}
\]
This concludes the proof.
\end{proof}

\begin{lemma}\label{le:EGF-convergence}
Consider the exponential generating function constructed over the set of rooted trees with $n$ vertices $\mathcal{T}_n$
\[
E(v)\coloneqq \sum_{n\geq 1}\frac{v^{n}}{n!}
\sum_{T\in\mathcal{T}_n} p_T,
\]
where the weight
\[
p_T = (e_T(1))!\prod_{j=2}^n (e_T(j)-1)!
\]
and $p_T=1$ if $T \in\mathcal{T}_1$. In the above, $e_T(j)$ is the number of edges in $T$ which contains the vertex $j$. Then, if $0\leq v<1/4$, the series defining $E(v)$ is absolutely convergent and
\[
E(v)=\frac{1-\sqrt{1-4v}}{2}.
\]
\end{lemma}
\begin{proof}
To compute the exponential generating function $E(v)$ we can use the methods similar to those in \cite{Flajolet_Sedgewick_2009}. Since the trees we are considering are rooted, once the root is fixed we can construct  
all the rooted trees attaching all possible rooted subtrees to the root we have already selected. 
Hence, the set of rooted trees $\mathcal{T}$ satisfies the symbolic equation (following the notation of \cite{Flajolet_Sedgewick_2009})
\[
\mathcal{T} = \{1\}\star \text{Set}(\mathcal{T}) 
\]
and the corresponding exponential generating function, defined as
\[
E'(v) \coloneqq \sum_{n\geq 0} \frac{v^n}{n!} \sum_{T\in\mathcal{T}_n}1,
\]
satisfies the functional equation  
\[
E'(v) = v e^{E'(v)} = v  \sum_{k\geq 0} \frac{E'(v)^k}{k!}.
\]
Notice that the term $\frac{E'(v)^k}{k!}$ counts exactly the number of trees which posses $k$ subtrees attached to the root.
Here, the factor $(k!)^{-1}$ appears because the subtrees are unordered. Consider now a rooted (hence ordered) tree $T$. The number of subtrees of higher order attached to the root is equal to $e_T(1)$. Similarly let $j$ be the internal edge of the rooted tree $T$. The number of subtrees of higher order attached to $j$ is $e_T(j)-1$. Hence, this implies that the first exponential generating function $E(v)$ can be rewritten as 
\[
E(v) =  v  \sum_{k\geq 0} k!\frac{E(v)^k}{k!} = v \frac{1}{1-E(v)}.
\]
From this equation we get that $E(v)-E(v)^2 = v$ and, since the above denominator should remain positive, the only admissible solution is
\[
E(v) = \frac{1-\sqrt{1-4v}}{2}.
\]
Thus, the series defining $E(v)$ is absolutely convergent if $0\leq v<1/4$.
\end{proof}

\section{Scattering length and ladder diagrams}\label{se:scattering-length-corrections}

In order to estimate the effect of the scattering length on the correlation functions, we analyse the truncated two-point function
\[
\omega_2(x,y) = - {\frac{1}{2}}
\left.\begin{pmatrix}
\frac{\delta^2}{\delta {J}(x) \delta {J}^*(y)} 
&
\frac{\delta^2}{\delta {J}(x) \delta {J}(y)}  
\\
\frac{\delta^2}{\delta J^*(x) \delta {J}^*(y)} 
&
\frac{\delta^2}{\delta J^*(x) \delta {J}(y)} 
\end{pmatrix}\left(\log Z(J,J^*) -\log Z(0,0) \right)
\right|_{J,{J}^*=0}
\]
and compare it with the expressions given in \eqref{eq:two-point-function} and in \eqref{eq:perturbation}.
Notice that in the previous formula the presence of ${\log Z(0,0)}$ is immaterial, as the latter does not depend neither on $J$ nor on $J^*$.

We proceed considering the expansion in terms of tree graphs given in \eqref{eq:LVE} and we analyse the contribution of specific graphs.
As a first analysis, we discard graphs containing the loop vertices because their contribution is rescaled with $1/\sqrt{\beta}$ and it vanishes in the limit $\beta\to \infty$. Namely, we consider only $W_{JJ^*}$, $W_{J\phi_0}$, $W_{\phi_0 \phi_0}$ and $W_{\phi_0J^*}$ which we recall here
\begin{align*}
&W_{JJ^*}= {-}2\langle {\mathrm{i}}J, \mathsf{G}^A_\beta {\mathrm{i}}{J}^*\rangle , 
\qquad
W_{\phi_0J^*}= {-}2\langle \phi_0 A, {\mathsf{G}^A_\beta} {\mathrm{i}}J^* \rangle,\\
&W_{J\phi_0}= {-}2\langle {\mathrm{i}}J, \mathsf{G}^A_\beta \phi_0 A\rangle
\qquad
W_{\phi_0\phi_0}= {-2}\langle \phi_0 A,{\mathsf{G}^A_\beta} A,\phi_0\rangle.
\end{align*}
In what follows, we analyse the possible trees among them which appear in the loop vertex expansion given in \eqref{eq:LVE}, under certain approximations. 
The first, and already discussed, approximation consists in estimating $\mathsf{G}^A_\beta$ with $\mathsf{G}_\beta$ in all vertices and considering the measure $d\mu_{\Upsilon(0)}$ at the place of 
$d\mu_{\Upsilon(w_T)}$ in $\log Z(J,{J}^*) -{\log Z(0,0)}$. This is a first perturbative order approximation, as, in the terminology of quantum field theory, we are considering in this way only diagrams at tree level. 

Although, depending on the regime of physical parameters of the theory, the first perturbative order is not the only relevant in this analysis when considering the Gross-Pitaevskii regime. Indeed, in that regime, contributions from higher order in perturbation theory are as relevant as some lower order terms. We shall discuss this further in the next section. 

Going back to the evaluation at first order in perturbation theory, with the above substitutions, the only non vanishing contribution to $\omega_2$ in 
\[
-\frac{1}{2}\frac{\delta^2}{\delta {J}(x) \delta {J}^*(y)} 
\left(
\log Z(J,{J}^*) -{\log Z(0,0)}
\right) 
\]
contains only a single $J$ and a single ${J}^*$ in
\[
\prod_{k=1}^n W(J,J^*,A_k,\phi_0) =
\prod_{k=1}^n \sum_{l_k\in \{ JJ^*, J\phi_0,\phi_0 J^*, \phi_0\phi_0\}} W_{l_k}.
\]
Furthermore, the trees over which the sum in 
$\log Z(J,{J}^*) -{\log Z(0,0)}
$ is taken are connected. Hence, because of the simplification in the measure, we have
\[
-\frac{1}{2}\frac{\delta^2}{\delta {J}(x) \delta {J}^*(y)} 
\left(-W_{JJ^*}  
- \sum_{n\geq 2} \frac{1}{n!} 
\sum_{T\in\mathcal{T}_n} 
\prod_{\{a,b\}\in E(T)} \Upsilon_{ab} \frac{\delta^2}{\delta A_a \delta A_b} 
\prod_{k=1}^n \sum_{l_k \in \{J\phi_0,\phi_0 J^*, \phi_0\phi_0\}} W_{l_k}
\right)\bigg|_{{J, {J}^*, A=0}}
\]
where the sum over $w_b$ has been taken.
Furthermore, in the sum over trees, the only contributing graphs are those with exactly two endpoints and $n-2$ bivalent internal vertices. From the above, the only allowed internal vertices are of the $W_{\phi_0 \phi_0}$ type. So, without loss of generality and up to a rearrangement, we have that
\begin{equation}\label{eq: multiploladder}
-\frac{1}{2}\frac{\delta^2}{\delta {J}(x) \delta {J}^*(y)} 
\left(-W_{JJ^*}  
- \sum_{n\geq 2}
\prod_{j = 1 }^{n-1} \Upsilon_{j(j+1)} \frac{\delta^2}{\delta A_j \delta A_{j+1}}
W_{J\phi_0}
\left(\prod_{k=2}^{n-2} W_{\phi_0\phi_0}
\right) W_{\phi_0J^*}
\right)\bigg|_{{J, {J}^*, A=0}}.
\end{equation}
Taking into account that in the internal vertices in the considered approximation there are two $A$s, we get that the operator to which this integral kernel corresponds is 
\begin{equation}\label{eq:omegaJJ}
\omega_{{J}J^*}
=  \mathsf{G}_\beta 
+
\sum_{n\geq 0} (-1)^{n+1}
 \mathsf{G}_\beta \left(\phi_0^2 \delta\otimes v (\mathsf{G}_\beta + \mathsf{G}_\beta^\dagger)\right)^n (\phi_0^2 \delta\otimes v)\mathsf{G}_\beta. 
\end{equation}
With similar arguments we obtain that
\begin{equation}\label{eq:omegaJJ1}
\begin{aligned}
\omega_{J^*{J}}
&=  \mathsf{G}^\dagger_\beta 
+
\sum_{n\geq 0} (-1)^{n+1}
 \mathsf{G}^\dagger_\beta \left(\phi_0^2 \delta\otimes v   (\mathsf{G}_\beta + \mathsf{G}_\beta^\dagger)\right)^n (\phi_0^2 \delta\otimes v)\mathsf{G}^\dagger_\beta 
\\
\omega_{J^*J^*}
&= 
\sum_{n\geq 0} (-1)^{n+1}
 \mathsf{G}^\dagger_\beta \left(\phi_0^2 \delta\otimes v   (\mathsf{G}_\beta + \mathsf{G}_\beta^\dagger)\right)^n (\phi_0^2 \delta\otimes v)\mathsf{G}_\beta 
\\
\omega_{{J}{J}}
&= 
\sum_{n\geq 0} (-1)^{n+1}
 \mathsf{G}_\beta \left(\phi_0^2 \delta\otimes v   (\mathsf{G}_\beta + \mathsf{G}_\beta^\dagger)\right)^n (\phi_0^2 \delta\otimes v)\mathsf{G}^\dagger_\beta.
\end{aligned}
\end{equation}
These are exactly the Dyson series given in \eqref{eq:two-point-function} with $\mathsf{v}$ replaced by $\phi_0^2 v$. In particular, computing the sum over the components of $\omega_2$ we get the operator
\[
\sum_{n\geq 0}
(-(\mathsf{G}_\beta + \mathsf{G}_\beta^{\dagger})
(\phi_0^2 \delta \otimes v))^n
(\mathsf{G}_\beta + \mathsf{G}_\beta^{\dagger})
\]
which is exactly the Dyson series representation of $\mathsf{M}$ in \eqref{eq:M} and in \eqref{eq:M1} with
$\mathsf{v}$ replaced by $\phi_0^2v$.\\

As last approximation, estimating $v$ with $\hat{v}(0) \delta(x-y)$ we get that the first order correction is proportional to $\phi^2_0 \hat{v}(0)$, which is the same result computed by Bogoliubov \cite{BogoliubovBEC}. This last approximation is seen as follows. Suppose that the function $\hat{v}(p)$ is analytic (e.g. the Fourier transform of a Schwartz function), and recall that the operator $v$ acts as a convolution on $L^2(\mathbb{R}^3)$. Namely, for any $f \in L^2(\mathbb{R}^3)$ its action in the Fourier space corresponds to 
\begin{equation*}
    \widehat{v f}(p) = \left( \sum_{n\geq 0} \frac{(\partial^n \hat{v})(0)}{n!} p^n\right)\hat{f}(p)
\end{equation*}
It is then manifest that the zeroth order contribution is the multiplicative operator containing  $\hat{v}(0)$ which, in position space, corresponds to the operator with integral kernel $\hat{v}(0) \delta(x-y)$. Formally speaking, this approximation treats the interactions among particles, as occurring when they spatially coincide. Therefore, it is a reasonable approximation in the so called dilute regime (small density of the condensate) where the inter-particle distance is much larger than the range of the interaction that can thus be approximated as being point-like.\\

The above corrections can be encoded at the level of the interaction Hamiltonian via a local term of the form 
\begin{equation}
\label{eq:ren-H}
\mathcal{H}_{Bog}=(\Psi+\Psi^*) \phi_0^2 \hat{v}(0)  (\Psi+\Psi^*)
=(\Psi+\Psi^*) \phi_0^2 8\pi \mathfrak{a}_0^{(0)} (\Psi+\Psi^*)
=(\Psi+\Psi^*) \hat{\mathcal{V}}_0 (\Psi+\Psi^*).
\end{equation}
We notice that in the above, the first term in the Born series of the scattering length appears, see e.g. \eqref{eq:born-series}. Moreover, we can also interpret it as the zeroth order contribution of the 
effective vertex\footnote{In the physics terminology this is often referred to rather as the self-energy.}: $\hat{\mathcal{V}}_0 = \phi_0^2 8\pi \mathfrak{a}_0^{(0)}$.

\subsection{First order correction: one loop approximation}

As a follow up to the above discussion, we proceed analysing the first loop corrections to $\log Z(J,J^*) - \log Z(0,0)$ due to $d\mu_{\Upsilon(w_T)}$. We recall that the role of $d\mu_{\Upsilon(w_T)}$ in \eqref{eq:LVE} is to add lines, and thus corrections, to the tree graphs analysed in the previous section. We also observe that the obtained corrections are, in the language of quantum field theory, at loop order.
Therefore, to compute the first corrections to the scattering length, we avoid
approximating $\mathsf{G}^{A}_\beta$ with 
$\mathsf{G}_\beta$. Diagrammatically, some of the extra factors $A$ are contracted with the $A$s at the extrema of the leg vertices, by the edges of the loop vertex expansion, and some other are contracted to other internal vertices.
The terms which are relevant, are those which are obtained considering loops which occur by doubling the edges $v$ connecting different leg vertices. These multiple $v$, and the corresponding loops, arise considering some of the first contributions of $d\mu_{\Upsilon(w_T)}$ in the loop vertex expansion \eqref{eq:LVE}. Specifically, the contributions to adjacent vertices.

For this reason, in order to analyse the first loop contributions, we start considering the corrections to the simplest vertex. Namely, we study corrections to
$
\phi_0^2 \mathsf{G}_\beta 
\delta \otimes v
\mathsf{G}_\beta.
$
At first order we get
\[
-\mathcal{M}\left. \Gamma_{12}^{2} \phi_0^2\mathsf{G}^A_\beta A 
 \otimes A
\mathsf{G}^A_\beta  \right|_{A=0}
=
\mathsf{G}_\beta
\mathcal{V}_1
\mathsf{G}_\beta
\]
where the integral kernel of $\mathcal{V}_1$ is such that
\begin{align*}
&\mathcal{V}_1(x,u;x',u')\\
&= 
-\phi_0^2(\mathsf{G}_\beta (\delta \otimes v )
\mathsf{G}_\beta)
(x,u;x',u') \delta(u-u')v(x,x'){-}
\phi_0^2(\mathsf{G}_\beta (\delta \otimes v ))(x,u;x',u')
(\mathsf{G}_\beta (\delta \otimes v ))(x,u;x',u').
\end{align*}
In Fourier domain, for both $x,u$ variables, it corresponds to
\[
\hat{\mathcal{V}}_1(q,m)
= -\frac{1}{\beta} \phi_0^2\sum_{n \in \mathbb{Z}}\int d^3p 
\;\left(
\hat{\mathsf{G}}_\beta(p,n)
\hat{v}(p)
\hat{\mathsf{G}}_\beta(p,n)
\hat{v}(q-p)
+
\hat{\mathsf{G}}_\beta(p,n)
\hat{v}(p)
\hat{\mathsf{G}}_\beta(q-p,m-n)
\hat{v}(q-p)
\right).
\]
Hence, in order to get an estimate of the integral of $\mathcal{V}_1$, we look for the above at vanishing momentum 
\[
\hat{\mathcal{V}}_1(0,0)
= -\frac{1}{\beta} \phi_0^2\sum_{n \in \mathbb{Z}}\int d^3p 
\;\left(
\hat{\mathsf{G}}_\beta(p,n)
\hat{v}(p)
\hat{\mathsf{G}}_\beta(p,n)
\hat{v}(-p)
+
\hat{\mathsf{G}}_\beta(p,n)
\hat{v}(p)
\hat{\mathsf{G}}_\beta(-p,-n)
\hat{v}(-p)
\right),
\]
and since $v$ is symmetric, after summing over $n$, we get
\[
\hat{\mathcal{V}}_1(0,0) =  -\frac{1}{\beta}\phi_0^2  \left\langle v,
\left(\frac{\beta \coth(\frac{\beta K}{2})
}{2K}+\frac{\beta^2}{4} \mathrm{csch} \left(\frac{\beta K}{2}\right)
\right)
   v \right\rangle.
\]
where the $\epsilon$ regulator does not play any role in this expression and thus its limit $\epsilon \to 0$ has been implicitly taken. Taking the limit of vanishing temperature, $\beta \to \infty$, we have the following result
\[
\lim_{\beta\to \infty}\hat{\mathcal{V}}_1(0,0) =  -\phi_0^2  \left\langle v, \frac{1}{2K}
   v \right\rangle.
\]
We have thus obtained the first correction to $8\pi \mathfrak{a}_0^{(0)}=\hat{v}(0)$ of the scattering length, which is of the form 
\[
8\pi \mathfrak{a}_0^{(1)}  = -
  \left\langle v , \frac{1}{2K} v \right\rangle.
\]
This corresponds to the first order term in the Born series \eqref{eq:born}. The vertex $\phi_0^2 \hat{v}(0)$, as considered in \eqref{eq:ren-H},
gets then renormalised to \begin{equation}\label{eq:first-renormalisation}
\hat{\mathcal{V}}_0(0)=\phi_0^2 \hat{v}(0)\to \phi_0^2 8\pi(\mathfrak{a}_0^{(0)}+\mathfrak{a}_0^{(1)})=\hat{\mathcal{V}}_0(0)+\hat{\mathcal{V}}_1(0).
\end{equation}
The important remark is that higher orders in $v$ in the effective vertex, and thus higher in the perturbation, have higher powers in $\phi_0^2$ compared to the previous order. Therefore, in regimes where $\phi_0^2$ is small compared to $v$, these terms are of lower order and can thus be neglected in this preliminary analysis. Although, this is not occurring for all the regimes considered in this paper. Indeed, using $v_{N,b}$ given in \eqref{eq:vnb} at the place of $v$ and considering the large $N$ limit, we see that for $b<1$, as e.g. in the mean field regime, $\hat{\mathcal{V}}_1$ scales with a negative power of $N$ and can thus be discarded. However, in the Gross-Pitaevskii regime, $b=1$, $\hat{\mathcal{V}}_1$ is of the same order as $\hat{\mathcal{V}}_0$ and thus cannot be discarded.\\\\

To conclude this section, let us analyse the one loop corrections arising when joining with extra lines the leg vertices which are not adjacent. With an explicit computation, considering thus all possible one loop contributions arising in this way, we obtain approximating in the second step to vanishing momentum that the renormalisation of $\hat{\mathcal{V}}_0$, in the limit $\beta \to \infty$, is
\begin{equation}\label{eq:1loopnon adjacent}
\begin{aligned}
\hat{\mathcal{V}}_{1R}
&= \phi_0^2\hat{v}(0) - \frac{1}{2} 
\int d^3p \left( 1-\sqrt{\frac{\hat{K}(p)}{\hat{K}(p)+2 \hat{v}(p) \phi_0^2}}
 \right) \hat{v}(-p)
 \\
 &\simeq \phi_0^2 \hat{v}(0) - 
 \phi_0^2\left\langle v, \frac{1}{\sqrt{K^2+2\phi_0^2\hat{v}(0)K}+K+2\phi_0^2\hat{v}(0)} v\right\rangle.
\end{aligned}
\end{equation}
The latter reduces to the above first correction to the Born series if $\phi_0^2\hat{v}(0)$ is small.
Similarly using $v_N$ as in \eqref{eq:vn} at the place of $v$
\begin{align*}
\hat{\mathcal{V}}_{1R} 
 = \phi_0^2 \hat{v}(0) - 
 \phi_0^2\left\langle v, \frac{1}{\sqrt{K^2+\frac{2\phi_0^2\hat{v}(0)}{N^2}K}+K+\frac{2\phi_0^2\hat{v}(0)}{N^2}} v\right\rangle.
\end{align*}
Therefore, for constant $\phi_0^2 \hat{v}(0)$, the difference $\hat{\mathcal{V}}_{1R} - \hat{\mathcal{V}}_1$ scales as $N^{-2}$ and it is thus negligible in the $N \to \infty$ limit. This discussion confirms that the relevant effect at one loop is accounted for by considering only those diagrams with loops formed by adjacent leg vertices.
\\

We proceed, in the next section, with the analysis of higher order loop corrections to $\sum_n\hat{\mathcal{V}}_n$ arising from the loop vertex expansion.

\subsection{Higher loop corrections and the Gross-Pitaevskii regime}\label{se:higher-loop-corrections}

Consider the regime where $v$ is substituted by $v_{N,b}$ in \eqref{eq:vnb}. We have that $\hat{v}_{N,b}(p)=N^{-1}\hat{v}(p)$ and, keeping the renormalised chemical potential $\tilde{\mu}$ constant, according to \eqref{eq:phi0}, $\phi_0^2$ scales as $\hat{v}_{N,b}^{-1}$. Thus, $\phi_0^2 \hat{v}_{N,b}(0)$ is constant in $N$. 
Then, similarly to the discussion above, the scaling behaviour of $\hat{v}_{N,b}(p)$ for large $N$ makes the higher loop corrections negligible for $b<1$. However, in the Gross-Pitaevskii regime, as $v = v_N =v_{N,1}$ and by the scaling properties of $\hat{v}_N(p)$ and of $K$ in this regime, higher order contributions to $\hat{\mathcal{V}}=\sum_n\hat{\mathcal{V}}_n$ are as relevant as lower order one. We claim that the relevant contributions to the effective $\hat{\mathcal{V}} =\sum_n\hat{\mathcal{V}}_n$ arise considering only the ladder diagram among adjacent vertices in the loop vertex expansion. This fact will be proven in the next section.
Ladder diagrams are certain higher order loop corrections to the product of $\mathsf{G}_\beta^A A$ with ${\mathsf{G}^A_\beta}^\dagger A$. 
Let us introduce a notation for it, $\mathcal{L}^\beta$ which is implicitly defined by
\begin{equation}\label{eq:ell-scattering}
\begin{aligned}
 \left. (e^{-\tilde{\Gamma}_{12}}-1) \phi_0^2\mathsf{G}^A_\beta A 
  \otimes A
 {\mathsf{G}^A_\beta}^\dagger  \right|_{A=0}
  &= 
   -\left. \int_0^1 d{w_{12}}\tilde{\Gamma}_{12} e^{-w_{12}\tilde{\Gamma}_{12}} \phi_0^2\mathsf{G}^A_\beta A  \otimes A
{\mathsf{G}^A_\beta}^\dagger
   \right|_{A=0}
 \\
&=-\sum_{n}(-1)^{n}\phi_0^2\left. \int_{0}^1 dw_{12} w_{12}^n\frac{\tilde\Gamma_{12}^{n+1}}{n!} \mathsf{G}^A_\beta A 
 \otimes A
{\mathsf{G}^A_\beta}^\dagger  \right|_{A=0}
\\
&=
-\sum_{n}\mathsf{G}_\beta
\mathcal{L}_n
{\mathsf{G}_\beta}^\dagger
\\
&=
-\mathsf{G}_\beta
\mathcal{L}^\beta
{\mathsf{G}_\beta}^\dagger.
\end{aligned}
\end{equation}
To obtain the analytic expression corresponding to the Fourier transform of $\mathcal{L}_n$ evaluated in $0$ we observe that it is a sum over all possible {\bf crossed ladder diagrams} with $n+1$ rungs. A crossing corresponds to a permutation $\pi\in \mathcal{P}\{1,\dots, n+1\}$. 
See Figure \ref{fig:ladder-diagram} for an example of a crossed ladder diagram. We start off studying some features of these generic ladder diagrams proving they give rise to the terms in the Born series of the scattering length of the potential $v$.\\

\begin{figure}
    \centering
\begin{tikzpicture}[
    thick,
    vertex/.style={circle, fill=black, inner sep=1.5pt, draw},
    wavy/.style={decorate, decoration={snake, segment length=3mm, amplitude=1mm}},
    directed/.style={postaction={decorate, decoration={markings, mark=at position 0.5 with {\arrow{Stealth}}}}}
]
    \node[vertex] (B1) at (0,0) {};
    \node[vertex] (B2) at (2,0) {};
    \node[vertex] (B3) at (4,0) {};
    \node[vertex] (B4) at (6,0) {};
    \node[vertex, label={below:$\phi_0$}] (B5) at (8,0) {};
    
    \node[vertex] (T1) at (0,2) {};
    \node[vertex] (T2) at (2,2) {};
    \node[vertex] (T3) at (4,2) {};
    \node[vertex] (T4) at (6,2) {};
    \node[vertex, label={above:$\phi_0$}] (T5) at (8,2) {};

    \draw[directed] (-1.5, 0) -- (B1) node[midway, below=4pt] {$p_0$};
    \draw[directed] (-1.5, 2) -- (T1) node[midway, above=4pt] {$k_0$};

    \draw[directed] (T1) -- (T2) node[midway, above=4pt] {$k_1$};
    \draw[directed] (T2) -- (T3) node[midway, above=4pt] {$k_2$}; 
    \draw[directed] (T3) -- (T4) node[midway, above=4pt] {$k_3$};
    \draw[directed] (T4) -- (T5) node[midway, above=4pt] {$k_4$}; 
    
    \draw[directed] (B1) -- (B2) node[midway, below=4pt] {$p_1$};
    \draw[directed] (B2) -- (B3) node[midway, below=4pt] {$p_2$}; 
    \draw[directed] (B3) -- (B4) node[midway, below=4pt] {$p_3$};
    \draw[directed] (B4) -- (B5) node[midway, below=4pt] {$p_4$}; 

    \draw[wavy] (B1) -- (T4) node[pos=0.25, left=4pt] {};
    \draw[wavy] (B2) -- (T1) node[pos=0.75, right=4pt] {};
    
    \draw[wavy] (B3) -- (T2) node[pos=0.25, left=4pt] {};
    \draw[wavy] (B4) -- (T3) node[pos=0.75, right=4pt] {};
    
    \draw[wavy] (B5) -- (T5) node[midway, right=4pt] {};

\end{tikzpicture}

    \caption{
    A crossed ladder diagram with the assignment of momenta described in the text. The wavy lines correspond to $\delta \otimes v$. The solid line to $\mathsf{G}_\beta$.
    }
    \label{fig:ladder-diagram}
\end{figure}
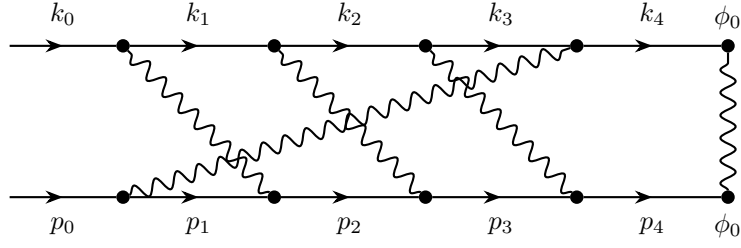

Each crossed ladder diagrams is formed by two rails. At order $n$, each rail is formed by $n$ edges (the propagator $(K\pm is)^{-1}$) and $n+1$ vertices $\{1,...,n +1\}$. 
We assign momenta $\{p_1,\dots, p_{n}\}$ to the edges of the lower rail, starting from the left and moving to the right. Hence, the momentum assigned to the edge joining the vertex $j$ to the vertex $j+1$ is $p_{j}$.
As the rungs are related to $v$, the rung attached to the lover vertex $j$ carries momentum $q_j=p_{j-1}-p_{j}$ with the convention that $p_{0}=0$. The momenta attached to the edges of the upper rail, by momentum conservation, depend on the order in which the rungs are attached. This order is given by the permutation $\pi$, so the momentum attached to the $m$-th upper edge (counted from the left to the right) is
\[
k_m = \sum_{j=1}^m q_{\pi{(j)}},
\]
where $k_{n+1}=k_{n}+q_{\pi(n+1)}=0$.
From this discussion we get that for vanishing external momentum
\[
k = A(\pi)p=L P(\pi)T p
\]
where $p=\{p_1,\dots, p_n,0\}$, $k=\{k_1,\dots, k_n,0\}$, $L,A,T$ are $n+1$ dimensional square matrices such that  
\[
L_{ij} = \begin{cases} 1, \qquad \text{if}\qquad i\geq j\\
0\qquad \text{otherwise}
\end{cases}
\qquad 
T_{ij} = \begin{cases} -1, 
\qquad \text{if}\qquad i=j\\
1, \qquad \text{if}\qquad j=i-1\\
0\qquad \text{otherwise}
\end{cases}
\]
and $P(\pi)_{ij}$ is the square matrix which realises the permutation $\pi$
\[
P(\pi)_{ij} = \begin{cases} 1, 
\qquad \text{if}\qquad i=\pi(j)\\
0\qquad \text{otherwise}.
\end{cases}
\]
We notice that $L^{-1}=-T$.
The corresponding vanishing momentum contribution at order $n$ is thus written as
\[
\hat{\mathcal{L}}_n(0,0) = 
\sum_{\pi\in \mathcal{P}\{1,\dots, n+1\}}  
\hat{\mathcal{L}}_{n,\pi}(0,0)
\]
where
\begin{equation}\label{eq:Lnpi}
\begin{aligned}
\hat{\mathcal{L}}_{n,\pi}(0,0)
&\coloneqq
(-1)^{n+1}\frac{\phi_0^2}{\beta^n}\int d^3p_1 \dots d^3p_n \sum_{\{m_1,\dots, m_n\}} \\
&\quad \prod_{j=1}^{n} 
\left(
\hat{v}(p_{j-1}-p_j)
\left(\frac{1}{K(p_j)+\mathrm{i} s_{m_j}}
\frac{1}{K((A(\pi)p)_j)+\mathrm{i} (A(\pi) s)_j }\right) 
\right)\hat{v}(p_n)
\end{aligned}
\end{equation}
with $p_0$, $p_{n+1}=0$ as the first is an auxiliary momentum introduced just for notation and the second is set to zero in this vanishing momentum analysis. Correspondingly, we set to zero also $s_{n+1}=0$. 
The contribution to $\mathcal{L}_{n,\pi}$ 
of $\pi=id$ is due to an ordinary ladder diagram, while the contributions to $\mathcal{L}_{n,\pi}$ for $\pi \neq id$ are due to crossed ladder diagrams. This remark is crucial, as the next lemma shows that the only relevant diagram, in the low temperature limit, are those without any crossing.

\begin{lemma}\label{le:vanishing-crossed-ladder}
There is  a single ladder diagram $\mathcal{L}_{n,\pi}$ which contributes to $\mathcal{L}_n$ in the limit $\beta \to \infty$. It is the ordinary ladder diagram without crossing
\[
\lim_{\beta \to \infty} \hat{\mathcal{L}}_{n,id} = (-1)^{n+1} \phi^2_0
\int d^3p_1 \dots d^3p_n
\prod_{j=1}^n\left(\hat{v}(p_{j-1}-p_j)\frac{1}{2K(p_j)}\right)\hat{v}(p_n),
\]
while for $\pi\neq id$, $\lim_{\beta \to \infty} \hat{\mathcal{L}}_{n,\pi} = 0 $.
\end{lemma}
\begin{proof}
Consider $\hat{\mathcal{L}}_{n,\pi}$ given in \eqref{eq:Lnpi} for a generic permutation and its limit for $\beta\to\infty$. In the limit $\beta\to\infty$ we observe that the sum over $m_1,\dots,m_n$ converges to an integral. The measure of this integral is
\[
\left(\frac{\beta}{2\pi}\right)^n dx_1 \dots dx_n 
\]
hence, 
\[
\lim_{\beta\to\infty}\hat{\mathcal{L}}_{n,\pi}(0,0)
=
(-1)^{n+1}{\phi_0^2}\int d^3p_1 \dots d^3p_n 
\prod_{j=1}^{n} 
\left(
\frac{v(p_{j-1}-p_j)}{2}
\right)
v(p_n)
I(p). 
\]
where, denoting with $x=(x_1,\dots, x_n,0)$, $p=(p_1,\dots, p_n,0)$ and 
$v=(v_1,\dots, v_n,0)$, we have
\begin{align*}
I(p)
&:= 
\frac{1}{(2\pi)^n}\int_{\mathbb{R}^n} dx_1 \cdots dx_n \prod_j
\left(\frac{1}{K(p_j)+\mathrm{i} x_j}
\frac{1}{K((A(\pi)p)_j)+\mathrm{i} (A(\pi)x)_j}
\right)
\\
&= 
\frac{1}{(2\pi)^n}\int_{\mathbb{R}^n} dx_1 \cdots dx_n 
\int_{\mathbb{R}_+^n} du_1 \cdots du_n \int_{\mathbb{R}_+^n} dv_1 \cdots dv_n
e^{- \sum_{j} (u_j K(p_j) +v_jK(Ap_j))  + \mathrm{i} x_j  (u_j + (A^{t} v)_{j})}
\end{align*}
where in the last integral we transformed 
the product of $K(p_j)$ into the exponential of a sum, using the appropriate integral transform. Integrating now over $x$ gives a product of delta functions
\begin{align*}
I(p) &= 
\int_{\mathbb{R}_+^n} du_1 \cdots du_n\int_{\mathbb{R}_+^n} du_1 \cdots du_n
e^{- \sum_{j} (u_j K(p_j) +v_jK(Ap_j))}
\prod_j\delta(u_j+(A^{t} v)_j)
\end{align*}
where, because of the delta function, the integral over $v$ is non vanishing only if for every $j$: $u_j=-(A^{t} v)_j$. So, in view of the constraints on the support of $u$ and $v$, we have that
$u+A^tv = 0$, for $A^t=T^t P(\pi)^tL^t$ and using $L^{t} = -(T^t)^{-1}$, only if it holds
\[
L^t u = P(\pi)^t L^t v.
\]
Denoting by $\tilde{u}=L^t u$ and by $\tilde{v} =L^t v$, since $v_j\geq 0$ and $u_{j} \geq 0$ for all $j$, it holds that the entries of $\tilde{u}$ and $\tilde{v}$ are ordered
\[
\tilde{u}_j \geq \tilde{u}_{j+1} \qquad \text{and} \qquad 
\tilde{v}_j \geq \tilde{v}_{j+1}
\]
as $\tilde{v}_j = \sum_{\ell = 1}^{n + 1 - j} v_{\ell}$ 
and the same for $\tilde{u}_j$. Consider now the points in the domain where for every $j$ we have $\tilde{u}_j > \tilde{u}_{j+1} $ and $\tilde{v}_j > \tilde{v}_{j+1}$. The elements in this region contribute to the integral forming $I(p)$ 
only if $P(\pi)^t \tilde{v}$ preserves the ordering of $\tilde{v}$. The only permutation which preserve the ordering is $\pi=id$. We furthermore observe that once the delta function is applied integrating over $v$, the region where $\tilde{u}_j=\tilde{u}_{j+1}$ for some $j$ is a zero measure set allowing to extend by continuity the above argument.
For these reasons, the integrand in $I$ is non vanishing only if $\pi=id$.
To conclude we notice that $A(id) = - \mathbb{1}$ and therefore
\begin{align}\label{eq:I(p)integrated}
I(p) = \int_{\mathbb{R}_+^n} du_1 \cdots du_n
e^{- \sum_{j} 2u_jK(p_j) } = \prod_j \frac{1}{2K(p_j)}
\end{align}
concluding the proof.
\end{proof}

\begin{rem}\label{rem:leg}
We observe that the case of 
\[
(-1)^{n}\phi_0^2\left. \int_{0}^1 dw_{12} w_{12}^n\frac{\Gamma_{12}^{n+1}}{n!} \mathsf{G}^A_\beta A 
 \otimes A
\mathsf{G}^A_\beta  \right|_{A=0}
=
\mathsf{G}_\beta
\mathcal{L}_n^\dagger
{\mathsf{G}_\beta} 
\]
can be treated analogously. The only difference is in the fact that the permutation which survives the limit $\beta\to\infty$ is the one which completely reverts the order. This is because
\begin{equation*}
    \hat{\mathsf{G}}_{\beta}^{\dagger}(p,n) = \hat{\mathsf{G}}_{\beta}(-p,-n)
\end{equation*}
and thus the rail $\hat{\mathsf{G}}_{\beta}^{\dagger}$ has the same properties as $\hat{\mathsf{G}}_{\beta}$ with reversed labelling of vertices. For this reason, the operator $\mathcal{L}_n^\dagger$, in the limit of vanishing temperature and momentum, coincides with $\mathcal{L}_n$.
\end{rem}

Because of Lemma \ref{le:vanishing-crossed-ladder} and Remark \ref{rem:leg}, in the limit $\beta\to\infty$, the relevant higher order corrections to the scattering length may follow only from the ladder diagrams. They produce operators which are of the following form for $n>1$ ($n=1$ was computed above)
\[
{\mathcal{V}}_n:L^{2}([0,\beta]\times\mathbb{R}^3)\to L^{2}([0,\beta]\times\mathbb{R}^3)
\]
which are given as 
\begin{equation}    \label{eq:Vnn}
{\mathcal{V}}_n \psi(u,x)
=
\int \tilde{\mathcal{V}}_n (\psi\otimes\phi^2_0) (u,x,y) d^3y,
\end{equation}
where $\psi\otimes\phi^2_0(u,x,y)=\psi(u,x)\phi_0^2$ and 
\[
\tilde{\mathcal{V}}_n :L^2([0,\beta]\times \mathbb{R}^3\times \mathbb{R}^3)
\to
L^2([0,\beta]\times \mathbb{R}^3\times \mathbb{R}^3)
\]
constructed as
\[
\tilde{\mathcal{V}}_n  = (-1)^n\left( v \delta(\mathsf{G}_\beta \otimes \mathsf{G}_\beta^\dagger)\delta\right)^{n}
( v).
\]
In this formula $v$ is meant as a multiplicative operator which acts on $
L^2([0,\beta] \times \mathbb{R}^3\times \mathbb{R}^3)
$ as 
\[
v\psi(u,x,y) = v(x-y)\psi(u,x,y).
\]
Furthermore, $(\mathsf{G}_\beta \otimes\mathsf{G}_\beta^\dagger)$ acts on $L^2([0,\beta]\times \mathbb{R}^3)\times L^2([0,\beta]\times \mathbb{R}^3)$ and the delta functions before and after, $\delta(\mathsf{G}_\beta \otimes \mathsf{G}_\beta^\dagger)\delta$,
forces the two operators to act at equal time. These operators can be given in terms of recursive relations. Actually, it holds that 
\[
\tilde{\mathcal{V}}_{n}=
- v\;\delta(\mathsf{G}_\beta \otimes  \mathsf{G}_\beta^\dagger)\delta\, 
\tilde{\mathcal{V}}_{n-1}, \qquad \tilde{\mathcal{V}}_0= v
\]
and the fixed point limit of these relations are a remnant of the inhomogeneous {\bf Bethe-Salpeter} equation for
$\tilde{\mathcal{V}}=\sum_{n}\tilde{\mathcal{V}}_n$, given in terms of 
the thermal propagators $\mathsf{G}_\beta$ and $\mathsf{G}_\beta^\dagger$ with interaction vertex $v$
\begin{equation}\label{eq:BSE}
\tilde{\mathcal{V}}=
 v
-   v \delta(\mathsf{G}_\beta \otimes \mathsf{G}_\beta^\dagger) \delta\,
\tilde{\mathcal{V}}.
\end{equation}
We expect similar relations to hold also in the evaluation of the trace of $\tilde{\mathcal{V}}$ also in the limit $\beta\to\infty$.

\vspace{8mm}

With this in mind, let us consider the $\mathcal{V}_n$ in \eqref{eq:Vnn} in terms of $\tilde{\mathcal{V}}_n$ starting in full generality. They are explicitly given as 
\begin{equation}\label{eq:Vn}
\hat{\mathcal{V}}_n(0,0)  = (-1)^n\phi_0^2\frac{1}{\beta^n}
\int d^3p_1 \dots d^3p_n \sum_{\{m_1,\dots, m_n\}}  \prod_{j=1}^{n} 
v(p_{j-1}-p_j)\left(\frac{1}{K^2(p_j)+s_{m_j}^2}\right)
v(p_n) 
\end{equation}
with $p_0=0$. The sum over the internal Matsubara frequencies can be taken directly observing that the sum we have to take has the form 
\[
\sum_{n \in \mathbb{Z}} \frac{1}{(K^2 + s_n^2)} =
\left(\frac{\beta \coth(\frac{\beta K}{2})
}{2K}+\frac{\beta^2}{4} \mathrm{csch} \left(\frac{\beta K}{2}\right)
\right).
\]
Therefore, outside the limit $\beta \to \infty$, additional thermal corrections arise from both this expression and from crossed ladder diagrams which are not guaranteed to vanish. Instead, in the limit $\beta\to\infty$, we have
\begin{align*}
\lim_{\beta\to\infty}\hat{\mathcal{V}}_n(0,0)  
&= (-1)^n\phi_0^2
\int d^3p_1 \dots d^3p_n \prod_{j=1}^{n} \left(
v(p_{j-1}-p_j)\frac{1}{2K(p_j)}\right)
v(p_n) 
\\
&= -\phi_0^2 \left\langle v, \left(-\frac{1}{2K}v\right)^{n-1}\frac{1}{2K} v\right\rangle.
\end{align*}
Up to the factor $\phi_0^2$, the operators on $L^2(\mathbb{R}^3)$ which represent these contributions satisfy the recursive relation
\[
\mathcal{K}_n = \frac{1}{2K} v \mathcal{K}_{n-1}  , \qquad \mathcal{K}_1=\frac{1}{2K}.
\]
These recursive relations are themselves remnant of the \textbf{Bethe-Salpeter} equations given in \eqref{eq:BSE} for the sum of all the contributions $\mathcal{K}_j$. 
Namely, writing
\[
\mathcal{A}_n = \sum_{j=1}^n\mathcal{K}_j ,
\]
we have that 
\[
\mathcal{A}_n= \frac{1}{2K}-\frac{1}{2K}v \mathcal{A}_{n-1}, \qquad \mathcal{A}_1=\frac{1}{2K}
\]
and we observe that the fixed point of these recursive relation is the operator $\mathcal{A}$ given by
\[
\mathcal{A} = \frac{1}{2K} - \frac{1}{2K} v \mathcal{A},
\]
and thus
\[
\mathcal{A} = \frac{1}{2K +v}.
\]
Now we are in the position to understand the crucial assumption on the scaling regime of the interaction. Considering $v_{N,b}$ at the place of $v$ in the above, $K=-\Delta/ 2m$ and recalling that for fixed $\tilde\mu$ the condensate density $\phi_0^2$ scales as $N$ it holds that
\[
\hat{\mathcal{V}}^{v_{N,b}}_n(0,0)
=
\frac{1}{N^{(3-3b)(n+1)}}\hat{\mathcal{V}}_n(0,0).
\]
Hence, for $b<1$, contributions from higher order $n$ ladder diagrams are suppressed in the limit $N\to\infty$. At the same time, 
in the Gross-Pitaevskii regime, $b=1$, $\hat{\mathcal{V}}_n(0,0)$ does not scale with $N$ and thus all contributions are of the same relevance. Therefore, the above computed corrections to the scattering length, are relevant only in the Gross Pitaevskii regime. Namely, for $n\geq 1$ we have
\[
8\pi\sum_{j=0}^n \mathfrak{a}_0^{(j)} = \hat{v}(0) - 
\langle v, \mathcal{A}_n v\rangle
\]
and the full scattering length is obtained in the limit $n\to\infty$
\[
8\pi\sum_{j=0}^{\infty} \mathfrak{a}_0^{(j)} = \hat{v}(0) - 
\left\langle v, \frac{1}{2K + v} v\right\rangle.
\]
Hence,
\[
\hat{\mathcal{V}}(0,0)= \phi_0^2 8\pi \mathfrak{a}_0.
\]

\subsection{Technical lemmata on multiple ladder diagrams}

We conclude this section with some technical results about multiple rails ladder diagrams which arise in the evaluation of higher order contributions to $\mathcal{V}$ (see Figure \ref{fig:multi-rail-ladder-diagram} for an example). A result is the proof of the claim, made at the beginning of previous section, that ladder diagrams among adjacent vertices are the most relevant one. We start off considering diagrams with multiple rails where multiple rungs connect only adjacent vertices.
We also keep $\epsilon$ finite and positive to avoid infrared divergences. At this stage, we shall actually consider the limit $\epsilon$ to $0$ after having considered the limit $N\to\infty$. In this first lemma below, to keep its proof simple, we shall scale $\epsilon$ with $N^2$. In the following lemmata we extend this result dropping this assumption.

\begin{figure}
    \centering

\begin{tikzpicture}[
    thick,
    vertex/.style={circle, fill=black, inner sep=1.5pt, draw},
    wavy/.style={decorate, decoration={snake, segment length=3mm, amplitude=1mm}},
    directed/.style={postaction={decorate, decoration={markings, mark=at position 0.5 with {\arrow{Stealth}}}}}
]

    \node[vertex] (B1) at (0,0) {};
    \node[vertex] (B2) at (3,0) {};
    \node[vertex, label={right:$\phi_0$}] (B3) at (7.5,0) {};
    
    \node[vertex, label={left:$\phi_0$}] (M1) at (0,1.5) {};
    \node[vertex] (M2) at (1.5,1.5) {};
    \node[vertex] (M3) at (3,1.5) {};
    \node[vertex] (M4) at (4.5,1.5) {};
    \node[vertex] (M5) at (6,1.5) {};
    \node[vertex, label={right:$\phi_0$}] (M6) at (7.5,1.5) {};

    \node[vertex, label={left:$\phi_0$}] (T1) at (0,3) {};
    \node[vertex] (T2) at (3,3) {};
    \node[vertex] (T3) at (7.5,3) {};

    \draw[directed] (-1.5, 0) -- (B1) node[midway, below=4pt] {$p$};
    \draw[directed] (T3) -- (9, 3)  node[midway, above=4pt] {$p$};

    
    \draw[directed] (T1) -- (T2) node[midway, above=4pt] {};
    \draw[directed] (T2) -- (T3) node[midway, above=4pt] {}; 

    \draw[directed] (M1) -- (M2) node[midway, above=4pt] {};
    \draw[directed] (M2) -- (M3) node[midway, above=4pt] {}; 
    \draw[directed] (M3) -- (M4) node[midway, above=4pt] {};
    \draw[directed] (M4) -- (M5) node[midway, above=4pt] {}; %
    \draw[directed] (M5) -- (M6) node[midway, above=4pt] {}; %
    
    \draw[directed] (B1) -- (B2) node[midway, below=4pt] {};
    \draw[directed] (B2) -- (B3) node[midway, below=4pt] {}; 

    
    \draw[wavy] (B1) -- (M1) node[pos=0.25, left=4pt] {};
    \draw[wavy] (B2) -- (M3) node[pos=0.75, right=4pt] {};
    
    \draw[wavy] (B3) -- (M6) node[pos=0.25, left=4pt] {};

    \draw[wavy] (M2) -- (T1) node[pos=0.75, right=4pt] {};
    
    \draw[wavy] (M4) -- (T2) node[midway, right=4pt] {};

    \draw[wavy] (M5) -- (T3) node[midway, right=4pt] {};

\end{tikzpicture}

\begin{tikzpicture}[
    thick,
    vertex/.style={circle, fill=black, inner sep=1.5pt, draw},
    wavy/.style={decorate, decoration={snake, segment length=3mm, amplitude=1mm}},
    directed/.style={postaction={decorate, decoration={markings, mark=at position 0.5 with {\arrow{Stealth}}}}}
]

    \node[vertex] (B1) at (6,0) {};
    \node[vertex, label={right:$\phi_0$}] (B2) at (7.5,0) {};
    
    \node[vertex, label={left:$\phi_0$}] (M1) at (0,1.5) {};
    \node[vertex] (M2) at (1.5,1.5) {};
    \node[vertex] (M3) at (3,1.5) {};
    \node[vertex] (M5) at (6,1.5) {};
    \node[vertex, label={right:$\phi_0$}] (M6) at (7.5,1.5) {};

    \node[vertex, label={left:$\phi_0$}] (T1) at (0,3) {};
    \node[vertex] (T2) at (1.5,3) {};
    \node[vertex] (T3) at (3,3) {};

    \draw[directed] (-1.5, 0) -- (B1) node[midway, below=4pt] {$p$};
    \draw[directed] (T3) -- (9, 3)  node[midway, above=4pt] {$p$};

    
    \draw[directed] (T1) -- (T2) node[midway, above=4pt] {};
    \draw[directed] (T2) -- (T3) node[midway, above=4pt] {}; 

    \draw[directed] (M1) -- (M2) node[midway, above=4pt] {};
    \draw[directed] (M2) -- (M3) node[midway, above=4pt] {}; 
    \draw[directed] (M3) -- (M5) node[midway, above=4pt] {};
    \draw[directed] (M5) -- (M6) node[midway, above=4pt] {}; %
    
    \draw[directed] (B1) -- (B2) node[midway, below=4pt] {};

    
    \draw[wavy] (B1) -- (M5) node[pos=0.25, left=4pt] {};
    \draw[wavy] (B2) -- (M6) node[pos=0.75, right=4pt] {};
    

    \draw[wavy] (M1) -- (T1) node[pos=0.75, right=4pt] {};
    
    \draw[wavy] (M2) -- (T2) node[midway, right=4pt] {};

    \draw[wavy] (M3) -- (T3) node[midway, right=4pt] {};

\end{tikzpicture}

    \caption{
    An example of 1PI ladder diagram with multiple rails in the first and
    an example of 1PR ladder diagram with multiple rails in the second.
    The wavy lines correspond to $\delta \otimes v$. The solid line to $\mathsf{G}^\beta$. 
    }
    \label{fig:multi-rail-ladder-diagram}
\end{figure}
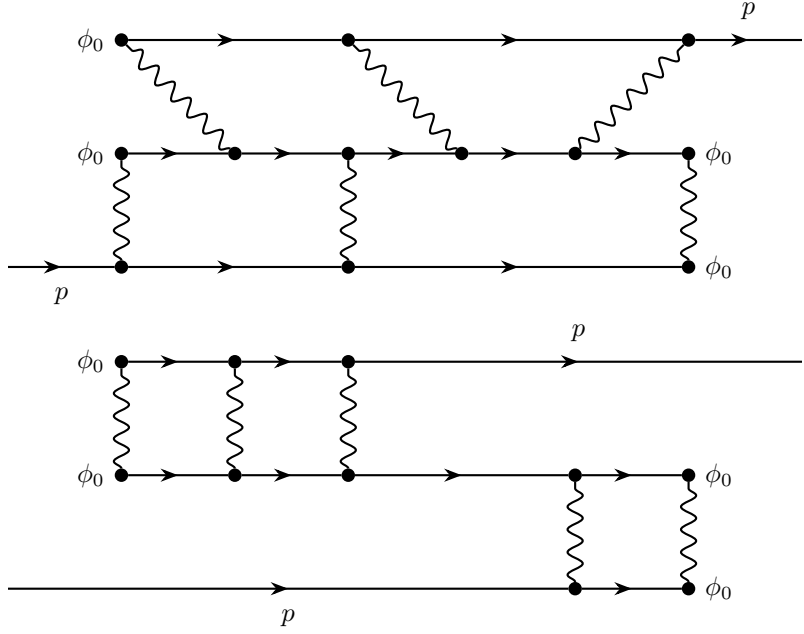

\begin{lemma}\label{le:power-counting} 
Multiple rails ladder diagrams which are 1PI\footnote{This stands for \textbf{one particle irreducible diagrams} that, borrowing terminology of quantum field theory, identifies diagrams that cannot be separated into two disconnected pieces by cutting a single internal (rail) line. Correspondingly we denote by 1PR the \textbf{one particle reducible} ones. Examples are provided in Figure \ref{fig:multi-rail-ladder-diagram}.} and which are characterised by a number of rails $n>2$ are negligible in the limit of $\beta \to \infty$ and for $N\to\infty$. 
In particular, in the Gross-Pitaevskii regime where $v_N$ is given in \eqref{eq:vn}, contributions with crossing rungs vanish in the limit $\beta\to\infty$ while contributions with non crossing rungs scale as $N^{4-2n}$ if $\epsilon=N^2\tilde{\epsilon}$ for fixed $\tilde{\epsilon}$ in $K_\epsilon$
\end{lemma}
\begin{proof}
To verify that the contributions of generic multi rails, which are 1PI, are negligible in the desired limits we substitute $v_N$ in $\tilde\Gamma$. In the limit $\beta \to \infty$ the contributions with crossed rungs are negligible for ''causality reasons" as proven in Lemma \ref{le:vanishing-crossed-ladder}. 

So, consider a generic contribution which is 1PI and that has no crossing rungs. 
If it has $n$ rails, it must have at least $r \geq 2n-2$ rungs, corresponding to $2n-2$ terms $v_N$. Furthermore, recalling the expression in \eqref{eq: multiploladder}, it has a number of $\phi_0$ attached to the rails with no exit momentum which is equal to $2n-2$. 
Finally, $l$, the number of loops in this diagram, corresponds to 
\[
l= r-n+1.
\]
To estimate the contribution of this generic diagram we observe that the number of integrations over the internal momenta,
which is equal to the number of sums over internal Matsubara frequencies, is the number of loops in the considered diagram. 
The number $k$ of propagators $\mathsf{G}_\beta$ in the diagrams is equal to 
\[
k= n+2r-2n = 2r-n,
\]
as each rung separates two propagators by two and the end point of a rail which has a $\phi_0$ or an external momentum is reached by only one single propagator and one single rung.

In the limit of $\beta\to\infty$ the sum over Matsubara frequency can be estimated with an integral, as in the proof of Lemma \ref{le:vanishing-crossed-ladder}, introducing the auxiliary parameters $u$ and $v$ and taking all the integrals. We end up with integrals over spatial momenta which involve various $\hat{v}_N$ and $K_{\epsilon}^{-1}$ (similar to \eqref{eq:I(p)integrated}). The final contribution of the graph we are analysing is further multiplied by $\phi_0$ with appropriate power.
Altogether, the contribution of one such 1PI diagram is
\[
C_N =  (\phi_0^2)^{n-1}
\int
\prod_{l=1}^{k-l} 
\frac{1}{K_l(Ap)} \prod_{h=1}^{r} \hat{v}(Bp)
\prod_{j=1}^{l}d^3p_j
\]
with $A$ and $B$ are suitable matrices realising the crossing of vertices similar to Lemma \ref{le:vanishing-crossed-ladder}. To obtain this expression, we used that each integration or sum over the Matsubara frequencies reduces the number of factors of $K_{\epsilon}$ by $1$. 
We study how this contribution scales with $N$ recalling that $\phi_0^2$ scales as $N$, changing the variable of integration from $p$ to $p/N$, $\hat{v}_N(p)= \hat{v}(p/N)/N$ and $K_{\epsilon}(p)= N^2 K_{\epsilon}(p/N)$. We have, denoting for $C=C_1$, that the scaling with $N$ is
\[
C_N  = C \frac{N^{n-1}}{N^{2 (k-l)}N^{r}}N^{3 l} = C N^{4-2n}.
\]
Therefore, since the number of rails $n$ is larger than $2$, we have the thesis.
\end{proof}

Recall that the case of multiple rail ladder diagrams that are 1PR, is considered in the definition of $\mathcal{L}^{\beta}$ and thus in the analysis of the scattering length of previous section. Next, in this analysis, we consider ladders where the rungs connect non adjacent rails.

\begin{lemma}\label{le:non-adjacent-loops}
Keeping $\epsilon>0$ and finite in $K_\epsilon$, multiple rail diagrams with insertion of $v$ to non adjacent rails are negligible in the Gross-Pitaevskii regime in the limit  $N\to\infty$. 
\end{lemma}
\begin{proof}
We analyse corrections to the self energy arising from loops formed joining non adjacent vertices and compare it with adjacent ones. So, we consider a path tree graph $T$ formed by $n$ points that have valence\footnote{The \textbf{valence} of a vertex in a tree graph $T$ is the number of edges that reach that vertex.} at most 2, and we analyse the contribution to the LVE 
\[
\prod_{b\in E_T}
\left(\int _0^1 dw_b\tilde\Gamma_{b}\right)e^{-\sum_{i<j }w_T(i,j) \tilde\Gamma_{ij}} \underbrace{e^{-\Gamma}W_{\phi_0\phi_0} \otimes \dots \otimes e^{-\Gamma}W_{\phi_0\phi_0} }_n.
\]
We focus on: the rung between the first two vertices and the interpolating parameter $w_{23}$. We observe that when $w_{23}=0$, we are removing all loops which contain vertex $1$ or $2$ and that are formed with the basic path tree adding extra joints from vertex $1$ to vertices $j\geq 3$ or from vertex $2$ to vertices $j\geq 3$. 
The extra joints between vertices $2$ and $3$, form loops between adjacent vertices that we treated separately. Hence, in estimating the effect of non adjacent vertices, we have to remove the multiple loops formed by joining adjacent rails $12$ and $23$ which we denote below by $R_{23}$. To estimate the reminder, we start analysing 
\begin{align*}
F(s)
&=\int_0^s dw_{13}\prod_{b\in E_T, b\neq \{23\}}\left(\int _0^1 dw_b\tilde\Gamma_{b}\right)
e^{-\sum_{i<j }w_T(i,j) \tilde\Gamma_{ij}} \underbrace{e^{-\Gamma}W_{\phi_0\phi_0} \otimes \dots \otimes e^{-\Gamma}W_{\phi_0\phi_0} }_n
\\
&=\prod_{b\in E_T}\left(\int _0^1 dw_b\tilde\Gamma_{b}\right)
e^{-\sum_{i<j }w^s_T(i,j) \tilde\Gamma_{ij}} \underbrace{e^{-\Gamma}W_{\phi_0\phi_0} \otimes \dots \otimes e^{-\Gamma}W_{\phi_0\phi_0} }_n
\end{align*}
where $w^s_T(i,j)$ is the minimum of the rescaled weights $\tilde{w}^s_{b}$ of the edges of the unique path in $T$ joining $i$ and $j$, and the rescaled weights are such that 
\[
\tilde{w}^s_{ij}= \begin{cases}
    sw_{ij}, \qquad ij=23\\
    w_{ij}, \qquad \text{otherwise}
\end{cases}.
\]
It holds that 
\begin{align*}
F(0)-F(1)=&-\int ds \frac{d}{ds}F(s)\\
=&-\int_0^1 ds \prod_{b\in E_T}\left(\int _0^1 dw_b\tilde\Gamma_{b}\right)
e^{-w_{12}\tilde{\Gamma}_{12}}
\frac{d}{ds}e^{-\sum_{i<j }w^s_T(i,j) \tilde\Gamma_{ij}} \underbrace{e^{-\Gamma}W_{\phi_0\phi_0} \otimes \dots \otimes e^{-\Gamma}W_{\phi_0\phi_0} }_n
\end{align*}
and we decompose it in
\[
F(0)-F(1)=R_1+R_2+R_{23}
\]
where for $m\in\{ 1,2\}$
\begin{align*}
R_m
=&\int_0^1 ds 
\sum_{l\geq 2+m}\int_{[0,1]^{n-1}, sw_{23}< w_{ik} \, \text{for} \, 2\leq i<k\leq l }
\prod_{b\in E_T}  dw_b
\tilde\Gamma_{b}
\left( 
w_{23}\tilde\Gamma_{ml}\right)
e^{-\sum_{i<j}w^s_T(i,j)\tilde{\Gamma}_{ij}}
(e^{-\Gamma}W_{\phi_0\phi_0})^{\otimes n}
\end{align*}
and
\begin{align*}
R_{23}
=&\int_0^1 ds 
\int_{[0,1]^{n-1}}
\prod_{b\in E_T}  dw_b
\tilde\Gamma_{b}
\left( 
w_{23}\tilde\Gamma_{23}\right)
e^{-\sum_{i<j}w^s_T(i,j)\tilde{\Gamma}_{ij}}
\underbrace{e^{-\Gamma}W_{\phi_0\phi_0} \otimes \dots \otimes e^{-\Gamma}W_{\phi_0\phi_0} }_n.
\end{align*}
The contribution $R_{23}$ is responsible for the occurrence of multiple loops between the joints $12$ and $23$. Hence, the remainder we need to estimate is
\[
R=R_1+R_2.
\]
To find a bound for $R_m$ we observe that the action of the exponential involving $-\sum_{i<j} w^s_T(i,j)\tilde\Gamma_{ij}$ and the exponentials $e^{-\Gamma}$ in front of the factors $W_{\phi_0\phi_0}$ can be bounded by Lemma \ref{le=estimate1}.
The terms which remain to be bounded are the one loop contributions formed joining non adjacent vertices. However, we have already computed in \eqref{eq:1loopnon adjacent} the action of these corrections. There, we have proven that they are negligible in the Gross-Pitaevskii regime when $N\to\infty$.
\end{proof}

Lemma \ref{le:power-counting} guarantees that in the limit of vanishing temperature, contributions to the effective vertices from ladder diagrams with multiple rails and which are 1PI are negligible in the Gross-Pitaevskii regime if $\epsilon$ is taken with some limitations. Instead, Lemma \ref{le:non-adjacent-loops} estimates the effect of loop contribution to 
$\mathcal{V}$ from loops involving non adjacent leg vertices. They turn out to be also negligible in the Gross-Pitaevskii regime. The next lemma specifies these observation for the graphs giving origin to the two-point function of the theory for finite $\beta$ and keeping $\epsilon>0$, without extra scaling, extending thus the results of the previous lemmata.

\begin{lemma}\label{lm:adjacent 1PI loops are negligible}
Let $\epsilon>0$ and finite in $K_\epsilon$. Consider adjacent loops which appear in the contribution to the loop vertex expansion on path graphs\footnote{A path graph is a diagram in which the internal vertices have valence $2$ while the external vertices have valence $1$.}. The adjacent loops which are 1PI are negligible in the Gross-Pitaevskii regime for $N\to\infty$.
\end{lemma}
\begin{proof}
To prove this result, we analyse the contribution arising form a path graph which is a tree whose vertices have valence at most $2$ and correspond to $e^{-\Gamma}W_{\phi_0\phi_0}$, the case with $e^{-\Gamma}W_{J\phi_0}$ or $e^{-\Gamma}W_{\phi_0J^*}$ at the extreme can be treated similarly. We start with a path formed by three vertices. The intermediate vertex is denoted by $2$.
The corresponding contribution has the form
\[
C_T=\mathcal{M}\prod_{b\in E_T}
\left(\int _0^1 dw_b\tilde\Gamma_{b}\right)e^{-\sum_{i<j }w_T(i,j) \tilde\Gamma_{ij}} \left(e^{-\Gamma}W_{\phi_0\phi_0}\right)^{\otimes 3}
\]
Here, $W_{\phi_0\phi_0} = \langle 
\phi_0 K \mathsf{G}_\beta^A,K\phi_0\rangle$ and we discarded the contribution linear in $A$.

Consider the functional derivatives $\delta/\delta A_2$  present in 
$e^{-w_T(1,2)\tilde{\Gamma}_{12}}$, and $\prod_b \tilde{\Gamma}_b$. Recall that when a functional derivative 
$\langle f,\delta/\delta A_2\rangle$ acts on
$\mathsf{G}_\beta^A$ gives 
$-\mathsf{G}_\beta^Af\mathsf{G}_\beta^A$. Hence, it splits  
$\mathsf{G}_\beta^A$ in two parts. 
But then, by Leibniz rule, a second functional derivative gives two terms: one corresponding to the action of $f$ on the left $\mathsf{G}_\beta^A$ and one for its action on the right $\mathsf{G}_\beta^A$.
Hence, after applying all the functional derivatives 
$\delta/\delta A_2$ in $\tilde{\Gamma}_{12}$ and in $\tilde{\Gamma}_b$ for $b=23$,
we can split their action in a sum over the order in which their effect on $\mathsf{G}_\beta^A$ occurs.

We have three options: the splitting due to $\delta/\delta A_2$ in $\tilde{\Gamma}_b$ for $b=23$ occurs in the middle of some of the splitting due to 
$e^{-w_T(1,2)\tilde{\Gamma}_{12}}\tilde{\Gamma}_{12}$, occurs on the right or on the left.
As we shall see below, the contributions when the splitting due to $\tilde{\Gamma}_b$, with $b=23$, occurs in the middle, are negligible. 
Let us discuss in detail the case where it occurs at the right. The left case is analogous.\\

By the combinatorics of the Leibniz rule and of the exponential, we can always assume that the functional derivative which splits the right-most propagator in the second vertex is the one which is not exponentiated.   
Hence, after applying $\tilde{\Gamma}_b$ with $b=12$, the second vertex is split in two parts.
Now, apply a replica trick and change the name of some of the auxiliary fields from $A$ to $B$ in the first vertex and in the propagator $\mathsf{G}_\beta^A$ present in the left part of the second vertex. 

By doubling the auxiliary fields we have to double also the corresponding action of 
$\tilde{\Gamma}$ and ${\Gamma}$
\[
\tilde{\Gamma}_{ij} =  2\Upsilon_{ij} \left(\frac{\delta}{\delta A_i}+\frac{\delta}{\delta B_i}\right)
\left(\frac{\delta}{\delta A_j}+\frac{\delta}{\delta B_j}\right).
\]
Hence, the covariance of these new fields is the tensor product of $\Upsilon\otimes R$ where $R=\mathsf{J}_2$ is a $2\times 2$ matrix whose entries are all equal to $1$. Let us now change the off diagonal entries, correlating the $A$ and $B$ fields, from $1$ to $s\in[0,1]$. Then, we denote the resulting covariance by $
\Upsilon \otimes R^s
$
and by ${\tilde{\Gamma}}^s$ the corresponding differential operator.
Let us also introduce
\[
C_T(s) = \mathcal{M}\prod_{b\in E_T}
\left(\int _0^1 dw_b\tilde\Gamma_{b}\right)e^{-\sum_{i<j }w_T(i,j) {\tilde\Gamma}^s_{ij}} \left(e^{-\Gamma^s}W_{\phi_0\phi_0}\right)^{\otimes 3}.
\]
We notice that $C_T(1)=C_T$ and $C_T(0)$ is the contribution where the first vertex appears attached to the third by 1PR graphs. Thus, we just need to estimate the remainder
\[
C_T(1)-C_T(0) =\int_0^1 ds \frac{d C_T(s)}{ds}.
\]
However, the improved covariance $\Upsilon\otimes R^{s}$ is positive, so Lemma \ref{le=estimate1} is used
to estimate $|C_T(1)-C_T(0)|$ and to bound the action of the various exponentials of $\tilde{\Gamma}^s_{ij}$ and $\Gamma$ by $1$.
The result is
\begin{align*}
|C(1)-C(0)|\leq&\int_0^1 ds 
\int_{[0,1]^2} dw_{12}dw_{23}\\
&
\left|\tilde{\Gamma}_{12}^{BB}
\tilde{\Gamma}_{23}^{AA}
( w_{23}\tilde\Gamma_{23}^{AB}
+ w_T(1,3)\tilde\Gamma_{13}^{AB} 
+\tilde{\Gamma}_{22}^{AB})
\underbrace{W_{\phi_0\phi_0}(B)}_{1}
\underbrace{
W_{\phi_0\phi_0}(A,B)}_{2} 
\underbrace{W_{\phi_0\phi_0}(A)}_{3}\right|_{A,B=0}
\end{align*}
The contribution due to 
$\tilde{\Gamma}_{12}^{BB}
\tilde{\Gamma}_{23}^{AA}\tilde{\Gamma}_{13}^{AB}$
is a one loop contribution formed joining non adjacent vertices. 
We have seen in  \eqref{eq:1loopnon adjacent} that it is negligible in the  in the Gross-Pitaevskii regime for large $N$ in the case of vanishing temperature. For finite temperature the same result holds because every $\beta$ which appears in the corresponding expression is multiplied by $K$. Indeed, rescaling the momenta, $K$ scales as $N^2$ in the GP limit and thus taking the limit $\beta\to\infty$ and then $N \to \infty$ is equivalent to taking directly the $N\to \infty$ limit. 

The contribution due to 
$\tilde{\Gamma}_{12}^{BB}
\tilde{\Gamma}_{23}^{AA}\tilde{\Gamma}_{22}^{AB}$
is also negligible. It corresponds to the insertion of terms proportional to the loop formed by $v$ and $\mathsf{G}^A_\beta$ which is negligible in the Gross-Pitaevskii regime when $N$ is sufficiently large.
Finally, the contribution 
$\tilde{\Gamma}_{12}^{BB}
\tilde{\Gamma}_{23}^{AA}\tilde{\Gamma}_{23}^{AB}$
is negligible because it corresponds to a loop similar to the one present in the Born approximation of the scattering length but with an insertion of an extra $\hat{v}$ $\mathsf{G}^A_\beta$. This extra insertion produces the sufficient decay to determine its negligibility in the limit of large $N$.
\\ 
For contributions of path graphs which have more than three vertices, we just need to use recursively the above argument. This concludes the proof.
\end{proof}

Lemma \ref{lm:adjacent 1PI loops are negligible} implies that the relevant contribution in the effective vertex $\hat{\mathcal{V}}$, which modifies also the two point function, is the one of the scattering length and that has already been computed above.

\section{Renormalisation and the Gross-Pitaevskii equation}\label{se:renormalisation}

\subsection{One point function}\label{se:one-point-function-renormalisation}

Having understood the necessary renormalisation in the two-point function of the fluctuations $\Psi$ and $\Psi^*$, we study in this section how the scattering length emerges on the background value of the fields through the ladder diagrams studied above. To this end we need to compute the one point functions of the $\Psi$ and $\Psi^*$. 
We work in the limit of low temperature ($\beta \to \infty$) and we compute how the one-point function behaves at lowest order in perturbation.
Since we are working at low temperature, as for the analysis of the fluctuations, 
we can discard the loop vertex $W_L$ and consider only $W_{J\phi_0}$, $W_{\phi_0 J^*}$ and $W_{\phi_0\phi_0}$. We are interested in
\[
\frac{1}{2\mathrm{i}}\langle{\Psi(x)+\Psi^*(x)}\rangle 
= \left. \frac{1}{2 \mathrm{i}}\left(\frac{\delta}{\delta {J}}+\frac{\delta}{\delta J^*}\right) \log Z(J,J^*) \right|_{J=J^*=0},
\]
that is going to be determined by just one vertex $W_{J\phi_0}$ or $W_{\phi_0 J^*}$, together with arbitrary number of vertices $W_{\phi_0\phi_0}$. At first order, only one vertex, no corrections arise. With two vertices, either $W_{J\phi_0}, W_{\phi_0\phi_0}$
or $W_{\phi_0 J^*}, W_{\phi_0\phi_0}$, we have for $\langle \Psi\rangle$
\[
\lim_{\beta\to\infty} 
\langle{\Psi}\rangle =  - \mathsf{G}_\beta(8\pi \mathfrak{a}_0-\hat{v}(0))\phi_0^3 +O,
\]
and similarly for $\langle\Psi^*\rangle$, where $\mathsf{G}_\beta$ is the fundamental solution of $-\partial_u+K_{\epsilon}$ and the infrared regulator $\epsilon$ is sufficiently large $\epsilon > 0$. Notice that, in virtue of the discussion at the end of the Appendix \ref{se:scattering-length}, this correction is relevant only for the Gross-Pitaevskii regime $b=1$.
To obtain this expression, we used the full analysis performed in section \ref{se:higher-loop-corrections} considering all ladder diagrams between $W_{\phi_0\phi_0}$ and either $W_{J\phi_0}$ or $W_{\phi_0J^*}$, together with the fact that they produce the scattering length $\mathfrak{a}_0$. 
The subtraction of $\hat{v}(0)/8\pi$ from the scattering length arises because the leg vertex $W_{\phi_0\phi_0}$ is treated as in in the proof of Lemma \ref{le:loops-legs} using the expression in \eqref{eq:Wphi0phi0}. 
Finally, the reminder $O$ accounts for the approximation of $\hat{\mathcal{V}}(p,n)$ with its value in $\hat{\mathcal{V}}(0,0)$. 

Discarding from here on that reminder $O$, and considering the fact that from the equations of motion $\tilde{\mu}\phi_0 = \hat{v}(0)\phi_0^3$, we have that the first correction to the background value of the field $\langle\Phi\rangle = \phi_0+\langle\Psi\rangle$ is  
\[
\phi_1 = \phi_0 - \mathsf{G}_\beta(8\pi \mathfrak{a}_0\phi_0^3 -\tilde{\mu}\phi_0 ).
\]
In order to repeat the analysis for higher orders, we consider only the loop diagrams among adjacent vertices of the tree. Namely, the diagrams that produce $\mathfrak{a}_{0}$. Moreover, we consider explicitly only tree graphs which have leg vertices $W_{\phi_0\phi_0}$ with at most three edges attached to it.
Then, similarly as at first order discarding the remainder, we get some recursive relations for the higher order corrections  
\[
\phi_{j+1} = \phi_0 - \mathsf{G}_\beta(8\pi \mathfrak{a}_0\phi_{j}^3 -\tilde{\mu}\phi_j).
\]
The fixed point of these recursive relations furnishes an estimate of the renormalised $\phi_R$ 
\[
\phi_{R} = \phi_0 - \mathsf{G}_\beta( 8\pi \mathfrak{a}_0 \phi_{R}^3 -\tilde{\mu}\phi_R),
\]
which satisfies the equation 
\[
(-\partial_u+K)\phi_{R} 
-\tilde{\mu} \phi_{R} + 8\pi \phi_R \mathfrak{a}_0 \phi_R^2 
= 
(-\partial_u+K)\phi_{0} .
\]
Since $\phi_0$ is constant, taking the limit $\epsilon\to0$, this equation implies that also the contributions constant in $u$ to $\phi_R$ satisfy the Gross-Pitaevskii equation with the scattering length $\mathfrak{a}_0$.
Consequently, the stationary solution for the renormalised background is
\[
\phi_R^2 = \frac{\tilde{\mu}}{8\pi \mathfrak{a}_0}.
\]

\subsection{Revisiting the convergence of the loop vertex expansion}
In order to revisit the convergence of the LVE to include the case in which $v$ is substituted to $v_{N,b}$ in \eqref{eq:vnb}, let us renormalise the background and the linear theory including the contributions of the scattering length. Effectively, this corresponds to a shift
\[
\phi_0\to\phi_R, \qquad \hat{v}(0) \to 8\pi \mathfrak{a}_0
\]
and thus
\[
\phi_0^2\hat{v}(0) \to  \phi_R^28\pi\mathfrak{a}_0.
\]
In the GP limit, $\mathfrak{a}_0$ scales as $1/N$, $\phi_R^2$ as $N$ and their product remains constant. 
If we reinterpret the loop vertex expansion with this renormalisation, we notice that a lot of graphs (up to vanishing errors as $N \to \infty$) are contained in $\phi_R$ and $\mathfrak{a}_0$. We expect this to simplify the analysis of convergence of the loop vertex expansion, enlarging its range of convergence at least in the limit of vanishing temperature. 
To see this explicitly we decompose the fields with respect to the renormalised background
\begin{equation}\label{eq:renormalised-background}
\Phi = \phi_R + X, \qquad 
\Phi^* = \phi_R + X^*
\end{equation}
where now $\phi_R$ is a (constant) solution of the Gross-Pitaevskii equation
\begin{equation}\label{eq:GR-ren-background}
-\frac{\Delta}{2m} \phi_R -\tilde\mu \phi_R + 8\pi \mathfrak{a}_0 \phi_R^3 =0.
\end{equation}
The Hamiltonian density is decomposed correspondingly
\[
\tilde{\mathcal{H}}=\tilde{\mathcal{H}}_{0R}+\tilde{\mathcal{H}}_{1R}+\tilde{\mathcal{H}}_{2R}+\tilde{\mathcal{H}}_{3R}+\tilde{\mathcal{H}}_{4R}
\]
where
\begin{align*}
\tilde{\mathcal{H}}_{0R}&= \phi_R({K}-\tilde\mu)\phi_R + 2\phi_R^2 v *(\phi_R^2 )    
-\frac{1}{2} \phi_R^2(\hat{v}(0)-8\pi\mathfrak{a}_0)\frac{1}{v}(\hat{v}(0)-8\pi\mathfrak{a}_0)\phi_R^2     
\\
\tilde{\mathcal{H}}_{1R}&= 
(X+X^*)\left( ({K}-\tilde\mu)\phi_R + 8\pi\mathfrak{a}_0 \phi_R^3 \right)
\\
\tilde{\mathcal{H}}_{2R}&= X^* {K}_{\epsilon} X  + \frac{1}{2}  \phi_R (X+X^*) v* (\phi_R(X+X^*))\\
\tilde{\mathcal{H}}_{3R}&= :|X^2|:_{\phi_R} v* ((X+X^*)\phi_R)
\\
\tilde{\mathcal{H}}_{4R}&= \frac{1}{2} :|X^2|:_{\phi_R} v* :|X^2|:_{\phi_R}.
\end{align*}
Here, ${K}_{\epsilon}=-\frac{\Delta}{2m}+\epsilon$ and the fact that the renormalised background $\phi_R$ solves the Gross-Pitaevskii equation, rather than the equation with $v$, has been compensated with a redefinition of the normal ordering
\[
:|X^2|:_{\phi_R}=:|X^2|:_\beta+\frac{2}{v} (\hat{v}(0)-8\pi\mathfrak{a}_0)\phi_R^2.
\]
Here we are assuming that the operator $f\to v*f$ is invertible, namely $\hat{v}(p)>0$ for every $p \in \mathbb{R}^3$. In this way, $\tilde{\mathcal{H}}_{1R}$ is again vanishing.\\

As before, we introduce an auxiliary Hubbard-Stratonovich field $A(u)$ to reduce the order of the interaction Hamiltonian in the fields $X$
\begin{equation}\label{eq:renormalised-decomposition}
\tilde{\mathcal{H}}_A = 
\tilde{\mathcal{H}}_{20R} 
+Q_{AR}
\end{equation}
where
\[
\tilde{\tilde{\mathcal{H}}}_{20R} = X^* {K}_{\epsilon} X 
,
\qquad
Q_{AR} = g A 
 :|X|^2:_{\phi_R} +g  A\phi_R(X+X^*).
\]
We apply now the same technique used in section \ref{se:LVE} 
to get the generating function of the connected correlation functions of the theory. More precisely, to obtain $\log Z(J,J^*)-\log Z(0,0)$ we start from \eqref{eq:logZ-JJ-Gc} where the updated form of the vertices $V=e^{-\Gamma}W$ is now 
\begin{equation}\label{eq:WREN}
W = -2A\frac{1}{v}(\hat{v}(0)-8\pi\mathfrak{a}_0)\phi_R^2 - 2\langle \mathrm{i} J^*+A\phi_R,\mathsf{G}^{A}_{\beta}(\mathrm{i} J+A\phi_R)\rangle + 2\Tr \int_0^1 ds (\mathsf{G}^{sA}_{\beta}-\mathsf{G}_{\beta}) A
\end{equation}
which is similar to the expression given in \eqref{eq:vertexW}, up to the contributions arising from the change in the normal ordering. We can now study the tree graph expansion of $\log Z(J,J^*)$ obtained as in \eqref{eq:LVE}, performing all the partial resummation necessary to get
\eqref{eq:ell-scattering}.
It is important to stress what do we precisely mean by that. Namely, in a first step, we sum the expansion over simple connected graphs among $n$ vertices producing $\tilde{\mathsf{G}}^A_{\mathcal{L}}$ in the line vertices. 

From here on, $\tilde{\mathsf{G}}^A_{\mathcal{L}}$ is the propagator studied in appendix \ref{se:pretrubed-propagators}, the inverse of the operator
\begin{equation*}
\begin{pmatrix}
    &-\partial_u+K+A &0\\
    &0 &\partial_u+K+A
\end{pmatrix}
\end{equation*}
perturbed by $\mathcal{L}^\beta \mathsf{J}_2$, for $\mathsf{J}_2$ the $2\times 2$ matrix with entries equal to $1$ and $\mathcal{L}^\beta$ given implicitly in 
\eqref{eq:ell-scattering}. Namely
\begin{equation*}
\begin{pmatrix}
    &\mathcal{L}^\beta &\mathcal{L}^\beta\\
    &\mathcal{L}^\beta &\mathcal{L}^\beta
\end{pmatrix}.
\end{equation*}
In this way, $\tilde{\mathsf{G}}^A_{\mathcal{L}}$ is a quadratic form on $(J,J^*)\times (J^*,J)$, see \eqref{eq:two-point-function} for its precise form\footnote{Although we are working at vanishing external momenta, a similar formula holds also if $\hat{\mathcal{L}}^\beta(p,n)$ depends non trivially on $n$.}.
Recall that $\hat{\mathcal{L}}^\beta (0,0)=\frac{\phi_R^2}{\phi_0^2}\hat{\mathcal{V}}(0,0)$ and 
\begin{equation}
\lim_{\beta \to\infty} \hat{\mathcal{L}}^\beta(0,0) = \phi_R^28\pi \mathfrak{a}_0.
\end{equation}
This partial sum has already been explicitly discussed in appendix \ref{se:pretrubed-propagators}, used first in \eqref{eq:omegaJJ1} and later throughout section \ref{se:higher-loop-corrections} in the limit $\beta\to\infty$. Here, it produces line vertices which now have $\phi_R$ as extreme points and can be given in terms of $\mathsf{M} = \sum_{ij} \tilde{\mathsf{G}}^A_{\mathcal{L} ij}$ and $\mathsf{F}_i=\sum_{j} \tilde{\mathsf{G}}^A_{\mathcal{L} ij}$ as
\begin{equation}\label{eq:resummed-vertices}
\begin{aligned}
W'_{\phi_R\phi_R} &= -2 
\langle A \phi_R, \mathsf{M} \phi_R A\rangle +2 A \phi_R^2v^{-1}{(\hat{v}(0)-8\pi \mathfrak{a}_0)}  
\\
W'_{\phi_R J} &= -2 
\langle A \phi_R, \mathsf{F}_2^\dagger
\mathsf{i} J \rangle   
\\
W'_{J^*\phi_R } &= -2 
\langle \mathsf{i}J^*, 
\mathsf{F}_1
 \phi_R A\rangle   
\\
W'_{J(i)J(j)} &= -2 
\langle \mathsf{i}J(i), 
\tilde{\mathsf{G}}^A_{\mathcal{L} ij}
\mathsf{i}J(j)
\rangle   
\end{aligned}
\end{equation}
where in the last expression $J(1) = J$ and $J(2)=J^*$. The nice feature of this step is seen in the better infrared behaviour of $\mathsf{M}$ and $\mathsf{F}$ compared to $\mathsf{G}^A_\beta$. Actually, the $L^1$ norm of the functions associated to both $\mathsf{M}$ and $\mathsf{F}_i$ are bounded by $2(16\pi\phi_R^2\mathfrak{a}_0)^{-1}$.\\

Let us discuss in the next Theorem how this partial resummation is performed. Namely, basing the proof on results proven later in Lemma \ref{le:loop-LVE}, we discuss the equivalence of the LVE after performing the partial sum and the original LVE given in \eqref{eq:LVE}. This comparison is done at fixed $\epsilon > 0$ in $K_\epsilon$ and for large $N$.

\begin{theorem}\label{thm:partial-sum}
Consider in the various propagators of the loop vertex expansion \eqref{eq:LVE} the rescaled interaction $v_N$, at the place of $v$, and 
$K_\epsilon$ with $\epsilon>0$.
At fixed $\epsilon>0$, the loop vertex expansion after the partial resummations, namely with $W'_{ij}$ in \eqref{eq:resummed-vertices} at the place of the corresponding $W_{ij}$, is at any order equivalent to the original loop vertex expansion in the limit $N\to\infty$.
\end{theorem}
\begin{proof}
We discuss in details the partial resummation which produces $W'_{\phi_r\phi_R}$, the other vertices are treated analogously.
Consider a tree graph $T$ with a single contribution $W_{\phi_R\phi_R}$
and denote by $\mathcal{F}$ the family (forest) of similar tree graphs obtained substituting the vertex
$W_{\phi_R\phi_R}$ with all possible path graphs formed with $n$ arbitrary many vertices $W_{\phi_R\phi_R}$. 
If in $T$ multiple edges reach the vertex $W_{\phi_R\phi_R}$, $\mathcal{F}$ has all possible graphs obtained distributing these edges over the added vertices. 

Next, we proceed as in the proof of Lemma \ref{lm:adjacent 1PI loops are negligible}.
Namely, in every $T'\in\mathcal{F}$ 
we change the name of the auxiliary field $A$ into $B$ in the vertices of the path graphs which form the forest. 
We also double the measure
\[
\Gamma = \Gamma^{AA}+\Gamma^{AB}+\Gamma^{BB}.
\]
Hence, the covariance of these new fields is the tensor product of $\Upsilon\otimes R$ where $R=\mathsf{J}_2$ is a $2\times 2$ matrix whose entries are all equal to $1$. 

We change from $1$ to $s\in[0,1]$ the off diagonal elements, denote the resulting covariance by $\Upsilon \otimes R^s
$
and by ${\tilde{\Gamma}}^s$ the associated differential operator.
Then, we introduce
\[
C_\mathcal{F}(s) = 
-\sum_{T\in\mathcal{F}}\frac{
1}{n(T)!}
\mathcal{M}\prod_{b\in E_T}
\left(\int _0^1 dw_b\tilde\Gamma_{b}\right)e^{-\sum_{i<j }w_T(i,j) {\tilde\Gamma}^s_{ij}} 
\prod_{l=1 }^{\otimes n(T)}e^{-\Gamma^s} W_{l}
\]
where $n(T)$ is the cardinality of the set of vertices $V(T)$ of the graph $T$ and $W_l$ is a notation to indicate the vertices at the place $l$ of the graph. We notice that $C_\mathcal{F}(1)=C_\mathcal{F}$ which is the contribution 
in the original LVE in \eqref{eq:LVE}. Instead, in $C_\mathcal{F}(0)$ the vertices of the inserted path graphs at the place of $W_{\phi_R\phi_R}$ are connected to the vertices of the original graph only through $\prod_b \tilde{\Gamma}_b$. In other words, in $C_\mathcal{F}(0)$ there are no loops formed between $A$ and $B$ vertices. Therefore, the evaluation of $B$ on the corresponding Gaussian state can be taken and the sum over the forest $\mathcal{F}$ realises the partial sum.\\

We now estimate the remainder, 
\[
C_\mathcal{F}(1)-C_\mathcal{F}(0) =\int_0^1 ds \frac{d C_\mathcal{F}(s)}{ds}.
\]
The improved covariance $\Upsilon\otimes R^{s}$ is positive, so Corollary \ref{le=estimate1} is used
to estimate $|C_\mathcal{F}(1)-C_\mathcal{F}(0)|$ and to bound the action of the various exponentials of $\tilde{\Gamma}^s_{ij}$ and $\Gamma$ by $1$.
It remains to estimate the contributions of tree graphs with one extra loop formed joining one of the component of the path graphs with the other vertices. 
This is estimated as in Lemma \ref{le:loop-LVE}, according to which the contributions in the remainder vanish in the limit $N\to\infty$ keeping $\epsilon$ fixed. We have thus proven that at fixed $\epsilon$, 
\[
\lim_{N\to\infty}
|C_\mathcal{F}(1)-C_\mathcal{F}(0)| = 0
\]
and in that limit $C_\mathcal{F}(1)$ is equivalent to $C_\mathcal{F}(0)$ where the partial sum producing the updated propagators $W'_{\phi_R\phi_R}$ has been taken.
\end{proof}

The technical Lemma used in the proof of the previous Theorem is stated and proven in the following.

\begin{lemma}\label{le:loop-LVE}
Consider $v_N$ at the place of $v$ and 
$K_\epsilon$ with $\epsilon>0$ in the various propagators of the loop vertex expansion \eqref{eq:LVE}.
Let $T$ be a path graph among $n+1$ vertices
and let $O_T$ be a one loop contribution obtained after inserting an extra factor $\Gamma$. Moreover, assume that it has $e$ external momenta $\{p_1,\dots, p_e\}$ and that it does not involve only two vertices $W_{\phi_R\phi_R}$. Then, $O_T$ is
\[
O_T(p_1,\dots p_e) 
= 
\prod_{k=1}^e
\left\langle e^{\mathrm{i} p_k(\cdot)},\frac{\delta}{\delta A(\cdot)}\right\rangle
\mathcal{M}\prod_{b\in E_T}
\left(\int _0^1 dw_b\tilde\Gamma_{b}\right)
\tilde\Gamma_{1(n+1)}
e^{-\sum_{i<j }w_T(i,j) {\tilde\Gamma}_{ij}} 
\prod_{l=1 }^{\otimes (n+1)}e^{-\Gamma} W_{l}
\]
where $\{W_l\}_{l\in\{1,\dots,n+1\}}$ is a sequence of $W_l\in\{W_L,W_{J\phi_R},W_{\phi_RJ^*},W_{\phi_R,\phi_R}\}$ such that $O_T$ is proportional to $\phi_R^{2t}$ with $t\leq n$.
It holds that 
for large  $N$ and for a positive $c\in [0,\frac{1}{2})$ it exists a constant $C_c$ such that
\[
\|O_T\|_\infty\leq \frac{C_c}{\epsilon^{e+n-1-c}} \|\hat{v}\|_\infty^{n-1-c} \left(\frac{\phi_R^{2}}{N}\right)^{t}\frac{1}{N^{n+2c-t}},
\]
with
\[
C_c = C_T(J,J^*)\int d^3p \frac{\hat{v}(p)^{2+c}}{\hat{K}_0(p)^{1+c}}
\]
and $C_T$ a suitable constant depending also on $J$.
Hence, since $\phi_R^2$ scales as $N$, in the limit $N\to\infty$ keeping $\epsilon > 0$ fixed, $O_T$ vanishes.
Finally, the same bound holds also if in $O_T$ the terms $e^{-\Gamma}$ and $e^{-\sum_{i<j} w_T(i,j)\tilde{\Gamma}_{ij}}$ are substituted by $1$.
\end{lemma}

\begin{proof}
Recall that the limit $N\to\infty$ of the limit $\beta\to \infty$ is equivalent to the limit $N\to\infty$ as already discussed in Lemma \ref{lm:adjacent 1PI loops are negligible}. Actually, for large $N$, the rescaled $f(\beta N^2\hat{K}_0(p))$ behaves as $f(\beta \hat{K}_0(p))$ for large $\beta$. Moreover, notice that after taking all the functional derivatives in $O_T$, the action of $e^{-w_T \tilde{\Gamma}^s}$ and $e^{-\Gamma^s}$ can be discarded as its action is bounded by $1$. This is a consequence of the Feynman-Kac formula, see
\cite{GalandaPinamonti} or Corollary \ref{le=estimate1} and Lemma \ref{le:positive-matrices}.\\

Let us start considering the case $e=0$, namely, the case without external momenta. We start studying the contributions involving only vertices $W_{\phi_R \phi_R}$, which is part of any loop diagram involving all kind of vertices, of the form 
\[
O'_T=\lim_{\beta\to\infty}
\frac{1}{\beta}\sum_{h\in\mathbb{Z}}
\int d^3p \hat{v}_N(p)^{n+1}\hat{\mathsf{G}}_\beta(p,h)^{w} \hat{\mathsf{G}}_\beta(-p,-h)^{n+1-w} \phi_R^{2 t}
\]
where among the $n+1$ propagators $\mathsf{G}_\beta$ in $O'_T$, $w$ of them have positive momenta and $n+1-w$ negative momenta. In the limit $\beta\to \infty$ of $O'_T$ and
operating as in the proof of Lemma \ref{le:vanishing-crossed-ladder} to turn the sum over Matsubara frequencies in an integral, we have that
\begin{align*}
|O_{T,\infty}'| &:= \lim_{\beta\to\infty}|O_T'| \\
&=  
\frac{\phi_R^{2 t}}{(w-1)!(n-w)!}\int d^3p \hat{v}_N(p)^{n+1}\int_{-\infty}^{\infty} ds 
\int_0^\infty du_+ u_+^{w-1}
\int_0^\infty du_- u_-^{n-w}
e^{-(\hat{K}(p)+\mathrm{i} s) u_+}
e^{-(\hat{K}(p)+\mathrm{i} s) u_-}.
\end{align*}
The integration over $s$ gives a Dirac delta $\delta(u_+-u_-)$. Then, performing the integration over $u_-$ and $u_+$ we obtain
\begin{align*}
|O_{T,\infty}'| 
&=  
\phi_R^{2 t}\frac{n!}{(w-1)!(n-w)!}
\int d^3p \hat{v}_N(p)^{n+1} \frac{1}{(2\hat{K}(p))^n}.
\end{align*}
Consider now $c\in [0,\frac{1}{2})$ and recall that $\hat{v}_N=\hat{v}(p/N)/N$ so $\|N\hat{v}_N\|_\infty=\|\hat{v}\|_\infty$. The following bound is obtained
\begin{align*}
|O_{T,\infty}'| 
&\leq  
\phi_R^{2 t}\frac{n!}{(w-1)!(n-w)!}
\frac{1}{2^n}
\frac{\|\hat{v}_N\|_\infty^{n-1-c}}{(\epsilon)^{n-1-c}}
\int d^3p \hat{v}_N(p)^{2+c} \frac{1}{(K_0(p))^{1+c}} 
\\
&\leq  
\frac{1}{(\epsilon)^{n-1-c}}
N^{t-n-2c}
\tilde\phi_R^{2 t}  \frac{n!}{(w-1)!(n-w)!}
\frac{1}{2^n}
\|\hat{v}\|_\infty^{n-1-c}
\int d^3p \hat{v}(p)^{2+c} \frac{1}{(K_0(p))^{1+c}} 
\end{align*}
where ${\phi}_R^2 = N\tilde{\phi}_R^2$. Summing over all possible orderings in which the functional derivates appear, corresponds to a sum over all possible $w$. In this sum, the factors $n!/((w-1)!(n-w)!2^n)$ give $n/2$. Therefore, a bound similar to the one given in the thesis follows and $\lim_{N\to\infty} |O_{T,\infty}'| $ vanishes.\\

We now observe that the estimate of $O_T'$ can be used to bound the generic $O_T$.
Indeed, considering additional derivatives contributes with factors $1/
\epsilon$ which do not change the analysis.
Also the presence of the loop vertices does not change the analysis. Indeed, $W_L$ is treated as in Lemma \ref{le:loops-legs} giving an extra factor $1/\sqrt{\beta}$ for every $W_L$ and the $s$ derivatives applied to $W_L$ give $(\|v\|_1/\epsilon)^{s-2}$ extra factors. Thus, the contributions of loop vertices $W_L$ are bounded relying on the same above analysis.
Finally, considering $W_{J\phi_R}$ or $W_{\phi_R J^*}$ changes the constant in front of the estimate but not the overall bound.
\end{proof}

After this preliminary discussion and after having proved that for fixed $\epsilon$ the partial sum producing the updated propagators in \eqref{eq:resummed-vertices} is equivalent to the original LVE in the limit for large $N$, we can state the main result of the section.

\begin{theorem}\label{thm:renormalisation-convergenve}
Consider the loop vertex expansion of $\log Z(J,J^*) - \log Z(0,0)$ given in \eqref{eq:LVE}, with vertices as in \eqref{eq:WREN} and obtained for the decomposition \eqref{eq:renormalised-background} of the field with respect to the renormalised background $\phi_R$, which  
is a constant solution of the Gross-Pitaevskii equation \eqref{eq:GR-ren-background}.
Take the partial resummations of the tree graphs which produce everywhere in the edge vertices $\tilde{\mathsf{G}}_{\mathcal{L}}^A$ as discussed in Theorem \ref{thm:partial-sum} and thus consider directly the renormalised vertices \eqref{eq:resummed-vertices}. After renormalising the edges between $W_L$ and other vertices, the remaining tree graphs converge absolutely in the limit $\beta\to\infty$, taken before the sum, if $N$ is sufficiently large for $\epsilon\sim N^{-1}$ and $\mathfrak{a}_0 \phi_R^2$ kept constant. 
\end{theorem}
\color{black}
\begin{proof}
The proof follows closely the strategy used for Theorem \ref{thm:converge} up to some partial resummations and changes in the estimate of certain contributions. 
We discuss in particular the role of the vertex $W'_{\phi_R\phi_R}$ when it appears in a generic tree graph $T$ depending on its valence
\begin{itemize}
    \item If its valence is $1$ (i.e.~a single edge reaches it), its contribution can be shown to decay as $1/N$ thanks to the new normal ordering and Lemma \ref{le:one-point}. 
    \item If $W'_{\phi_R\phi_R}$ appears with valence $2$ there are two options. First, the vertex is attached only to $W'_{\phi_R\phi_R}$, to $W'_{J\phi_R}$ or to $W'_{\phi_RJ^*}$. Then, because of the partial resummation, the corresponding vertex contains $\mathcal{L}_\beta$. Due to Lemma \ref{le:reminder-resummation}, this remainder is negligible in the limit $\beta\to\infty$ (taken before the sum) and for $N\to\infty$. Instead, if the vertex is attached only to $W_L$, we can resum the corresponding graphs with a similar graph entering the LVE. This is the graph which contains only $v$ at its place. Therefore, what we are actually considering, is a graph where the following redefinition occurs 
    \[
        v\to v-v \frac{1}{K+16\pi \phi_R^2 \mathfrak{a}_0 } v.
    \]
    Thus, we may just consider this second graph with the new edge. However, the latter scales as $v_N$ for large $N$. 
    \item If the valence is $n>2$ we have that 
    \[
    \left|\frac{\delta^n}{\delta A\dots \delta A}W'_{\phi_R\phi_R}\right| \leq \frac{1}{N^{n-2}},
    \]
    as $\mathsf{M}$ and $\mathsf{F}$ and $\mathsf{F}^\dagger$ are $L^1$ bounded.  
\end{itemize}
Taking into account these cases, we estimate the contribution of a generic tree graph. The remaining loop vertex expansion is controlled by a power series similar to the one in \eqref{eq:power-series} in the proof of Theorem \ref{thm:converge}
\begin{align*}
    |\log Z(J,J^*)-\log Z(0,0)|\leq & |V_1|+ \beta  (\|J\|_{1}+\|J\|_{K,1}+\|J^*\|_{1}+\|J^*\|_{K,1})\left(2\frac{\|v\|_1}{\epsilon}\right)^{-1}
    &\sum_{n\geq 2} r^n  \sum_{T\in \mathcal{T}_n}e(1)p_T
\end{align*}
where now
\[
 \frac{r}{4}=\left(\frac{\|v\|_1}{\epsilon} \right) \left( \frac{C}{\sqrt{\beta}}+\left(\frac{\|8\pi \mathfrak{a}_0\phi_R\|_\infty}{\sqrt{N}}+\frac{\|J\|_{K,1}+\|J^*\|_{K,1}}{2}\right)^2 \right)
\]
and by Lemma \ref{le:EGF-convergence}  the series converges if $r< 1/4$. 
We notice that now, thanks to the extra factor $1/\sqrt{N}$, the bound can be made valid also in the Gross-Pitaevskii regime. 
\end{proof}

Finally, we present in the following two technical Lemmata used in the proof of the previous Theorem.

\begin{lemma} \label{le:reminder-resummation}
In the limit of vanishing temperature we have that the remainder of the partial resummation which produces $W'_{\phi_R\phi_R}$ vanishes. Namely,
considering the remainder at the order $n$ contribution in the LVE giving $W_{\phi_R\phi_R}'$
\begin{align*}
\mathcal{R}_n=\prod_{i=1}^{n-1}\int_0^1 dw_{i,i+1}
&\left(
\tilde{\Gamma}_{i,i+1}
e^{-\frac{1}{2}\sum_{j\neq k}w_T(j,k)\tilde\Gamma_{jk}}
\underbrace{e^{-\Gamma}W_{\phi_R\phi_R}\otimes \dots \otimes e^{-\Gamma}W_{\phi_R\phi_R}}_{n} 
+
\right.
\\
&-\left.
\left\langle
\frac{\delta}{\delta  A_i},
\mathcal{L}^\beta
\frac{\delta}{\delta  A_{i+1}}\right\rangle 
\underbrace{W_{\phi_R\phi_R}\otimes \dots \otimes W_{\phi_R\phi_R}}_{n} 
\right),
\end{align*}
where $w_T(j,k) = \inf \{w_{jl} | l\in\{j+1,k\}\}$, we have that for interaction $v_N$ in $\tilde{\Gamma}$ and for $N\to\infty$
\begin{equation*}
    \mathcal{R}_n=O(N^{-1}).
\end{equation*} 
\end{lemma}
\begin{proof}
To prove this lemma we proceed as follows. At order $n$, $\mathcal{R}_n$ is the remainder to the resummation that has produced $W'_{\phi_R\phi_R}$. Moreover, using $v_N$ for large $N$ some contributions are immediately negligible. Actually
\[
\mathcal{R}_n 
=
\prod_{i=1}^{n-1}\int_0^1 dw_{i,i+1}
\left(
\tilde{\Gamma}_{i,i+1}
e^{-w_{i,i+1}\tilde\Gamma_{i,i+1}}
-\left\langle
\frac{\delta}{\delta  A_i},
\mathcal{L}^\beta
\frac{\delta}{\delta  A_{i+1}} \right\rangle
\right)
\underbrace{W_{\phi_R\phi_R}\otimes \dots \otimes W_{\phi_R\phi_R}}_{n} 
+ \, O(N^{-1}),
\]
essentially because $e^{-\Gamma}$ adds exponential factors which vanish in the limit of large $N$. 
Furthermore, as proved in Lemma \ref{le:power-counting},
the ladder diagrams with multiple rails that are 1PI gives extra factors $1/N^{2l-2}$ where $l$ is the number of rails. Hence, in the evaluation of the remainder, only the  multiple ladder diagrams which are completely reducible to product of ladder diagrams with two rails survive. These corresponds exactly to products of $\mathcal{L}^\beta$, see also Lemma \ref{lm:adjacent 1PI loops are negligible}. Furthermore, in the limit where $\beta\to\infty$, we have, from Lemma \ref{le:vanishing-crossed-ladder}
and remark \ref{rem:leg}, that only the plain ladder diagrams survive in $\mathcal{L}^\beta$ and thus $\mathcal{R}_n = O(N^{-1})$.
\end{proof}

\begin{lemma}\label{le:one-point}
With $v_N$ in $\Gamma$ and $\tilde{\Gamma}$, it holds that in the limit of vanishing temperature
\begin{equation}\label{eq:rest-renormalised}
|(e^{\Gamma_{12}}-1)
\mathsf{G}_\beta^A \otimes W'_{\phi_R\phi_R}
|
\leq \frac{C}{N}
\end{equation}
\end{lemma}
\begin{proof}
Before taking the action of $e^{\tilde{\Gamma}_{12}}-1$, we operate as in \eqref{eq:Wphi0phi0} to rewrite part of $W'_{\phi_R\phi_R}$. We have that  
\[
\langle A \phi_R, \mathsf{M}^A \phi_R A\rangle 
=
   \langle
   \phi_R, (K+2 \mathcal{L}^\beta) \left(\mathsf{M}^{A}
   -
\mathsf{M}\right)
(K+2 \mathcal{L}^\beta)
\phi_R \rangle
   -
   \langle
   \phi_R, A \mathsf{M} 
   (K+2 \mathcal{L}^\beta)
   \phi_R \rangle.
\]
When we apply $e^{\tilde{\Gamma}_{12}}-1$ with $v_N$ in $\tilde{\Gamma}$ we may proceed as in 
\eqref{eq:ell-scattering}. 
We get that $\mathsf{M}$ at the place of $\mathsf{G}_\beta^A$, introduces a negligible error. Hence
\[
|(e^{\Gamma_{12}}-1)
\mathsf{G}_\beta^A \otimes W'_{\phi_R\phi_R}|
\leq C
|
(\mathcal{L}^\beta - v)-(\phi_R^2)(\hat{v}(0)-8\pi \mathfrak{a}_0)| + O(N^{-1}),
\]
where the remainder is due to the fact that we have here used 
$\mathsf{G}_\beta^A$ at the place of $\mathsf{M}$.
Furthermore, to estimate the second contribution, we used that 
 $v*\frac{\delta}{\delta A} 
  \langle
   \phi_R, A \mathsf{M} 
   (K+2 \mathcal{L}^\beta)
   \phi_R \rangle
 =
\phi_R^2 \hat{v}(0)$. Finally, in the limit of vanishing $p,n$ and $\beta \to\infty$, we have that $\mathcal{L}^\beta\to 8\pi{\mathfrak{a}_0}$
from which the thesis follows.
\end{proof}
\appendix

\section*{Acknowledgments}
The authors want to thank Chiara Boccato, Serena Cenatiempo and Benjamin Schlein for illuminating discussions on BEC.
The authors thank also Klaus Fredenhagen for many discussions and insights about the construction of equilibrium states for interacting fields and Edoardo D'Angelo for discussions on the loop vertex expansions. 
The authors are grateful for the support of the National Group of Mathematical Physics (GNFM-INdAM). 
The work of S.G. is funded by the EPSRC Open Fellowship EP/Y014510/1 PI Benito A Juárez Aubry.
The research presented in this paper was supported in part by the MIUR Excellence Department Project 2023-2027 awarded to the Department of Mathematics of the University of Genova, CUPD33C23001110001.
\\

{\bf Data availability:} 
Data sharing not applicable to this article as no datasets were generated or analysed during
the current study. 
\\

{\bf Conflicts of interest statement:} 
All authors have no conflicts of interest to declare.

\section{Thermal propagator and its perturbations} \label{se:thermal-propagators-appendix}

\subsection{Thermal propagator}

Recall that the two-point function of the thermal state \eqref{eq:Bpm} is 
\begin{align*}
\mathsf{B}^\beta_-(x,y)=\omega_\beta(\phi(x) \phi^*(y))  =\frac{1}{(2\pi)^3}\int  \frac{1}{1-e^{-\beta w(p)}} e^{\mathrm{i}p(x-y)} d^3p
\\
\mathsf{B}^\beta_+
(x,y)=\omega_\beta(\phi^*(x) \phi(y)) = -\frac{1}{(2\pi)^3}\int  \frac{1}{1-e^{\beta w(p)}}e^{\mathrm{i}p(x-y)}  d^3p,
\end{align*}
where
\[
w(p) = \frac{p^2}{2m}-\mu.
\]
Hence, the operators associated to the integral kernels are given by
\[
\mathsf{B}^{\beta}_-(f,h) =\langle \bar{f}, B_{- } h \rangle, \qquad \mathsf{B}^{\beta}_+(f,h) =\langle \bar{h}, B_{+ } f\rangle,  \qquad  B_{\mp} = \pm \frac{1}{1-e^{\mp\beta K}},
\]
where the operator $K$ is multiplicative in the Fourier domain and it is such that its Fourier transform is $w(p) = \frac{p^2}{2m}-\mu$. Notice that for $0<u<u'<\beta$ the following are well defined operations 
\begin{equation}
\label{eq:thermal-1}
\omega_\beta(
\tau_{\mathrm{i}u} \Phi(f) 
\tau_{\mathrm{i}u'} \Phi^*(h))
=\langle \bar{f},e^{-(u'-u)K }B_{-} h\rangle, \qquad 
\omega_\beta(
\tau_{\mathrm{i}u} \Phi^*(h)
\tau_{\mathrm{i}u'} \Phi(f) )
=\langle \bar{f},e^{(u'-u)K }B_{+} h\rangle
\end{equation}
and 
\[
e^{-u K} B_- =  e^{(\beta-u) K} B_+ .
\]
Motivated by this observation, we also introduce the {\bf thermal propagator} which is defined first on $u\in [0,\beta)$ and then extended by periodicity
\[
\mathsf{G}_{\beta}(u) := \frac{e^{- u K}}{1-e^{-\beta K}}, \qquad u\in [0,\beta), \qquad \mathsf{G}_{\beta}(u+\beta)=\mathsf{G}_{\beta}(u),  \qquad u\in\mathbb{R}.
\]
To be precise, $\mathsf{G}_\beta$ is a periodic function with values in the linear operators on $\mathfrak{H}$. Notice that $\mathsf{G}_\beta(u)$ is a continuous function into the set of bounded operators for $u\neq \mathbb{N} \beta$, while it is discontinuous in $0$ such that 
\[
\lim_{u\to 0^+} \mathsf{G}_\beta(u) = B_-, \qquad 
\lim_{u\to 0^-} \mathsf{G}_\beta(u) = B_+.
\]
Namely, the operators $B_\pm$ can be recovered from $\mathsf{G}_\beta(u)$ and, for $u<u'$, the operators appearing in \eqref{eq:thermal-1} are obtained from $\mathsf{G}_\beta$ as follows
\begin{equation}\label{eq:operators-thermals-relations}
\mathsf{G_\beta}(u'-u)=B_- e^{-(u'-u)K}
,
\qquad
\mathsf{G_\beta}(u-u')=B_+ e^{(u'-u)K}.
\end{equation}

Even if $\mathsf{G}_\beta$ is defined as a functions with values in the set of linear operators on 
$\mathfrak{H}$,
we can extend $\mathsf{G}_\beta$ as an operator with a suitable domain on  $\mathfrak{H}_\beta = L^2(\mathbb{R}^3\times [0,\beta])$ by defining its action  on $\psi\in\mathfrak{H}_\beta$ as
\[
\mathsf{G}_{\beta} \psi(u)=
\int \frac{e^{-(u'-u)K}}{1-e^{-\beta K}}  \psi (u') du'.
\]
Let us also introduce the corresponding two-point function using the Schwarz kernel theorem
\begin{equation}\label{eq:kernel-thermal}
\langle \psi_1,\mathsf{G}_\beta \psi_2 \rangle_\beta
=
\int \psi_1(u)(\mathsf{G}_{\beta}\psi_2)(u)   du
= 
\int \psi_1(u) \frac{e^{-(u'-u)K}}{1-e^{-\beta K}}  \psi_2(u') du' du
\end{equation}
and thus using the same symbol we denote the kernel as $\mathsf{G}_\beta(x,u;x',u')$.

We observe that $\mathsf{G}_\beta$ is not even formally selfadjoint. Indeed, the convolution kernel of $\mathsf{G}_\beta^\dagger$ 
is such that
\[
\mathsf{G}_\beta^\dagger
(u) = \mathsf{G}_\beta(-u).
\]
Furthermore, the thermal propagator is the fundamental solution of
\[
(-\partial_u + K)\int \mathsf{G}_{\beta}(u'-u) {\psi}(u') du' = {\psi}(u) 
\]
as
\begin{align*}
(-\partial_u+K)\int \mathsf{G}_\beta(u'-u) {\psi}(u') du' 
&= 
(-\partial_u+K)\left( \int_0^u \mathsf{G}_{\beta}(u'-u) {\psi}(u') du' + \int_u^\beta \mathsf{G}_{\beta}(u'-u) {\psi}(u') du' \right) \\
&= 
(-\mathsf{G}_{\beta}(0^{-})+\mathsf{G}_{\beta}(0^{+})) \psi(u) = \psi(u)
\end{align*}
The Fourier transform of $\mathsf{G}_\beta(u)$ in $u$ is
\[
\hat{\mathsf{G}}_\beta(s_n)=\int_0^\beta  \frac{e^{- u K}}{1-e^{-\beta K}} e^{\mathrm{i}u s_n} du = \frac{1}{ K - \mathrm{i} s_n}
\]
where $s_n = \frac{2\pi}{\beta} n$ are {\bf the Matsubara frequencies}.
Here we used the convention in which the normalization $1/\beta$ is in the inverse Fourier transform.

\subsection{Perturbation of the thermal propagators}\label{se:pretrubed-propagators}
Consider the operators
\[
T = \begin{pmatrix}
-\partial_u +K & 0\\
0 & \partial_u +K 
\end{pmatrix}
,
\qquad
V =  \begin{pmatrix}
\tilde{\mathsf{v}} &  \tilde{\mathsf{v}}\\
\tilde{\mathsf{v}} & \tilde{\mathsf{v}}
\end{pmatrix}
\]
as acting on $(\Psi,\Psi^*)$ in $\mathfrak{H}_\beta \oplus\mathfrak{H}_\beta$, and where
\[
\tilde{\mathsf{v}} \Psi (u,x)  =\int du'\int d^3x'\delta(u-u') \mathsf{v}(x-x')\Psi(u', x').
\]
$T$ describes the free equation of motion satisfied by the fields, and $V$ is a perturbation of them. Although we restrict the analysis of this appendix to $\tilde{\mathsf{v}}$ of the form given above, the very same results hold for cases where $\tilde{\mathsf{v}}$ is a non trivial convolution also in $u$.
In Fourier domain we have
\[
\hat{T} = \begin{pmatrix}
-\mathrm{i} s_n+K & 0\\
0 & \mathrm{i} s_n+K 
\end{pmatrix}
,
\qquad
\hat{V} =  \begin{pmatrix}
{\mathsf{v}} &  {\mathsf{v}}\\
{\mathsf{v}} & {\mathsf{v}}.
\end{pmatrix}
\]
Notice that, the the determinant of their sum is
\[
{\det(\hat{T}+\hat{V})} = {s_n^2+K^2+2K \mathsf{v}}
=(-\mathrm{i} s_n+\sqrt{K^2+2K \mathsf{v}})(\mathrm{i} s_n+\sqrt{K^2+2K \mathsf{v}})
\]
and it is related to the modified energy density of the fluctuations.\\

Let us now understand the corresponding interacting propagators. We know that  
\[
\tilde{\mathsf{G}}_\beta = \hat{T}^{-1}= 
\begin{pmatrix} 
\frac{1}{K-\mathrm{i}s_n}&0\\0& \frac{1}{K+\mathrm{i}s_n}
\end{pmatrix}
\]
and thus correspondingly 
\begin{equation}\label{eq:two-point-function}
\tilde{\mathsf{G}}^V_\beta 
=(\hat{T}+\hat{V})^{-1} = 
\frac{1}{s_n^2 + K^2 + 2K\mathsf{v}}
\begin{pmatrix} {K+\mathrm{i}s_n + \mathsf{v}}&-\mathsf{v}\\ 
-\mathsf{v}& K-\mathrm{i}s_n + \mathsf{v}
\end{pmatrix}
\end{equation}
We furthermore observe that utilizing $K_\epsilon$, with $\epsilon>0$, at the place of $K$ in $T$, $\tilde{\mathsf{G}}_\beta$ and $\tilde{\mathsf{G}}_\beta^V$,
we obtain that $\tilde{\mathsf{G}_\beta}$ is a bounded operator, with operator norm bounded by $\epsilon^{-1}$. Similarly, if $\mathsf{v}$ is also positive,   $\tilde{\mathsf{G}}_\beta^V$ is also bounded and has operator norm bounded by $\epsilon^{-1}$. In order to relate these two propagators we use the Dyson equation
\[
\tilde{\mathsf{G}}^V_\beta=
\tilde{\mathsf{G}}_\beta-
\tilde{\mathsf{G}}_\beta \hat{V}
\tilde{\mathsf{G}}^V_\beta
\]
that recursively gives
\begin{equation}\label{eq:perturbation}
\tilde{\mathsf{G}}^V_\beta
= \sum_{n\geq 0} (-\tilde{\mathsf{G}}_\beta \hat{V})^n\tilde{\mathsf{G}}_\beta.
\end{equation}
The above series, if we use $K_\epsilon$ with $\epsilon>0$ in $\tilde{\mathsf{G}}_\beta$ and if $\|\mathsf{v}\|_1<\epsilon$, is absolutely convergent. On the other hand, if we expand
$\tilde{\mathsf{G}}^V_\beta$ in powers of $\hat{V}$ we obtain back the power series at the right hand side of \eqref{eq:perturbation} without 
any restriction on $\hat{V}$
even if for $\|\mathsf{v}\|_1\geq \epsilon$ that power series is only an asymptotic representation of $\tilde{\mathsf{G}}^V_\beta$.\\

With this discussion in mind, we now consider 
\begin{equation}\label{eq:M1}
\hat{\mathsf{M}}\coloneqq \sum_{ij} {{{\tilde{\mathsf{G}}}^V_\beta}} ._{ij} = \frac{2K}{s_n^2+K^2+2K\mathsf{v}}
= \frac{K}{\sqrt{K^2+2K\mathsf{v}}} \left(\frac{1}{-\mathrm{i}s_n+\sqrt{K^2+2K\mathsf{v}}}+\frac{1}{\mathrm{i}s_n+\sqrt{K^2+2K\mathsf{v}}}\right)
\end{equation}
This quantity is interesting, as contains the modification to the energy density in its denominator. Furthermore, for $u\in(0,\beta)$
\begin{equation}\label{eq:M}
\mathsf{M}(u)= \frac{K}{\sqrt{K^2+2K\mathsf{v}}} 
\frac{e^{-u \sqrt{K^2 +2K\mathsf{v}}}+e^{-(\beta-u) \sqrt{K^2+2K\mathsf{v}} }}{1-e^{-\beta \sqrt{K^2+2K\mathsf{v}} }}
\end{equation}
and the limit $u\to 0^+$ gives
\begin{equation}\label{eq:M}
\mathsf{M}= \frac{K}{\sqrt{K^2+2Kv}} \frac{1+e^{-\beta \sqrt{K^2+2Kv} }}{1-e^{-\beta \sqrt{K^2+2Kv} }}.
\end{equation}
This corresponds to the integral kernel of the operator
\[
\omega^{\beta V}((\Psi(J_1)+\Psi^*({J}_1))
(\Psi(J_2)+\Psi^*(J_2)))
= \langle J_1, \mathsf{M} J_2\rangle
\]
making clear the connection between the above $\mathsf{M}$ and the two point function of the interacting equilibrium state. Similarly, taking the trace of $\tilde{\mathsf{G}}_\beta^V$, we get
\[
\omega^{\beta V}(\Psi(J_1)\Psi^*({J_2})+\Psi^*({J}_1)\Psi( J_2)
)
= \langle J_1, \frac{K+\mathsf{v}}{K}\mathsf{M} J_2\rangle.
\]

In order to connect to the discussion in the paper, we notice that the interaction Hamiltonian which produces the perturbation $V$ here is
\[
\int (\Psi+\Psi^*) \delta(u-u') \mathsf{v}(x,x') (\Psi+\Psi^*).
\]

\section{Scattering length}\label{se:scattering-length}

Consider a potential $V$ and the corresponding solution of the equation
\[
\left(-\Delta +\frac{V}{2}\right) f =0
\]
with the boundary condition $\lim_{|x|\to\infty}f(x)=1$.
The scattering length is then obtained from the asymptotic form for large $|x|$
\[
f(x) =1 - \frac{\mathfrak{a}_0}{|x|}
\]
where $\mathfrak{a}_0$ is the {\bf scattering length}. An alternative definition can be given by capturing the asymptotic decay of $f$. Namely,
\[
8\pi \mathfrak{a}_0=\int V(x) f(x) dx.
\]
With the scattering length at disposal, the spectrum of the effective Hamiltonian describing fluctuations around a condensate was understood to be modified from $p^2$ to $\sqrt{p^4 + 16\pi \mathfrak{a}_0 \phi_R^2 p^2}$. To find $f$, we observe that $\eta(x)=f(x)-1$ vanishes for large $|x|$ and satisfies 
\[
\left(-\Delta + \frac{V}{2}\right)\eta = -\frac{V}{2}. 
\]
Hence, in terms of the Green function of $-\Delta +V/2$
\[
\eta =  -\frac{1}{-\Delta + \frac{V}{2}}\frac{V}{2} 
\qquad  \Longleftrightarrow  \qquad
f = 1 -\frac{1}{-\Delta + \frac{V}{2}}\frac{V}{2} 
\]
and 
\begin{equation}\label{eq:born}
\begin{aligned}
8\pi  \mathfrak{a}_0 
&= \hat{V}(0)  - \frac{1}{2}\left\langle V, \frac{1}{-\Delta + \frac{V}{2}}V\right\rangle .
\end{aligned}
\end{equation}
Finally, from the above, we notice that the scattering length $\mathfrak{a}_0$ depends on the potential $V$ and admits and expansion in powers of $V$
\begin{equation}\label{eq:born-series}
\begin{aligned}
8\pi \mathfrak{a}_0 &= \sum_{n\geq 0 } 8\pi \mathfrak{a}_0^{(n)} 
\\
&=
\hat{V}(0) - \frac{1}{2}\left\langle V , \frac{1}{-\Delta} V\right\rangle 
- \frac{1}{2}\sum_{n\geq 1} (-1)^n\left\langle V , 
\left( \frac{1}{-\Delta} \frac{V}{2}\right)^n\frac{1}{-\Delta} V\right\rangle.
\end{aligned}
\end{equation}
This is known as the convergent Born series of the scattering length and $\mathfrak{a}_0^{(0)} = \hat{V}(0)/(8\pi)$ is the first order Born approximation.

\printbibliography

@article{Araki,
	author = {Araki, Huzihiro},
	doi = {10.2977/PRIMS/1195192744},
	journal = {Publ. Res. Inst. Math. Sci.},
	number = {1},
	pages = {165-209},
	title = {Relative Hamiltonian for Faithful Normal States of a von Neumann Algebra},
	volume = {9},
	year = {1973}}

@article{HW05,
	author = {Hollands, Stefan and Wald, Robert M.},
	doi = {10.1142/S0129055X05002340},
	journal = {Rev. Math. Phys.},
	number = {03},
	pages = {227-311},
	title = {Conservation of the Stress Tensor in Perturbative Interacting Quantum Field Theory in Curved Spacetimes},
	url = {https://doi.org/10.1142/S0129055X05002340},
	volume = {17},
	year = {2005}}

@inbook{Fredenhagen2015,
	author = {Fredenhagen, Klaus and Rejzner, Katarzyna},
	booktitle = {Mathematical Aspects of Quantum Field Theories},
	doi = {10.1007/978-3-319-09949-1_2},
	editor = {Calaque, Damien and Strobl, Thomas},
	isbn = {978-3-319-09949-1},
	pages = {17--55},
	publisher = {Springer International Publishing},
	title = {Perturbative Algebraic Quantum Field Theory},
	url = {https://doi.org/10.1007/978-3-319-09949-1_2},
	year = {2015}}

@article{FredenhagenLindnerKMS_2014,
	author = {K. Fredenhagen and F. Lindner},
	doi = {10.1007/s00220-014-2141-7},
	journal = {Commun. Math. Phys.},
	month = {8},
	number = {3},
	pages = {895--932},
	publisher = {Springer Science and Business Media {LLC}},
	title = {Construction of {KMS} States in Perturbative {QFT} and Renormalized Hamiltonian Dynamics},
	url = {https://doi.org/10.1007/s00220-014-2141-7},
	volume = {332},
	year = 2014}

@article{BrunettiFredenhagen00,
	author = {R. Brunetti and K. Fredenhagen},
	doi = {10.1007/s002200050004},
	journal = {Commun. Math. Phys.},
	month = {1},
	number = {3},
	pages = {623--661},
	publisher = {Springer Science and Business Media {LLC}},
	title = {Microlocal Analysis and Interacting Quantum Field Theories: Renormalization on Physical Backgrounds},
	url = {https://doi.org/10.100/s002200050004},
	volume = {208},
	year = 2000}

@article{DragoHackPinamonti,
	author = {N. Drago and T.-P. Hack and N. Pinamonti},
	doi = {10.1007/s00023-016-0521-6},
	journal = {Ann. Henri Poincare},
	month = {10},
	number = {3},
	pages = {807--868},
	publisher = {Springer Science and Business Media {LLC}},
	title = {The Generalised Principle of Perturbative Agreement and the Thermal Mass},
	url = {https://doi.org/10.1007/s00023-016-0521-6},
	volume = {18},
	year = 2016}

@article{HollandsWald2001,
	author = {S. Hollands and R. M. Wald},
	doi = {10.1007/s002200100540},
	journal = {Commun. Math. Phys.},
	month = {10},
	number = {2},
	pages = {289--326},
	publisher = {Springer Science and Business Media {LLC}},
	title = {Local Wick Polynomials and Time Ordered Products of Quantum Fields in Curved Spacetime},
	url = {https://doi.org/10.1007/s002200100540},
	volume = {223},
	year = 2001}

@book{HaagLQP,
	author = {Haag, R.},
	publisher = {Springer Berlin, Heidelberg},
	title = {{Local quantum physics}},
	year = {1992}}

@book{KasiaBook,
	address = {New York},
	author = {Rejzner, K.},
	doi = {10.1007/978-3-319-25901-7},
	publisher = {Springer},
	series = {Mathematical Physics Studies},
	title = {{Perturbative Algebraic Quantum Field Theory}: {An Introduction for Mathematicians}},
	year = {2016}}

@article{HW02,
	author = {S. Hollands and R. M. Wald},
	doi = {10.1007/s00220-002-0719-y},
	journal = {Commun. Math. Phys.},
	month = {12},
	number = {2},
	pages = {309--345},
	publisher = {Springer Science and Business Media {LLC}},
	title = {Existence of Local Covariant Time Ordered Products of Quantum Fields in Curved Spacetime},
	url = {https://doi.org/10.1007/s00220-002-0719-y},
	volume = {231},
	year = 2002}

@article{ArakiWoods,
	author = {{Araki}, H. and {Woods}, E.~J.},
	doi = {10.1063/1.1704002},
	journal = {J. Math. Phys.},
	number = {5},
	pages = {637-662},
	title = {{Representations of the Canonical Commutation Relations Describing a Nonrelativistic Infinite Free Bose Gas}},
	volume = {4},
	year = 1963}

@article{FroehlichKnowlesSchleinSohinger2,
	author = {J{\"u}rg Fr{\"o}hlich and Antti Knowles and Benjamin Schlein and Vedran Sohinger},
	doi = {10.1007/s00220-017-2994-7},
	eprint = {1605.07095},
	journal = {Commun. Math. Phys.},
	pages = {883-980},
	title = {Gibbs Measures of Nonlinear Schr{\"o}dinger Equations as Limits of Many-Body Quantum States in Dimensions $d \leq 3 $},
	volume = {356},
	year = {2017}}

@article{BoccatoBrenneckeCenatiempoSchlein,
	author = {Boccato, Chiara and Brennecke, Christian and Cenatiempo, Serena and Schlein, Benjamin},
	doi = {10.1007/s00220-017-3016-5},
	issn = {1432-0916},
	journal = {Commun. Math. Phys.},
	number = {3},
	pages = {975--1026},
	publisher = {Springer Science and Business Media LLC},
	title = {Complete Bose--Einstein Condensation in the Gross--Pitaevskii Regime},
	url = {http://dx.doi.org/10.1007/s00220-017-3016-5},
	volume = {359},
	year = {2017}}

@article{KMS,
	author = {Haag, R. and Hugenholtz, N. M. and Winnink, M.},
	doi = {10.1007/BF01646342},
	journal = {Commun. Math. Phys.},
	pages = {215--236},
	title = {{On the Equilibrium states in quantum statistical mechanics}},
	volume = {5},
	year = {1967}}

@book{Strocchi,
	author = {Strocchi, Franco},
	doi = {10.1007/978-3-540-73593-9},
	isbn = {978-3-540-73592-2},
	title = {{Symmetry Breaking}},
	volume = {732},
	year = {2008}}

@article{BDF09,
	author = {R. Brunetti and M. D{\"u}tsch and K. Fredenhagen},
	journal = {Adv. Theor. Math. Phys.},
	number = {5},
	pages = {1541 -- 1599},
	publisher = {International Press of Boston},
	title = {{Perturbative algebraic quantum field theory and the renormalization groups}},
	volume = {13},
	year = {2009}}

@article{LPR05,
	author = {Lewin, Mathieu and Nam, Phan Th{\`a}nh and Rougerie, Nicolas},
	date = {2021/05/01},
	doi = {10.1007/s00222-020-01010-4},
	id = {Lewin2021},
	isbn = {1432-1297},
	journal = {Invent. Math.},
	number = {2},
	pages = {315--444},
	title = {Classical field theory limit of many-body quantum Gibbs states in 2D and 3D},
	url = {https://doi.org/10.1007/s00222-020-01010-4},
	volume = {224},
	year = {2021}}

@article{FKSS3,
	author = {J{\"u}rg Fr{\"o}hlich and Antti Knowles and Benjamin Schlein and Vedran Sohinger},
	doi = {10.1090/jams/987},
	eprint = {2001.01546},
	journal = {J. Amer. Math. Soc.},
	number = {4},
	pages = {955--1030},
	title = {The mean-field limit of quantum Bose gases at positive temperature},
	volume = {35},
	year = {2022}}

@article{HaagKastler,
	author = {Haag, Rudolf and Kastler, Daniel},
	doi = {10.1063/1.1704187},
	eprint = {https://pubs.aip.org/aip/jmp/article-pdf/5/7/848/19154934/848\_1\_online.pdf},
	issn = {0022-2488},
	journal = {J. Math. Phys.},
	month = {07},
	number = {7},
	pages = {848-861},
	title = {An Algebraic Approach to Quantum Field Theory},
	url = {https://doi.org/10.1063/1.1704187},
	volume = {5},
	year = {1964}}

@article{Bose,
	author = {S. N. Bose},
	doi = {10.1007/BF01327326},
	journal = {Z. Phys.},
	pages = {178--181},
	title = {Planck's Law and the Hypothesis of Light Quanta},
	volume = {26},
	year = {1924}}

@article{Einstein,
	author = {A. Einstein},
	journal = {Sitzungsber. Preuss. Akad. Wiss., Phys.-Math. Kl.},
	note = {Part I; followed by Part II (1925, pp. 3--14) and Part III (1925, pp. 18--25)},
	pages = {261--267},
	title = {Quantum Theory of the Monatomic Ideal Gas},
	year = {1924}}

@article{BFP21,
	author = {Romeo Brunetti and Klaus Fredenhagen and Nicola Pinamonti},
	doi = {10.1007/s00023-020-00984-4},
	eprint = {1911.01829},
	journal = {Ann. Henri Poincar{\'e}},
	number = {3},
	pages = {951--1000},
	title = {Algebraic Approach to Bose--Einstein Condensation in Relativistic Quantum Field Theory: Spontaneous Symmetry Breaking and the Goldstone Theorem},
	volume = {22},
	year = {2021}}

@book{Dutsch:2019aa,
	author = {Michael D{\"u}tsch},
	description = {The book develops a novel approach to perturbative quantum field theory: starting with a perturbative formulation of classical field theory, quantization is done by deformation quantization of the underlying free theory and by the principle that we want from the classical structure.},
	doi = {10.1007/978-3-030-04738-2},
	publisher = {Birkh{\"a}user Cham},
	title = {From Classical Field Theory to Perturbative Quantum Field Theory},
    year = {2019}}

@article{Rivasseau,
	author = {Rivasseau, V.},
	doi = {https://doi.org/10.1155/2009/180159},
	eprint = {https://onlinelibrary.wiley.com/doi/pdf/10.1155/2009/180159},
	journal = {Adv. Math. Phys.},
	number = {1},
	pages = {180159},
	title = {Constructive Field Theory in Zero Dimension},
	url = {https://onlinelibrary.wiley.com/doi/abs/10.1155/2009/180159},
	volume = {2009},
	year = {2009}}

@book{RivasseauBook,
	author = {Vincent Rivasseau},
	publisher = {Princeton University Press},
	title = {From Perturbative to Constructive Renormalization},
	year = {1991}}

@article{GalandaPinamonti,
    author = "Galanda, Stefano and Pinamonti, Nicola",
    title = "{Equilibrium states for non relativistic Bose gases with condensation}",
    eprint = "2509.25101",
    archivePrefix = "arXiv",
    primaryClass = "math-ph",
    month = "9",
    year = "2025"
}

@article{LiebSeiringer,
  title = {Proof of Bose-Einstein Condensation for Dilute Trapped Gases},
  author = {Lieb, Elliott H. and Seiringer, Robert},
  journal = {Phys. Rev. Lett.},
  volume = {88},
  issue = {17},
  pages = {170409},
  numpages = {4},
  year = {2002},
  month = {Apr},
  publisher = {American Physical Society},
  doi = {10.1103/PhysRevLett.88.170409},
  url = {https://link.aps.org/doi/10.1103/PhysRevLett.88.170409}
}

@article{BogoliubovBEC,
    author = "Bogoliubov, N. N.",
    title = "{On the theory of superfluidity}",
    journal = "J. Phys. (USSR)",
    volume = "11",
    pages = "23--32",
    year = "1947"
}

@book{Flajolet_Sedgewick_2009, place={Cambridge}, title={Analytic Combinatorics}, publisher={Cambridge University Press}, author={Flajolet, Philippe and Sedgewick, Robert}, year={2009}}

@article{LSY-07,
title = {Bose-Einstein Condensation and Spontaneous Symmetry Breaking1},
journal = {Reports on Mathematical Physics},
volume = {59},
number = {3},
pages = {389-399},
year = {2007},
issn = {0034-4877},
doi = {https://doi.org/10.1016/S0034-4877(07)80074-7},
url = {https://www.sciencedirect.com/science/article/pii/S0034487707800747},
author = {Elliott H. Lieb and Robert Seiringer and Jakob Yngvason}
}

@article{LY-ground,
  title = {Ground State Energy of the Low Density Bose Gas},
  author = {Lieb, Elliott H. and Yngvason, Jakob},
  journal = {Phys. Rev. Lett.},
  volume = {80},
  issue = {12},
  pages = {2504--2507},
  numpages = {0},
  year = {1998},
  month = {Mar},
  publisher = {American Physical Society},
  doi = {10.1103/PhysRevLett.80.2504},
  url = {https://link.aps.org/doi/10.1103/PhysRevLett.80.2504}
}

@article{Lieb-Seiringer-Yngavson,
  title = {Bosons in a trap: A rigorous derivation of the Gross-Pitaevskii energy functional},
  author = {Lieb, Elliott H. and Seiringer, Robert and Yngvason, Jakob},
  journal = {Phys. Rev. A},
  volume = {61},
  issue = {4},
  pages = {043602},
  numpages = {13},
  year = {2000},
  month = {Mar},
  publisher = {American Physical Society},
  doi = {10.1103/PhysRevA.61.043602},
  url = {https://link.aps.org/doi/10.1103/PhysRevA.61.043602}
}

@article{NRS,
title = {Ground states of large bosonic systems: the Gross–Pitaevskii limit revisited},
author = {Nam, Phan Thành and  Rougerie, Nicolas and Seiringer,  Robert},
journal = {Anal. PDE},
volume = {9},
issue = {2},
pages = {459–485},
year = {2016},
doi = {10.2140/apde.2016.9.459},
url = {http://dx.doi.org/10.2140/apde.2016.9.459}
}

@article{HainzlSchleinTriay
, title={Bogoliubov theory in the Gross-Pitaevskii limit: a simplified approach}, volume={10}, DOI={10.1017/fms.2022.78}, journal={Forum Math. Sigma}, author={Hainzl, Christian and Schlein, Benjamin and Triay, Arnaud}, year={2022}, pages={e90}}

@article{DerPet26,
title ={Damping of phonons in Bose gas at low temperatures},
author = {Dereziński, Jan and Pettinari, Lorenzo},
year = {2026},
archiveprefix = {arXiv},
eprint = {2602.07701},
primaryclass = {math-ph},
url = {https://arxiv.org/abs/2602.07701},
doi = {10.48550/arXiv.2602.07701}
}

@incollection{Benfatto,
  author    = {G. Benfatto},
  title     = {Renormalization group approach to zero temperature Bose condensation},
  booktitle = {Constructive Results in Field Theory, Statistical Mechanics and Condensed Matter Physics},
  series    = {Lect. Notes Phys.},
  volume    = {446},
  pages     = {219--247},
  publisher = {Springer},
  address   = {Berlin},
  year      = {1995}
}

@article{SerenaGiuliani,
  author    = {S. Cenatiempo and A. Giuliani},
  title     = {Renormalization theory of a two dimensional Bose gas: quantum critical point and quasi-condensed state},
  journal   = {J. Stat. Phys.},
  volume    = {157},
  number    = {4–5},
  pages     = {755--829},
  year      = {2014},
  doi       = {10.1007/s10955-014-1034-7}
}

@article{Salmhofer,
   title={Functional Integral and Stochastic Representations for Ensembles of Identical Bosons on a Lattice},
   volume={385},
   ISSN={1432-0916},
   url={http://dx.doi.org/10.1007/s00220-021-04010-4},
   DOI={10.1007/s00220-021-04010-4},
   number={2},
   journal={Commun. Math. Phys.},
   publisher={Springer Science and Business Media LLC},
   author={Salmhofer, Manfred},
   year={2021},
   pages={1163–1211} }

@article{DerNap2014,
	author = {Derezi{\'n}ski, Jan and Napi{\'o}rkowski, Marcin},
	date = {2014/12/01},
	doi = {10.1007/s00023-013-0302-4},
	id = {Derezi{\'n}ski2014},
	isbn = {1424-0661},
	journal = {Ann. Henri Poincar{\'e}},
	number = {12},
	pages = {2409--2439},
	title = {Excitation Spectrum of Interacting Bosons in the Mean-Field Infinite-Volume Limit},
	url = {https://doi.org/10.1007/s00023-013-0302-4},
	volume = {15},
	year = {2014}}

@article{BrydgesKennedy,
	author = {Brydges, D. C. and Kennedy, T.},
	date = {1987/07/01},
	doi = {10.1007/BF01010398},
	id = {Brydges1987},
	isbn = {1572-9613},
	journal = {J. Stat. Phys.},
	number = {1},
	pages = {19--49},
	title = {Mayer expansions and the Hamilton-Jacobi equation},
	url = {https://doi.org/10.1007/BF01010398},
	volume = {48},
	year = {1987}}

@InProceedings{AbdesselamRivasseau,
author="Abdesselam, Abdelmalek
and Rivasseau, Vincent",
editor="Rivasseau, Vincent",
title="Trees, forests and jungles: A botanical garden for cluster expansions",
booktitle="Constructive Physics Results in Field Theory, Statistical Mechanics and Condensed Matter Physics",
year="1995",
publisher="Springer Berlin Heidelberg",
address="Berlin, Heidelberg",
pages="7--36",
isbn="978-3-540-49222-1"
}

@article{RivasseauGurau,
	author = {Gurau, Razvan and Rivasseau, Vincent},
	date = {2015/08/01},
	doi = {10.1007/s00023-014-0370-0},
	id = {Gurau2015},
	isbn = {1424-0661},
	journal = {Ann. Henri Poincar{\'e}},
	number = {8},
	pages = {1869--1897},
	title = {The Multiscale Loop Vertex Expansion},
	url = {https://doi.org/10.1007/s00023-014-0370-0},
	volume = {16},
	year = {2015}}

@article{MagnenRivasseauSeneor,
title = {Construction of YM4 with an Infrared Cutoff},
journal = {Commun. Math. Phys.},
volume = {155},
pages = {325-383},
year = {1993},
author= {Magnen, Jacques and  Rivasseau, Vincent and  Seneor, Roland},
doi = {10.1007/BF02097397}
}

@article{Hairer,
	author = {Hairer, M. },
	date = {2014/11/01},
	doi = {10.1007/s00222-014-0505-4},
	id = {Hairer2014},
	isbn = {1432-1297},
	journal = {Invent. Math.},
	number = {2},
	pages = {269--504},
	title = {A theory of regularity structures},
	url = {https://doi.org/10.1007/s00222-014-0505-4},
	volume = {198},
	year = {2014}}

@article{DuchGubinelliRinaldi,
      title={Parabolic stochastic quantisation of the fractional $\Phi^4_3$ model in the full subcritical regime}, 
      author={Paweł Duch and Massimiliano Gubinelli and Paolo Rinaldi},
      year={2025},
      eprint={2303.18112},
      archivePrefix={arXiv},
      primaryClass={math.PR},
      url={https://arxiv.org/abs/2303.18112}, 
}

@article{Gubinelli1,
	author = {Gubinelli, Massimiliano and Hofmanov{\'a}, Martina},
	date = {2019/06/01},
	doi = {10.1007/s00220-019-03398-4},
	id = {Gubinelli2019},
	isbn = {1432-0916},
	journal = {Commun. Math. Phys.},
	number = {3},
	pages = {1201--1266},
	title = {Global Solutions to Elliptic and Parabolic $\Phi^4$ Models in Euclidean Space},
	url = {https://doi.org/10.1007/s00220-019-03398-4},
	volume = {368},
	year = {2019}}

@article{Gubinelli2,
	author = {Gubinelli, Massimiliano and Hofmanov{\'a}, Martina},
	date = {2021/05/01},
	doi = {10.1007/s00220-021-04022-0},
	id = {Gubinelli2021},
	isbn = {1432-0916},
	journal = {Commun. Math. Phys.},
	number = {1},
	pages = {1--75},
	title = {A PDE Construction of the Euclidean $\Phi^4_3$ Quantum Field Theory},
	url = {https://doi.org/10.1007/s00220-021-04022-0},
	volume = {384},
	year = {2021}}

@article{BBCS18,
	author = {Boccato, Chiara and Brennecke, Christian and Cenatiempo, Serena and Schlein, Benjamin},
	date = {2020/06/01},
	doi = {10.1007/s00220-019-03555-9},
	id = {Boccato2020},
	isbn = {1432-0916},
	journal = {Commun. Math. Phys.},
	number = {2},
	pages = {1311--1395},
	title = {Optimal Rate for Bose--Einstein Condensation in the Gross--Pitaevskii Regime},
	url = {https://doi.org/10.1007/s00220-019-03555-9},
	volume = {376},
	year = {2020}
}

@article{BoccatoBastiCenatiempoDeuchert,
      title={A new upper bound on the specific free energy of dilute Bose gases}, 
      author={Giulia Basti and Chiara Boccato and Serena Cenatiempo and Andreas Deuchert},
      year={2025},
      eprint={2507.20877},
      archivePrefix={arXiv},
      primaryClass={math-ph},
      url={https://arxiv.org/abs/2507.20877}, 
}

@article{Beliaev1,
  author  = {S. T. Beliaev},
  title   = {Energy Spectrum of a Non-Ideal Bose Gas},
  journal = {Sov. Phys. J. Exptl. Theoret. Phys.},
  volume  = {7},
  number  = {2},
  pages   = {299--307},
  year    = {1958}
}

@article{Beliaev2,
  author  = {S. T. Beliaev},
  title   = {Application of the Methods of Quantum Field Theory to a System of Bosons},
  journal = {Sov. Phys. J. Exptl. Theoret. Phys.},
  volume  = {7},
  number  = {2},
  pages   = {289--298},
  year    = {1958}
}

@article{NapiorkowskiSolovej,
  author  = {Marcin Napi{\'o}rkowski and Robin Reuvers and Jan Philip Solovej},
  title   = {The Bogoliubov Free Energy Functional II: The Dilute Limit},
  journal = {Commun. Math. Phys.},
  volume  = {360},
  number  = {1},
  pages   = {347--403},
  year    = {2018},
  doi     = {10.1007/s00220-017-3064-x}
}

\end{document}